\documentclass[12pt]{article}
\usepackage{amsmath}
\usepackage{graphicx}
\usepackage{enumerate}
\usepackage{url} % not crucial - just used below for the URL 
\usepackage{amsthm,mathtools,scalerel}
\usepackage{color,soul,latexsym,amsmath,amssymb,amsfonts,amsthm,dsfont,enumitem,xcolor,bbm,threeparttable,graphicx,float,subfigure}
\usepackage{authblk}
\usepackage[round]{natbib}
\newcommand{\blind}{0}
\RequirePackage[colorlinks=true,allcolors=blue,hypertexnames=false]{hyperref}%

\usepackage{setspace}
\usepackage{outlines}
\usepackage{array}
\usepackage{selectp}
\usepackage[ruled,linesnumbered, vlined, noend]{algorithm2e}
\mathtoolsset{showonlyrefs}
\usepackage{titlesec}

\newtheorem{thm}{Theorem}%[section]
\newtheorem{prop}{Proposition}
\newtheorem{lemma}{Lemma}

\DeclareMathOperator*{\argmin}{argmin}

\newcommand{\E}{\mathbb{E}}
\newcommand{\PP}{\mathbb{P}}
\newcommand{\Normal}{{\mathcal{N}}}

\newcommand{\red}[1]{{\color{black}#1}}

\begin{document}

\def\spacingset#1{\renewcommand{\baselinestretch}%
{#1}\small\normalsize} \spacingset{1}

%%%%%%%%%%%%%%%%%%%%%%%%%%%%%%%%%%%%%%%%%%%%%%%%%%%%%%%%%%%%%%%%%%%%%%%%%%%%%%

\if0\blind
{
  \title{\bf Selection and Aggregation of Conformal Prediction Sets}
  \author{Yachong Yang\hspace{.2cm}\\
    Department of Statistics and Data Science, University of Pennsylvania\\
    and \\
    Arun Kumar Kuchibhotla
    % \thanks{
    % The authors gratefully acknowledge \textit{please remember to list all relevant funding sources in the unblinded version}}
    \\
    Department of Statistics and Data Science, Carnegie Mellon University}
  \maketitle
} \fi

 \if1\blind
{
  \bigskip
  \bigskip
  \bigskip
  \begin{center}
    {\LARGE\bf Selection and Aggregation of Conformal Prediction Sets}
\end{center}
  \medskip
}  \fi

\bigskip
\begin{abstract}
%{\color{red}Write in terms of selection and aggregation, then efficiency.}
    Conformal prediction is a generic methodology for finite-sample valid distribution-free prediction. This technique has garnered a lot of attention in the literature partly because it can be applied with any machine learning algorithm that provides point predictions to yield valid prediction regions. Of course, the efficiency (width/volume) of the resulting prediction region depends on the performance of the machine learning algorithm. In the context of point prediction, several techniques (such as cross-validation) exist to select one of many machine learning algorithms for better performance. In contrast, such selection techniques are seldom discussed in the context of set prediction (or prediction regions). In this paper, we consider the problem of obtaining the smallest conformal prediction region given a family of machine learning algorithms. We provide two general-purpose selection algorithms and consider coverage as well as width properties of the final prediction region. The first selection method yields the smallest width prediction region among the family of conformal prediction regions for all sample sizes but only has an approximate coverage guarantee. The second selection method has a finite sample coverage guarantee but only attains close to the smallest width. The approximate optimal width property of the second method is quantified via an oracle inequality. As an illustration, we consider the use of aggregation of non-parametric regression estimators in the split conformal method with the absolute residual conformal score.
    % study the selection problem in detail when the family of algorithms is given by ridge regression with different penalty parameters.
\end{abstract}

\noindent%
\vfill

\newpage
\spacingset{1.5} % DON'T change the spacing!

\section{Introduction}\label{sec:introduction}
Prediction is one of the major focuses of machine learning. Most methods in machine learning are designed to provide point predictions, but quantifying the uncertainty in the point prediction is crucial for several applications. In criminal justice applications, point predictions from classification and regression algorithms are commonly used in problems of predicting whether a convicted individual will re-offend if given parole and in problems of predicting the length of a prison sentence. In these examples, it is very important to supplement point predictions with a valid uncertainty estimate~\citep{berk2023}. Such an uncertainty estimate can be obtained by using a prediction set/interval for the response. For applicability in a wide range of applications, it would be beneficial to perform uncertainty quantification under as minimal assumptions on the data generating process as possible. Conformal prediction methodology introduced by~\cite{vovk2005algorithmic} provides a generic algorithm for solving this problem with the only assumption of exchangeable data. While exchangeability is the minimum requirement, we will, in this paper, assume independent and identically distributed observations for reasons to be discussed.

Formally, suppose we have independent and identically distributed (i.i.d.) training data $Z_1, \ldots, Z_N$ from a distribution $P$  on some space $\mathcal{Z}$. The space $\mathcal{Z}$ can be arbitrary such as a function space, a space of images, etc. The goal of conformal prediction is to predict $Z_{N+1}$ using a prediction set $\widehat{C}_{N,\alpha}=\widehat{C}_{N,\alpha}(Z_1,\dots,Z_{N})$ such that
% \begin{equation}
$\mathbb{P}(Z_{N+1} \in \widehat{C}_{N,\alpha}) \geq 1-\alpha,$
% \end{equation}
for a fixed $0 < \alpha <1$, and the probability is taken with respect to the new test point $Z_{N+1}$ along with the training data $Z_1, \ldots, Z_N$. Given any point prediction algorithm, conformal prediction works as a wrapper to obtain a valid prediction set, irrespective of what the point prediction algorithm is. Although the validity guarantee holds very generally, efficiency (in terms of the width/volume) of the prediction region depends crucially on the performance of the algorithm as well as the data generating process. In this paper, we develop methods to attain coverage validity and width efficiency using conformal prediction regions based on several training methods and selecting one of them. 

We now provide a brief introduction to the literature on conformal prediction. In the recent years, several authors have discussed different ways of applying the conformal prediction in order to attain better statistical and computational efficiency. \cite{vovk2005algorithmic} introduced a version of conformal prediction called the transductive conformal method that requires fitting the learning algorithm to the data $Z_i, 1\le i\le N$ and $Z_{N+1} = z$ for all $z\in\mathcal{Z}$. This method is discussed in~\cite{lei2013distribution} as the full conformal method and is computationally intensive in practice, but makes full use of the data for prediction. \cite{papadopoulos2002inductive} proposed an alternative method called the inductive conformal method (or the split conformal method in~\cite{lei2014distribution}) that splits the data into two parts. The learning algorithm is trained only on the first part and the prediction set is constructed using conformal ``scores'' on the second part of the data. For concreteness, we describe the basic algorithm here in regression setting with $Z_i = (X_i, Y_i)\in\mathcal{X}\times\mathbb{R}$. Let $\mathcal{A}:\mathcal{X}\to\mathbb{R}$ be any predictive algorithm trained on the first split of the data, i.e., for any $x\in\mathcal{X}$, $\mathcal{A}(x)$ is the point prediction for $Y$. If $m$ is the number of observations in the second split, $R_i = |Y_i - \mathcal{A}(X_i)|$ are residuals for the second split of the data, and $\widehat{Q}_{\alpha}$ is the $\lceil(m+1)(1-\alpha)\rceil$-th largest element of $R_i$'s, then for any $(X_{N+1}, Y_{N+1})\sim P = P_X\otimes P_{Y|X}$, $\mathbb{P}(|Y_{N+1} - \mathcal{A}(X_{N+1})| \le \widehat{Q}_{\alpha}) \ge 1 - \alpha.$ In particular, for $\widehat{C}_{N,\alpha} = \{(x, y)\in\mathcal{X}\times\mathbb{R}:\,|y - \mathcal{A}(x)| \le \widehat{Q}_{\alpha}\}$, we have $\mathbb{P}((X_{N+1}, Y_{N+1}) \in \widehat{C}_{N,\alpha}) \ge 1 - \alpha$. {In general, marginal coverage is not sufficient for practical purposes and require conditional coverage conditional on the features. This is, however, impossible to attain the finite samples without strong assumptions as shown in~\cite{barber2019limits}.}

The split conformal method only uses part of the data for prediction and several methods have since been developed to make better use of the data. These methods involve aggregation methods~\citep{johansson2014regression,carlsson2014aggregated,vovk2015cross,linusson2017calibration}, leave-one-out and K-fold methods~\citep{barber2019predictive,kim2020predictive}. These methods generally train the learning algorithm on multiple folds of the data and combine the predictions to create a valid prediction set. All these aggregation methods provide similar finite sample distribution-free guarantees to that of the full and split conformal methods, irrespective of the learning algorithm; see~\cite{gupta2019nested} for details. It should be clarified that the word ``aggregation'' in these works is used in the sense of combining prediction algorithms obtained from multiple splits of the data and usually the same prediction algorithm is applied on each split of the data, which need not yield width efficiency.  

These algorithms are designed with the intention of making full use of the data, unlike the split conformal method. They, however, may not necessarily lead to the smallest prediction region even in the asymptotic sense (i.e., they are not constructed to achieve a smaller width). Construction of optimal prediction regions was considered in~\cite{lei2013distribution} for the case of no covariates, in~\cite{lei2018distribution} for the case of symmetric and homoscedastic regression problem, in~\cite{sadinle2019least} and ~\cite{romano2020classification} for the classification problem, and in~\cite{gyofi2020nearest} in a general setting. Also, see~\cite{vovk2014efficiency}, and~\cite{vovk2017criteria}. With no symmetric or homoscedasticity assumptions, the problem of the smallest prediction interval for regression was considered in~\cite{gupta2019nested},~\cite{sesia2020comparison},~\cite{izbicki2020cd}, and~\cite{kivaranovic2020adaptive}. The construction of prediction regions to attain the smallest width/volume essentially requires estimating some non-parametric quantity such as the conditional mean, conditional quantile, or the conditional density function consistently; see~\citet[Assumption 3]{sesia2020comparison} and~\citet[Assumption 9]{izbicki2020cd}. Estimating these nonparametric quantities at the optimal rate would imply that the corresponding conformal prediction region would converge to the smallest valid prediction region as shown in~\citet[Section 3]{lei2018distribution}. Somewhat counterintuitively, ``optimal'' prediction regions converging to the smallest valid prediction region need not make full use of the data and can be obtained using the split conformal method. This follows partly because the optimal prediction set does not shrink to a point as the sample size increases.

Note that almost all estimators of nonparametric quantities require user specification of tuning parameters such as bandwidth, number of nearest neighbors, width and number of hidden layers in neural networks, number of trees in {boosting}, and so on. In traditional nonparametric estimation, these tuning parameters can be chosen via the cross-validation technique. But this choice focuses on obtaining an estimate with the smallest $L_2$ or prediction risk and need not yield the smallest prediction region. In this paper, we consider the problem of selecting tuning parameters in order to obtain the smallest valid prediction region. This problem was discussed in~\citet[Section 4]{lei2013distribution} in the case where there are no covariates. They present an algorithm to solve this problem but do not prove that the algorithm results in the smallest prediction region. This is one of the inspirations for the current paper. 

More formally, suppose $\mathcal{A}_1, \ldots, \mathcal{A}_K$ are $K$ different training algorithms. The question we discuss is how to select $\widehat{k}\in\{1,2,\ldots,K\}$ and construct a prediction region with valid coverage and the smallest width or close to the smallest width. This closeness will be proved in terms of an oracle inequality; see Section~\ref{sec:VFCP} for details. Why does this lead to width efficiency, in general? This is best understood via the regression example used in the description of the split conformal procedure. If the data generating process for $(X, Y)$ satisfies $Y = m_0(X) + \xi$ with $\xi\perp X$ and is symmetrically distributed, then~\cite{lei2018distribution} proves that the conformal prediction interval is asymptotically optimal in width if $\mathcal{A}(\cdot)$ is consistent for $m_0(\cdot)$ in an appropriate sense. Suppose now one uses different kernel regression estimators $\mathcal{A}_1, \ldots, \mathcal{A}_K$ differing in the choice of their bandwidths. If $m_0(\cdot)$ is smooth, then there exists a choice of bandwidth in the collection that yields a minimax rate-optimal estimator for $m_0(\cdot)$. Hence, selecting a bandwidth from a finite collection for a minimum width prediction set yields the width efficient conformal prediction set. A similar argument can be made in the context of conditional quantiles~\citep{sesia2020comparison}.

We will discuss two selection methods and two corresponding prediction regions; both are designed to obtain the smallest width among the family of regions considered. In each case, we consider finite sample coverage as well as finite sample width properties. 
We define different versions of oracle prediction sets and derive oracle inequalities with respect to the width/volume; these oracles are similar to those considered in~\cite{lei2018distribution}. We will provide the exact oracle inequality when the family of training algorithms is based on ridge regression with different penalty levels. None of our results assume any specific structure on the true data generating mechanism but only moment conditions. Furthermore, we need the observations to be independent and identically distributed for these oracle inequalities. In proving the validity guarantee for the first selection method, we find a surprising phenomenon about the split conformal method that allows repeated splitting of the data, and reporting (any) one of the resulting split conformal prediction sets would still yield a prediction set that is valid in terms of coverage probability. This is not an obvious fact. For instance, in Section 2.3 of~\cite{lei2018distribution}, the authors use a union bound for the validity of such a reporting from a collection of split conformal prediction intervals. Our results essentially imply that no such correction based on union bound is needed for approximate validity; see Theorem~\ref{thm:coverage-guarantee-minimum-width} in Section~\ref{sec:EFCP}. 

Finally, we note that a recent paper~\cite{chen2021learning} considers the problem of constructing the smallest valid prediction regions out of all prediction intervals with lower and upper end points of the interval belonging to VC classes. The authors make use of concentration inequalities for VC classes and directly target attaining the smallest prediction set without specifically using a family of training algorithms. This does not use conformal prediction. Our approach in this paper is quite different.

\paragraph{Organization.} The remaining article is organized as follows. In Section~\ref{subsec:problem-formulation}, we formulate the problem along with the notation used throughout the article. In Section~\ref{sec:EFCP}, we provide our first algorithm called ``efficiency first conformal prediction'' (EFCP) that attains a finite sample smallest width but only has approximate coverage validity. {We also discuss the implications of selection on conditional coverage.} In Section~\ref{sec:VFCP}, we provide our second algorithm called ``validity first conformal prediction'' (VFCP) that has finite sample coverage validity and only attains approximately the smallest width. In Sections~\ref{sec:EFCP} and~\ref{sec:VFCP}, we study the coverage and width properties of both algorithms and quantify how close these methods come to attaining validity and efficiency (i.e., the smallest width). The algorithms presented in these sections are discussed in the setting of a general prediction problem, not just in the regression setting. In Section~\ref{sec:comparison-EFCP-VFCP}, we compare EFCP and VFCP algorithms theoretically and contrast their performance. 
In Section~\ref{sec:linear-prediction-fixed-width}, we focus on the smallest width prediction set in the regression setting. In Section~\ref{sec:general-algorithm}, we take the training algorithms to be fixed linear functions of covariates and the prediction sets are \textit{``fixed width"} in that they are of the form $\{(x, y):\,y\in[\theta^{\top}x - T, \theta^{\top}x + T]\}$ for a value $T$ independent of $x$. 
% In Section~\ref{sec:algorithm}, we take the training algorithms to be ridge regressions with different penalty parameters and the prediction sets are again ``fixed width". In these sections, we present the assumptions and prove the corresponding oracle inequality.
{In Section~\ref{subsec:aggregation-conditional-coverage}, we consider the application of the results in Section~\ref{sec:general-algorithm} to aggregation of non-parametric regression estimators and also study conditional coverage of the resulting prediction set.}
In Section~\ref{sec:simulations}, we present a few simulations to validate the theoretical results.
% , and we compare our choice of tuning parameter to the one based on cross-validation. 
In Section~\ref{sec:conclusions}, we summarize the article and discuss some interesting future directions.
The proof of all the results and some supporting lemmas are provided in the supplementary file. For convenience, the sections and equations of the supplementary file are prefixed with ``S.'' and ``E.'', respectively. {\color{red}R code to reproduce the experiments in this work is provided at \href{https://github.com/Elsa-Yang98/Selection_and_aggregation_of_conformal_prediction_sets}{https://github.com/Elsa-Yang98/Selection\_and\_aggregation\_of\_conformal\_prediction\_sets}.}

\subsection{Problem Formulation}\label{subsec:problem-formulation}
% In spite of these nice features of conformal prediction, questions remain in light of over-parametrization and hyper-parameter tuning. 
%We note that \cite{burnaev2014efficiency} also investigated the performance of conformalized ridge regression and without mentioning how to tune the hyper-parameter for the ridge penalty, they showed that under Gaussian assumptions, the conformal prediction sets differ little from the Bayes prediction interval. These questions go beyond linear regression. 
In this section, we describe the problem of constructing valid smallest prediction region in terms of width or volume and outline the content of the following sections. 
% For example, conformal methodology can be applied to neural networks with arbitrary architectures but construct architectures with smallest prediction sets that also have validity guarantees is still unknown. 
Based on the conformal prediction methodology, one simple method would be to construct several algorithms $\mathcal{A}_1, \ldots, \mathcal{A}_K$ on the training data, corresponding prediction sets based on the test data, and select the one that yields the smallest prediction set, in terms of width. Although this leads to the smallest prediction set, its validity does not follow readily. The reason is that this is a form of cherry-picking.
Thus, we formulate the problem we tackle in this paper as
\begin{quote}
Given a collection of training algorithms, construct a selection rule that yields a near optimal prediction set as well as validity. Further, quantify how close one can get to the smallest prediction set in terms of width and to the required coverage.
\end{quote}
Inspired by Algorithm 1 of~\cite{lei2013distribution}, we will present two algorithms for this problem based on the split conformal method. \cite{lei2013distribution} discuss the construction of the smallest prediction set using level sets of kernel density estimators. They tune the bandwidth of the kernel density estimator by minimizing the volume of the prediction set.
% This problem was mentioned in~\cite{lei2013distribution}, where  
% Cross-validation (CV) might be the first choice for choosing an algorithm out of a given set of them, but CV targets the prediction risk instead of the width of a prediction set. 
Our algorithms are a generalization of the one from~\cite{lei2013distribution}.
The first of our selection algorithms targets finite sample efficiency and the second targets finite sample validity. The first algorithm will be called ``efficiency first conformal prediction'' (EFCP). It splits the data into two parts, constructs $K$ conformal prediction sets, and simply returns the one with smallest width. The second algorithm will be called ``validity first conformal prediction'' (VFCP). It splits the data into three parts, constructs $K$ conformal prediction sets, selects one of them using the second split, and then recomputes the prediction set using the third split of the data. The EFCP and VFCP methods as pseudocode are presented in Algorithms~\ref{alg:efficiency-first-conformal} and~\ref{alg:validity-first-conformal}. In both these algorithms, the construction of conformal prediction sets is described in terms of nested sets following~\cite{gupta2019nested}. As mentioned in~\cite{gupta2019nested}, the nested set formulation of conformal methodology is equivalent to the usual formulation in terms of conformal scores. This is useful because the nested set formulation provides a geometric intuition of the conformal framework.

\paragraph{Notation.}
% The following notation will be useful. 
For any $m\ge1$, we use $[m]$ to indicate the set $\{1,2, \ldots, m\}$. For any set $J \subseteq [d]$, $|J|$ denotes the cardinality of $J$. For any vector or matrix $v$, let $v^\top$ denote its transpose. The Euclidean norm in any dimension is given by $\|\cdot\|_2$. 
% For any matrix $A$, let $\| A \|_{\mathrm{op}}$ represent the operator norm of $A$, that is, $\|A\|_{\mathrm{op}}=\sup _{\| \theta \|_2=1}\|A \theta\|_2$. For any vector $x \in \mathbb{R}^{d}$, $\|x\|_{A}:=\sqrt{x^{\top} A x}$. Given ${x}=\left(x_{1}, \ldots, x_{d}\right)^{\top} \in \mathbb{R}^{d},$ we write $\mathrm{diag}({x}) \in \mathbb{R}^{d \times d}$ for the diagonal matrix whose $j$th diagonal entry is $x_{j}$.
% % \todo{Do we use this diag notation??} 
% We use $S^{d-1}$ to denote the sphere of unit radius in $\mathbb{R}^d$, $\mathbb{O}^{d_1 \times d_2}$
% % \todo{Do we use this?} 
% to denote the set of matrices in $\mathbb{R}^{d_1 \times d_2}$ with orthonormal columns, $\mathbb{S}^{d \times d}$ and  $\mathbb{S}_{+}^{d \times d}$ to denote the set of symmetric matrices and semi-positive matrices in $\mathbb{R}^{d \times d}$ respectively. If ${A} \in \mathbb{S}^{d \times d}$ has the eigendecomposition ${A}={Q} \mathrm{diag}\left(\mu_{1}, \ldots, \mu_{d}\right) {Q}^{\top}$ for some ${Q} \in \mathbb{O}^{d \times d}$ and $\mu_{1} \geq \cdots \geq \mu_{d},$ we write $\lambda_{\min}({A}):=\mu_{d}$ for its smallest eigenvalue. In addition if $A \in \mathbb{S}_+^{d \times d}$, then $A^{1 / 2}:= {Q} \mathrm{diag}\left( \sqrt{\mu_{1}}, \ldots, \sqrt{\mu_{d}}\right) {Q}^{\top} $ represent the matrix square root. The identity matrix in dimension $d$ is denoted by $I_{d}$.  When we write $A^{-1}$, it is implicitly assumed that $A$ is invertible with its inverse $A^{-1}$. Denote $A^+$ as the Moore--Penrose pseudo-inverse of matrix $A$. 
We use $\mathfrak{C}$ to denote a universal constant.% throughout the paper. 
% \todo[inline]{Looks like there are many notations that we do not use in this paper. Please remove things we don't use. Also, we are currently writing $I$ for identity rather than $I_d$. I think we should use $I_d$ throughout the paper for identity matrix.}

Throughout the paper, we require the observations to be exchangeable when using conformal prediction, and hence the observations are identically distributed.
%%%%%%%%%%%%%%%%%%%%%%%%%%%%%%%%%%%%%%%%%%%%%%%%%%%%%%%%%%%%%%%
%%%%%%%%%%%%%%%%%%%%%%%%%%%%%%%%%%%%%%%%%%%%%%%%%%%%%%%%%%%%%%%
\section{Efficiency First Conformal Prediction}\label{sec:EFCP}
In Section~\ref{sec:introduction}, we have introduced the split conformal procedure with an absolute residual score. In the current and the following sections, we will use the nested conformal framework of~\cite{gupta2019nested}. For the convenience of the readers, we rewrite the ``absolute residual"-based split conformal in terms of nested sets. Split the observed data $(X_i, Y_i)\in\mathcal{X}\times\mathbb{R}, 1\le i\le N$ into two parts $\mathcal{D}_1$ and $\mathcal{D}_2$, where a point prediction algorithm $\mathcal{A}$ is trained on $\mathcal{D}_1$. Now construct the nested sets
\begin{equation}\label{eq:nested-set-absolute-residual}
\mathcal{F}_t ~:=~ \{(x, y)\in\mathcal{X}\times\mathbb{R}:\, y\in [\mathcal{A}(x) - t, \mathcal{A}(x) + t]\},\quad t\ge0.
\end{equation}
Note that these sets are random through $\mathcal{A}$ (and $\mathcal{D}_1$). These are (increasing) nested sets because for $0 \le t \le t' < \infty$, $\mathcal{F}_t \subseteq \mathcal{F}_{t'}$. For $(X_i, Y_i)\in\mathcal{D}_2$ (the second split of the data), define the score
$t(X_i, Y_i) = \inf\{t\ge0:\, (X_i, Y_i)\in\mathcal{F}_t\}.$
From the definition of $\mathcal{F}_t$, it is easy to show that $t(X_i, Y_i) = |Y_i - \mathcal{A}(X_i)|$. Set $T_{\alpha}$ to be the $(1-\alpha)(1 + 1/|\mathcal{D}_2|)$-th quantile of $t(X_i, Y_i)$ for observations in $D_2$. Then the split conformal prediction set is $\widehat{C}_{N,\alpha} = \mathcal{F}_{T_{\alpha}}$. Proposition 1 of~\cite{gupta2019nested} proves that $\widehat{C}_{N,\alpha}$ constructed as above has a finite sample coverage of at least $1 - \alpha$ for a future test point $(X_{N+1}, Y_{N+1})$ from the same distribution as that of training data. This validity guarantee does not require the special structure of $\mathcal{F}_t$ in~\eqref{eq:nested-set-absolute-residual} and holds for any sequence of (increasing) nested sets.

With this background, we are now ready to describe our first algorithm for efficient conformal prediction. The core idea is to train $K$ algorithms on $\mathcal{D}_1$, construct $K$ different collections of increasing nested sets as well as the corresponding split conformal prediction set using $\mathcal{D}_2$, and then return the prediction set that has the smallest width. Algorithm~\ref{alg:efficiency-first-conformal} presents the pseudocode for our efficiency first conformal prediction method, where for generality, it is stated without assuming a regression setting. The observations $Z_1, \ldots, Z_N$ are assumed to lie in $\mathcal{Z}$, a general measurable space. 

\begin{algorithm}[htbp]
    \SetAlgoLined
    \SetEndCharOfAlgoLine{}
    \KwIn{Data $\mathcal{D} = \{Z_i, i\in [N]\}$, coverage probability $1 - \alpha$, and training methods $\mathcal{A}_k, k\in [K]$.}
    \KwOut{A valid prediction set $\widehat{C}^{\mathrm{EFCP}}_{N,\alpha}$ with smallest width.}
    Randomly split $[N] = \{1, 2, \ldots, N\}$ into two (disjoint) parts $\mathcal{I}_1, \mathcal{I}_2$. Set $\mathcal{D}_1 = \{Z_i: i\in\mathcal{I}_1\}$, $\mathcal{D}_2 = \{Z_i: i\in \mathcal{I}_2\}$.\;
    Fit training methods $\mathcal{A}_1, \ldots, \mathcal{A}_K$ on $\mathcal{D}_1$ and using fitted method $\mathcal{A}_k$, construct an increasing (nested) sequence of sets $\{\mathcal{F}_t^{(k)}\}_{t\in\mathcal{T}}$. Here $\mathcal{T}$ is a subset of $\mathbb{R}$.\;
    % \hspace{0.1in} 
    For each $i\in\mathcal{I}_2$ (i.e., $Z_i\in\mathcal{D}_2$) and $k\in[K]$, define the conformal score
    \[
    t_k(Z_i) := \inf\{t\in\mathcal{T}:\,Z_i\in\mathcal{F}_t^{(k)}\}.
    \]
    Compute the corresponding conformal prediction set as
    % \[
    $\widehat{C}_k ~:=~ \left\{z:\, t_k(z) \le T_{\alpha,k}\right\},$
    % \]
    where $T_{\alpha,k}$ is the $\lceil(1 + |\mathcal{I}_2|)(1-\alpha)\rceil\mbox{-th largest element of }t_k(Z_i), i\in\mathcal{I}_2.$\;
    % \hspace{0.1in} 
    Set $$\widehat{k} := \argmin_{1\le k\le K}\,\mbox{Width}(\widehat{C}_k),$$
    Here $\mbox{Width}(\cdot)$ can be any measure of width or volume of a prediction set. The quantity $\widehat{k}$ need not be unique and any minimizer can be chosen.\;
    % Compute the conformal prediction set $\widehat{C}_{\alpha}^{\mathrm{final}}$ based on $\mathcal{D}_3$ for algorithm $\mathcal{A}_{\widehat{k}}$ and $\{\mathcal{F}_t^{(\widehat{k})}\}_{t\in\mathcal{T}}$. This repeats the process in step 3 with $\mathcal{D}_2$ replaced by $\mathcal{D}_3$ and $k$ replaced by $\widehat{k}$.\;
    % Report the conformal prediction set $\widehat{C}^{\mathrm{VFCP}}_{N,\alpha}$ obtained from $\mathcal{D}_3$ and $\mathcal{A}_{\widehat{k}}$.\;
    \Return the prediction set $\widehat{C}_{\widehat{k}}$ as $\widehat{C}^{\mathrm{EFCP}}_{N,\alpha}$.
    \caption{Efficiency First Conformal Prediction for the Smallest Prediction Set}
    \label{alg:efficiency-first-conformal}
\end{algorithm}
% Algorithm~\ref{alg:validity-first-conformal} provides a viable way to pick a training method to yield a smaller prediction set. It is written in terms of the nested conformal framework of~\cite{gupta2019nested}.

The nested sets $\mathcal{F}_t^{(k)}, t\in\mathcal{T}$ in step 2 are subsets of the domain $\mathcal{Z}$ of $Z_i$'s. Further, we do not place any restriction on $\mathcal{Z}$; it does not have to be a Euclidean or metric space.
The construction of nested sets in steps 2 \& 3 of algorithm \ref{alg:efficiency-first-conformal} can be found in~\cite{gupta2019nested}. Here we will provide three simple examples. 
\begin{itemize}
    \item \textbf{Density Level Sets.} In the problem of predicting $Z_{N+1}$ (without covariates), training methods $\mathcal{A}_1, \ldots, \mathcal{A}_K$ can represent kernel density estimators $\widehat{p}_{h_1}, \ldots, \widehat{p}_{h_K}$ with different bandwidths $h_1, \ldots, h_K$ computed on data $\mathcal{D}_1$. For each $1\le k\le K$ and $t \ge 0$, consider the nested sets
    $\mathcal{F}_t^{(k)} := \left\{z:\, 1/\widehat{p}_{h_k}(z) \le t\right\}.$ Clearly, these are increasing sets over $t\in \mathcal{T} = [0, \infty)$. In this case, $\mbox{Width}(\cdot)$ can be taken to be the Lebesgue measure of the prediction set. 
    % This is the setting of~\cite{lei2013distribution}.
    %%%%%%%%%%%%%%%%%%%%%%%%%%%%%%%%%%%%%%%%%%%%%%%%%%%%%%
    %%%%%%%%%%%%%%%%%%%%%%%%%%%%%%%%%%%%%%%%%%%%%%%%%%%%%%
     \item \textbf{Fixed Width Regression.} Suppose $Z_i = (X_i, Y_i), 1\le i\le N$. In the problem of predicting $Y_{N+1}$ when given $X_{N+1}$, training methods $\mathcal{A}_1, \ldots, \mathcal{A}_K$ can represent neural network regression estimators $\widehat{m}_1(\cdot), \ldots, \widehat{m}_K(\cdot)$ with different network widths $L_1, \ldots, L_K$ computed on data $\mathcal{D}_1$. An example of nested sets is $\mathcal{F}_t^{(k)} = \{(x, y):\, y\in[\widehat{m}_k(x) - t, \widehat{m}_k(x) + t]\}$, for $t\ge 0$. Clearly, these are increasing sets over $t\in\mathcal{T} = [0, \infty)$. In this case, $\mbox{Width}(\cdot)$ cannot be taken to be the Lebesgue measure on the support of $(X, Y)$ for the reason that it would be infinite if the support of $X$ is unbounded. An alternative is to consider average Lebesgue measure of the $x$-cross-section of the prediction set. With nested sets as described, the prediction sets $\widehat{C}_k$ take the form of $\{(x, y):\, y\in [\widehat{m}_k(x) - T_{\alpha,k}, \widehat{m}_k(x) + T_{\alpha,k}]\}$. The $x$-cross-section of the prediction set is $[\widehat{m}_k(x) - T_{\alpha,k}, \widehat{m}_k(x) + T_{\alpha,k}]$ and has a Lebesgue measure of $2T_{\alpha,k}$. Note that, in this setting, width of the prediction set is a random quantity and the EFCP prediction set has finite sample efficiency in the sense of having the smallest random width.
    %%%%%%%%%%%%%%%%%%%%%%%%%%%%%%%%%%%%%%%%%%%%%%%%%%%%%%%
    %%%%%%%%%%%%%%%%%%%%%%%%%%%%%%%%%%%%%%%%%%%%%%%%%%%%%%%
    \item \textbf{Conformalized Quantile Regression (CQR).} Suppose $Z_i = (X_i, Y_i)\in\mathcal{X}\times\mathbb{R}, 1\le i\le N$. In the problem of predicting $Y_{N+1}$ when given $X_{N+1}$, training methods $\mathcal{A}_1, \ldots, \mathcal{A}_K$ can represent quantile regression forest estimators $\{\widehat{q}_{\alpha/2}^{(k)}(\cdot), \widehat{q}_{1-\alpha/2}^{(k)}(\cdot)\}, k\in[K]$ from~\cite{meinshausen2006quantile} with different \texttt{mtry} (number of variables selected at each split) and \texttt{ntree} (number of trees) parameters. The function $\widehat{q}_{\alpha/2}^{(k)}(x)$ represents an estimator of the conditional $\alpha/2$-th quantile of $Y|X$. All these estimators are computed on data $\mathcal{D}_1$ and  $1-\alpha$ is the required coverage probability. An example of nested sets is 
    % \[
    $\mathcal{F}_t^{(k)} = \{z = (x, y):\, y\in[\widehat{q}_{\alpha/2}^{(k)}(x) - t, \widehat{q}_{1-\alpha/2}^{(k)}(x) + t]\},$
    % \]
    for $t\in\mathbb{R}$. Clearly, these are increasing sets over $t\in\mathcal{T} = (-\infty, \infty)$. 
    These nested sets lead to the conformalized quantile regression prediction from~\cite{romano2019conformalized}.
    In this case,  again $\mbox{Width}(\cdot)$ cannot be taken to be the Lebesgue measure on the support of $(X, Y)$ for the same reason that it would be infinite if the support of $X$ is unbounded. With nested sets as described, the prediction sets take the form of $\widehat{C}_k := \{(x, y):\,y\in[\widehat{q}_{\alpha/2}^{(k)}(x) - T_{\alpha,k}, \widehat{q}_{1-\alpha/2}^{(k)}(x) + T_{\alpha,k}]\}$. The $x$-cross-section of the prediction set is $[\widehat{q}_{\alpha/2}^{(k)}(x) - T_{\alpha,k}, \widehat{q}_{1-\alpha/2}^{(k)}(x) + T_{\alpha,k}]$ and has the average Lebesgue measure of $2T_{\alpha,k} + \int (\widehat{q}_{1-\alpha/2}^{(k)}(x) - \widehat{q}_{\alpha/2}^{(k)}(x))P_X(dx)$, where $P_X$ is the probability measure of $X_i$. With $P_X$ unknown, we may take
    \begin{equation}\label{eq:width-x-cross-section}
        \mbox{Width}(\widehat{C}_k) := 2T_{\alpha,k} + \frac{1}{|\mathcal{I}_2|}\sum_{i\in\mathcal{I}_2}(\widehat{q}_{1-\alpha/2}^{(k)}(X_i) - \widehat{q}_{\alpha/2}^{(k)}(X_i)). 
    \end{equation}
    Similar nested sets based on conditional quantile regression estimators can be used to recover other conformal prediction methods; see Table 1 of~\cite{gupta2019nested}. 
\end{itemize}
It is important to note that the nested sequence of sets $\mathcal{F}_t^{(k)}$ are random via $Z_i, i\in\mathcal{D}_1$.
Regarding the data splits $\mathcal{D}_j, j = 1, 2$, there is no restriction on the number of observations within each split. In some cases, one might even have $|\mathcal{I}_1| = 0$, i.e., the nested sequence of sets $\{\mathcal{F}_t^{(k)}\}_{t\in\mathcal{T}}$ are constructed without using any of the observations. For example, in the regression setting, one may take $\mathcal{F}_t^{(k)}$ to be $\{(x, y):\, y\in [\theta_k^{\top}x - t, \theta_k^{\top}x + t]\}$ for some fixed $\mathbb{R}^d$ vectors $\theta_k, k\in[K]$; see Section~\ref{sec:general-algorithm} for a discussion.

Getting back to the properties of $\widehat{C}_{N,\alpha}^{\mathrm{EFCP}}$ in Algorithm~\ref{alg:efficiency-first-conformal}, 
from the definition of $\widehat{k}$ in step 4 of Algorithm~\ref{alg:efficiency-first-conformal}, it follows that
\begin{equation}\label{eq:minimum-width-empirical}
\mbox{Width}(\widehat{C}_{N,\alpha}^{\mathrm{EFCP}}) ~=~ \min_{1\le k\le K}\mbox{Width}(\widehat{C}_k).
\end{equation}
This implies that the prediction set $\widehat{C}^{\mathrm{EFCP}}_{\alpha}$ targets finite sample efficiency in terms of width and hence is an ``efficiency first'' prediction procedure.
Regarding the coverage, each of the prediction sets $\widehat{C}_k, k\in[K]$ in step 3 satisfies finite sample prediction coverage validity. Proposition 1 of~\cite{gupta2019nested} proves that
% \begin{equation}\label{eq:conformal-guarantee-unselected}
$\mathbb{P}(Z_{N+1}\in\widehat{C}_k) \ge 1 - \alpha$, for all $1\le k\le K,$
% \end{equation}
but $\widehat{C}_{N,\alpha}^{\mathrm{EFCP}} = \widehat{C}_{\widehat{k}}$ may not have a finite sample validity.
Because $\widehat{C}_{N,\alpha}^{\mathrm{EFCP}}$ is one of $\widehat{C}_k, k\in[K]$, it readily follows from the union bound that
% and~\eqref{eq:conformal-guarantee-unselected} that
\begin{equation}\label{eq:union-bound-coverage}
\mathbb{P}\left(Z_{N+1} \in \widehat{C}_{N,\alpha}^{\mathrm{EFCP}}\right) \ge 1 - \sum_{k=1}^K \mathbb{P}\left(Z_{N+1}\notin \widehat{C}_{k}\right) \ge 1 - K\alpha.
\end{equation}
This coverage guarantee is not very useful for large $K$ (and completely useless when $K = \infty$). Inequality~\eqref{eq:union-bound-coverage} does not make use of the structure in the construction of $\widehat{C}_k$. By making use of the nested set construction of $\widehat{C}_k$, we can obtain an coverage guarantee for $\widehat{C}_{N,\alpha}^{\mathrm{EFCP}}$ with a much weaker dependence on $K$. The following result proves such a predictive coverage guarantee for $\widehat{C}_{N,\alpha}^{\mathrm{EFCP}}$. 
% We use the notation $a\vee b = \max\{a, b\}$.
\begin{thm}\label{thm:coverage-guarantee-minimum-width}
If $Z_1, \ldots, Z_N, Z_{N+1}$ are independent and identically distributed, then the efficiency first conformal prediction region $\widehat{C}_{N,\alpha}^{\mathrm{EFCP}}$ from Algorithm~\ref{alg:efficiency-first-conformal} satisfies
\begin{equation}\label{eq:probability-unconditional-guarantee}
\mathbb{P}\left(Z_{N+1} \in \widehat{C}_{N,\alpha}^{\mathrm{EFCP}}\,\big|\,\mathcal{D}_1\right) ~\ge~ \left(1 + \frac{1}{|\mathcal{I}_2|}\right)(1 - \alpha) - \frac{\sqrt{\log(2K)/2} + 1/3}{\sqrt{|\mathcal{I}_2|}}.
\end{equation}
Moreover, for any $\delta > 0$, with probability at least $1 - \delta$,
\begin{equation}\label{eq:probability-conditional-guarantee}
    \mathbb{P}\left(Z_{N+1} \in \widehat{C}_{N,\alpha}^{\mathrm{EFCP}}\,\big|\,\mathcal{D}_1, \mathcal{D}_2\right) ~\ge~ \left(1 + \frac{1}{|\mathcal{I}_2|}\right)(1 - \alpha) - \frac{\sqrt{\log(2K/\delta)/2}}{\sqrt{|\mathcal{I}_2|}}.
\end{equation}
\end{thm}
The proof of Theorem~\ref{thm:coverage-guarantee-minimum-width} (in Section~\ref{appsec:proof-coverage-guarantee-minimum-width} of the supplementary file) is based on the fact that nested sequence of sets is a VC class of dimension 1; see Section 2.6 of~\cite{van1996weak} for definitions. We also use~\cite{massart1990tight} version of the DKW inequality for empirical distribution functions. 
% The constant $1.1$ in Theorem~\ref{thm:coverage-guarantee-minimum-width} can be improved to $\sqrt{\log(2)/2}\vee 1.0841|\mathcal{I}_2|^{-1/6}$; this improvement follows from Theorem 1 of~\cite{massart1990tight}. 
The result can be further improved numerically using inequality (2.11) of~\cite{massart1990tight}. After the mention of DKW inequality, Theorem~\ref{thm:coverage-guarantee-minimum-width} might seem obvious to some experts. But, to our knowledge, this is the first time this observation has been made explicit. Further, it is crucial for this result that we are using the split conformal prediction method; it may not hold for the full conformal prediction method. {One can, however, avoid sample splitting in Algorithm~\ref{alg:efficiency-first-conformal} by using the CV+ prediction sets from~\cite{barber2019predictive} instead of the split conformal prediction sets. This follows because CV+ prediction sets also satisfy training data conditional coverage~\citep{bian2022training}.}

It is of interest to also note that the proof of Theorem~\ref{thm:coverage-guarantee-minimum-width} also yields a two-sided bound on the probability (assuming that the scores $t_k(Z_i)$'s are almost surely distinct): 
\[
\left|\mathbb{P}\left(Z_{N+1}\in\widehat{C}_{N,\alpha}^{\mathrm{EFCP}}\,\big|\,\mathcal{D}_1\right) - \frac{\lceil(1-\alpha)(1 + |\mathcal{I}_2|)\rceil}{|\mathcal{I}_2|}\right| ~\le~ \frac{\sqrt{\log(2K)/2} + 1/3}{\sqrt{|\mathcal{I}_2|}}.
\]
Note that Theorem~\ref{thm:coverage-guarantee-minimum-width} requires $K$ to be finite for a non-trivial validity guarantee. (With $K = \infty$, the coverage guarantee is obsolete.) Given the finite sample nature of the result, we can take $K$ to depend on the sample size $|\mathcal{I}_2|$ and because the dependence on $K$ is logarithmic, we can take $K = \exp(|\mathcal{I}_2|^{\gamma})$ for some $\gamma\in[0, 1)$. With this choice, the coverage guarantee becomes
\[
\mathbb{P}\left(Z_{N+1} \in \widehat{C}_{N,\alpha}^{\mathrm{EFCP}}\,\big|\,\mathcal{D}_1\right) ~\ge~ \left(1 + \frac{1}{|\mathcal{I}_2|}\right)(1 - \alpha) - \frac{\mathfrak{C}}{|\mathcal{I}_2|^{(1-\gamma)/2}}.
\]

Combining inequalities~\eqref{eq:minimum-width-empirical} and~\eqref{eq:probability-unconditional-guarantee}, we obtain that the prediction set $\widehat{C}_{N,\alpha}^{\mathrm{EFCP}}$ has the smallest width among the prediction sets constructed and also has an asymptotic validity guarantee as $|\mathcal{I}_2|$ diverges. Interestingly, the dependence on the number of training methods is logarithmic and the guarantee holds true irrespective of how complicated the training methods $\mathcal{A}_k, k\in[K]$ are and how the nested sets $\mathcal{F}_t^{(k)}, t\in\mathcal{T}$ are constructed. This guarantee also does not depend on the dimension of the underlying data domain. The logarithmic dependence can be improved further when $K = \infty$ if $\{\mathcal{F}_t^{(k)}: t\in\mathcal{T}, k\in[K]\}$ is a VC class or has ``low complexity.'' Such a low complexity assumption holds true when the training methods are parametric, but in these cases, $\log K$ will be replaced by the VC dimension of the class of training methods; see Section~\ref{sec:general-algorithm} for an example. It also holds true in the case of kernel density estimation ~\citep{lei2013distribution}. The proof of Theorem~\ref{thm:coverage-guarantee-minimum-width} can also be used to justify choosing the smallest prediction set when the split conformal method is applied multiple times (with multiple splits of the data). Hence, Theorem~\ref{thm:coverage-guarantee-minimum-width} provides the validity guarantee even after cherry-picking a data split from multiple splits of the data. This also implies that one can make better use of the data even with the split conformal method. {Given the importance, we present a formal result here. As in Algorithm~\ref{alg:efficiency-first-conformal}, let $\mathcal{D}$ represent the data. For $1\le k\le K$, let $(\mathcal{D}_{1,k}, \mathcal{D}_{2,k})$ denote $K$ random splits of the data into two parts; without loss of generality, assume $|\mathcal{D}_{1,k}| = |\mathcal{D}_{2,k}|= N/2$. Use algorithm $\mathcal{A}_k$ on $\mathcal{D}_{1,k}$ to find nested sets $\{\mathcal{F}_t^{(k)}\}_{t\in\mathcal{T}}$ and the corresponding prediction set $\widehat{C}_k := \{z:\,t_k(z) \le T_{\alpha,k}\}.$
(This step is exactly the same as Step 3 of Algorithm~\ref{alg:efficiency-first-conformal} with $T_{\alpha,k}$ computed on $t_k(Z_i), Z_i\in\mathcal{D}_{2,k}$). Find $\widehat{k}$ as in step 4 of Algorithm~\ref{alg:efficiency-first-conformal} and return $\widehat{C}_{\widehat{k}}$. In the setting here, one can use the same algorithm on multiple splits, but in general, different algorithms are permitted. From the proof of Theorem~\ref{thm:coverage-guarantee-minimum-width}, we obtain the following result (proved in Section~\ref{appsec:proof-of-thm-multiple-splits}).
\begin{thm}\label{thm:multiple-splits-conformal}
    If $Z_1, \ldots, Z_N, Z_{N+1}$ are independent and identically distributed, then the best width prediction set from multiple splits satisfies
    \[
    \mathbb{P}(Z_{N+1}\in\widehat{C}_{\widehat{k}}) ~\ge~ \left(1 + \frac{2}{N}\right)(1 - \alpha) - \frac{\sqrt{\log(2K)/2} + 1/3}{\sqrt{N/2}}. 
    \]
\end{thm}
An advantage of the method in Theorem~\ref{thm:multiple-splits-conformal} compared to that in Theorem~\ref{thm:coverage-guarantee-minimum-width} is that there is no separate sample splitting required and hence makes better use of the data. 
\cite{lei2018distribution} considers a similar combination of split conformal prediction sets from multiple splits of the same data. Instead of picking the smallest prediction set from different splits, they consider the intersection of the prediction sets $\widetilde{C}_K = \cap_{k=1}^K \widehat{C}_k$. Our results do not apply to $\widetilde{C}_K$, but provide an interesting alternative to consider.}

Regarding the width of $\widehat{C}_{N,\alpha}^{\mathrm{EFCP}}$, equality~\eqref{eq:minimum-width-empirical} can be called an ``empirical'' oracle inequality, i.e.,
\begin{equation}\label{eq:empirical-oracle-inequality-EFCP}
\mbox{Width}(\widehat{C}_{N,\alpha}^{\mathrm{EFCP}}) ~\le~ \min_{1\le k\le K}\mbox{Width}(\widehat{C}_k) ~+~ 0.
\end{equation}
It is referred to as ``empirical'' because the right hand side can be random through the prediction set $\widehat{C}_k$ (and also potentially through $\mbox{Width}(\cdot)$ as in~\eqref{eq:width-x-cross-section}). This kind of empirical guarantee might be preferable in practice because the prediction sets $\widehat{C}_k, k\in[K]$ will often remain fixed as predictions are being done for future observations. This is particularly true in regression setting where future predictions are done for future values of covariates.
In contrast, one might ask for an ``asymptotic'' oracle inequality where $\widehat{C}_k$ is replaced by its asymptotic analog $C_k$. Such an asymptotic oracle inequality will be discussed in Sections~\ref{subsec:empirical-asymptotic-oracle} and~\ref{sec:linear-prediction-fixed-width}. 
% Finally, about $\widehat{C}_{N,\alpha}^{\mathrm{EFCP}}$, it is worth pointing out that its construction only requires two splits of the data similar to the usual split conformal method. 
{\paragraph{A comment on conditional validity.}
Theorems~\ref{thm:coverage-guarantee-minimum-width} and~\ref{thm:multiple-splits-conformal} focus on marginal coverage validity. In the regression/classification settings with covariates and a response, coverage conditional on covariates is of greater importance in practice. If the prediction sets $\widehat{C}_k$ constructed in Step 3 of Algorithm~\ref{alg:efficiency-first-conformal} satisfy asymptotic conditional coverage guarantee, then the EFCP prediction set is also asymptotically valid for any fixed $K$. We present a specific instance here using CQR sets that are known to be asymptotically conditionally valid~\citep{sesia2020comparison}. Let $\widehat{C}_k := \{(x, y):\,y\in[\widehat{q}_{\alpha/2}^{(k)}(x) - T_{\alpha,k}, \widehat{q}_{1-\alpha/2}^{(k)}(x) + T_{\alpha,k}]\}$ be the CQR prediction sets based on quantile estimators $\widehat{q}_{\alpha/2}^{(k)}(\cdot)$ and $\widehat{q}_{1-\alpha/2}^{(k)}(\cdot)$. Also, let $q^*_{\alpha/2}(\cdot)$ and $q^*_{1-\alpha/2}(\cdot)$ denote the true conditional quantiles of $Y|X$ at levels $\alpha/2, 1 - \alpha/2$, respectively. Then with probability at least $1 - \delta$,
\begin{equation}\label{eq:conditional-coverage-EFCP-CQR}
\mathbb{P}\left((X_{N+1}, Y_{N+1})\in\widehat{C}_{N,\alpha}^{\mathrm{EFCP}}\big|X_{N+1} = x, \mathcal{D}_1\right) \ge 1 - \alpha - \frac{1}{r}\sqrt{\frac{\log(K/\delta)}{2|\mathcal{I}_2|}} -\rho_N - M(1 + 1/r)\eta_N,
\end{equation}
if
\begin{equation}\label{eq:uniform-quantile-consistency}
\mathbb{P}\left(\max_{1\le k\le K}\|\widehat{q}_{\alpha/2}^{(k)} - q^*_{\alpha/2}\|_{\infty} + \max_{1\le k\le K}\|\widehat{q}_{1-\alpha/2}^{(k)} - q^*_{1-\alpha/2}\|_{\infty} \le \eta_N\right) \ge 1 - \rho_N,
\end{equation}
and $t^*(X, Y) = \max\{q^*_{\alpha/2}(X) - Y, Y - q^*_{1-\alpha/2}(X)\}$ has density bounded from below by $r > 0$ on $[-r^*, r^*]$ with $r^* \ge \eta_n + \sqrt{\log(2K/\delta)/(2|\mathcal{I}_2|)}$. See Appendix~\ref{appsec:proof-of-conditional-coverage-EFCP-CQR} for a proof. Assumption~\ref{eq:uniform-quantile-consistency} is similar to Assumption 3 of~\cite{sesia2020comparison}, but with uniform norm instead of $L_2$-norm.
}
%%%%%%%%%%%%%%%%%%%%%%%%%%%%%%%%%%%%%%%%%%%%%%%%%%%%%%%%%%%%%%%
%%%%%%%%%%%%%%%%%%%%%%%%%%%%%%%%%%%%%%%%%%%%%%%%%%%%%%%%%%%%%%%
\section{Validity First Conformal Prediction}\label{sec:VFCP}
Algorithm~\ref{alg:efficiency-first-conformal} targets finite sample efficiency by sacrificing finite sample validity in coverage. We now present the validity first conformal prediction (VFCP) method in Algorithm~\ref{alg:validity-first-conformal} that targets finite sample validity while only attaining approximately the smallest width. {In short, the method requires one additional split of data and runs the split conformal calibration step with the chosen ``model/algorithm'' on the third split of the data. This additional step readily yields finite sample validity for VFCP but incurs a potential loss of efficiency due to the reduction in the training data. Although arguably inferior to EFCP in terms of data usage, VFCP precedes EFCP in the literature (originally introduced in~\cite{lei2013distribution}) and serves as a useful benchmark for comparison. The main goal of this section is to study how the width of the prediction set returned by VFCP performs. Theorem~\ref{thm:empirical-oracle-inequality-general-VFCP} proves that the width of the VFCP set can be significantly larger than that of EFCP and hence, implies a loss of power due to additional splitting of data. Due to these reasons, we recommend EFCP over VFCP, in practice.}
\begin{algorithm}[htbp]
    \SetAlgoLined
    \SetEndCharOfAlgoLine{}
    \KwIn{Data $\mathcal{D} = \{Z_i, i\in [N]\}$, coverage probability $1 - \alpha$, and training methods $\mathcal{A}_k, k\in [K]$.}
    \KwOut{A valid prediction set $\widehat{C}^{\mathrm{VFCP}}_{N,\alpha}$ with width close to smallest width/volume.}
    Randomly split $[N] = \{1, 2, \ldots, N\}$ into three (disjoint) parts $\mathcal{I}_1, \mathcal{I}_2, \mathcal{I}_3$. Set $\mathcal{D}_1 = \{Z_i: i\in\mathcal{I}_1\}$, $\mathcal{D}_2 = \{Z_i: i\in \mathcal{I}_2\}$, and $\mathcal{D}_3 = \{Z_i: i\in \mathcal{I}_3\}$.\;
    % Fit training methods $\mathcal{A}_1, \ldots, \mathcal{A}_K$ on $\mathcal{D}_1$ and using fitted method $\mathcal{A}_k$, construct an increasing (nested) sequence of sets $\{\mathcal{F}_t^{(k)}\}_{t\in\mathcal{T}}$. Here $\mathcal{T}$ is a subset of $\mathbb{R}$.\;
    % % \hspace{0.1in} 
    % % For each $i\in\mathcal{I}_2$ (or equivalently $Z_i\in\mathcal{D}_2$) and $k\in[K]$, define the conformal score
    % % \[
    % % t_k(Z_i) := \inf\{t\in\mathcal{T}:\,Z_i\in\mathcal{F}_t^{(k)}\}.
    % % \]
    % % Compute the corresponding conformal prediction set as
    % % \[
    % % \widehat{C}_k ~:=~ \left\{z:\, t_k(z) \le T_{\alpha,k}\right\} = \mathcal{F}_{T_{\alpha,k}}^{(k)},
    % % \]
    % % where $T_{\alpha,k} := \lceil(1 + |\mathcal{I}_2|)(1-\alpha)\rceil\mbox{-th largest value of }t_k(Z_i), Z_i\in\mathcal{D}_2.$
    % Compute conformal prediction set $\widehat{C}_k$ based on $\{\mathcal{F}_t^{k)}\}_{t\in\mathcal{T}}$ as in step 3 of Algorithm~\ref{alg:efficiency-first-conformal}.\;
    % % \hspace{0.1in} 
    % Set $$\widehat{k} := \argmin_{1\le k\le K}\,\mbox{Width}(\widehat{C}_k),$$
    % % The quantity $\mbox{Width}(\cdot)$ can be any measure of width or volume of a prediction set.
    % \;
    {Perform Steps 2, 3, and 4 of Algorithm~\ref{alg:efficiency-first-conformal}.}\;
    For each $i\in\mathcal{I}_3$ (i.e., $Z_i\in\mathcal{D}_3$), define the conformal score
    % \[
    $t_{\widehat{k}}(Z_i) := \inf\{t\in\mathcal{T}:\,Z_i\in\mathcal{F}_t^{(\widehat{k})}\}.$
    % \]
    Compute the corresponding conformal prediction set as 
    % \[
    $\widehat{C}_{\widehat{k}}^* := \{z:\, t_{\widehat{k}}(z) \le {T}^*_{\alpha,\widehat{k}}\},$
    % \]
    where $T_{\alpha,\widehat{k}}^* := \lceil(1 + |\mathcal{I}_3|)(1 - \alpha)\rceil$-th largest element of $t_{\widehat{k}}(Z_i), i\in\mathcal{D}_3$.\;
    % Report the conformal prediction set $\widehat{C}^{\mathrm{VFCP}}_{N,\alpha}$ obtained from $\mathcal{D}_3$ and $\mathcal{A}_{\widehat{k}}$.\;
    \Return the prediction set $\widehat{C}_{\widehat{k}}^*$ as $\widehat{C}_{N,\alpha}^{\mathrm{VFCP}}$.
    \caption{Validity First Conformal Prediction for the Smallest Prediction Set}
    \label{alg:validity-first-conformal}
\end{algorithm}
% Theorem~\ref{thm:coverage} proves that the prediction set  $\widehat{C}^{\mathrm{VFCP}}_{N,\alpha}$ returned by Algorithm~\ref{alg:validity-first-conformal} has a guaranteed $1 - \alpha$ coverage probability in finite samples.

% Theorem~\ref{thm:coverage-guarantee-minimum-width} proves approximate validity of prediction coverage of~$\widehat{C}_{\widehat{k}}$. In contrast, the following result proves that the prediction set $\widehat{C}_{\alpha}^{\mathrm{final}}$ returned by Algorithm~\ref{alg:validity-first-conformal} has a finite sample coverage validity.
\begin{thm}\label{thm:coverage}
The prediction interval $\widehat{C}^{\mathrm{VFCP}}_{N,\alpha}$ from Algorithm~\ref{alg:validity-first-conformal} satisfies
$\mathbb{P}(Z_{N+1} \in \widehat{C}^{\mathrm{VFCP}}_{N,\alpha}) \ge 1 - \alpha,$
whenever $Z_i, 1\le i\le N, Z_{N+1}$ form an exchangeable sequence of random variables. Here the probability is computed with respect to $Z_{N+1}$ and $Z_1, \ldots, Z_N$. 
\end{thm}
\begin{proof}
The result follows from the guarantees of the split conformal method. See, for example, Lemma 2 of~\cite{romano2019conformalized}.  Proposition 2a of~\cite{vovk2012conditional} proves the conditional-on-data prediction coverage guarantee which also involves a slack similar to that in Theorem~\ref{thm:coverage-guarantee-minimum-width}.
% The conformal prediction set $\widehat{C}_{N,\alpha}^{\mathrm{VFCP}}$ is constructed based on the training method $\mathcal{A}_{\widehat{k}}$ and the conformal scores computed on data $\mathcal{D}_3$. 
% Conditional on $\mathcal{D}_1$ and $\mathcal{D}_2$, the conformal scores for $Z_{N+1}$ and $Z_i, i \in \mathcal{D}_3$ using $\widehat{k}$ are exchangeable (whenever $Z_{N+1}, Z_i$, $i \in \mathcal{D}_3$ are exchangeable). Hence, the discussion in Section 2 of \cite{lei2018distribution} proves the result. %Also, see Lemma 2 of~\cite{romano2019conformalized}.
\end{proof}
% Theorem~\ref{thm:coverage} can be compared to Theorem~\ref{thm:coverage-guarantee-minimum-width}. 
% Unlike Theorem~\ref{thm:coverage-guarantee-minimum-width}, Theorem~\ref{thm:coverage} yields a finite sample guarantee of $1 - \alpha$ (i.e., without any slack). However, Theorem~\ref{thm:coverage-guarantee-minimum-width} also provides a coverage guarantee conditional on the full data (i.e., $\mathcal{D}_1\cup\mathcal{D}_2$ in Algorithm~\ref{alg:efficiency-first-conformal}).

Although $\widehat{C}^{\mathrm{VFCP}}_{N,\alpha}$ has a guaranteed coverage, it may not have width close to minimum without some regularity conditions.
The reason is two-fold:
% \begin{enumerate}
    % \item 
    (1) Proving that width of $\widehat{C}_{\widehat{k}}^* = \widehat{C}_{N,\alpha}^{\mathrm{VFCP}}$ and $\min_k \mbox{Width}(\widehat{C}_k)$ are close requires proving that the quantiles $T_{\alpha,\widehat{k}}$ (in step 3 of Algorithm~\ref{alg:efficiency-first-conformal}) and $T^*_{\alpha,\widehat{k}}$ are close. Proving quantiles from two distribution functions are close requires some continuity assumptions on the inverse distribution functions;
    % \item 
    (2) Even with close quantiles, we still need certain continuity assumptions on $t\mapsto \mbox{Width}(\mathcal{F}_t^{(k)})$.
% \end{enumerate}
Further, we need a quantitative version of continuity in order to quantify how close the width of $\widehat{C}_{N,\alpha}^{\mathrm{VFCP}}$ is to the smallest width. Such a quantification can be considered in terms of an empirical oracle inequality similar to~\eqref{eq:empirical-oracle-inequality-EFCP}:
\begin{equation}\label{eq:empirical-oracle-inequality-VFCP}
\mbox{Width}(\widehat{C}^{\mathrm{VFCP}}_{N,\alpha}) ~\le~ \min_{1\le k\le K}\mbox{Width}(\widehat{C}_k) ~+~ R_N^*,    
\end{equation}
for some $R_N^*$ converging to zero (in probability) as $\min\{|\mathcal{I}_j|, 1\le j\le 3\} \to \infty$. Note that~\eqref{eq:empirical-oracle-inequality-VFCP} implies that $\widehat{C}_{N,\alpha}^{\mathrm{VFCP}}$ is ``close'' to $\widehat{C}_{N,\alpha}^{\mathrm{EFCP}}$, which already satisfies~\eqref{eq:minimum-width-empirical} and~\eqref{eq:empirical-oracle-inequality-EFCP}. It is worth recalling that $\widehat{C}_{N,\alpha}^{\mathrm{VFCP}}$ and $\widehat{C}_{N,\alpha}^{\mathrm{EFCP}}$ are computed based on different datasets $\mathcal{D}_3$ and $\mathcal{D}_2$, respectively. Finally, we observe that prediction sets, unlike confidence sets, do not usually shrink to a single point as the sample size increases. This implies that $\min_k\mbox{Width}(\widehat{C}_k)$ does not shrink to zero as $N\to\infty$ and hence~\eqref{eq:empirical-oracle-inequality-VFCP} implies that $\widehat{C}_{N,\alpha}^{\mathrm{VFCP}}$ converges to the smallest prediction set.

The continuity assumptions mentioned above are satisfied in some examples, but are certainly restrictive when compared to the only i.i.d. assumption required for the coverage as well as width properties of $\widehat{C}_{N,\alpha}^{\mathrm{EFCP}}$ from Algorithm~\ref{alg:efficiency-first-conformal}.
% The investigation of this question is our main focus in this paper.
% A general result proving that $\widehat{C}^{\mathrm{VFCP}}_{N,\alpha}$ has width close to the smallest in the general setting of Algorithm~\ref{alg:validity-first-conformal} seems out of reach currently. 
We will now formalize the continuity assumptions and prove a general empirical oracle   inequality~\eqref{eq:empirical-oracle-inequality-VFCP}. Recall from Algorithm~\ref{alg:efficiency-first-conformal} that $t_k(z) = \inf\{t\in\mathcal{T}:\, z\in\mathcal{F}_t^{(k)}\}$. Also, recall that $\mathcal{F}_t^{(k)}$ are \emph{random} nested sets random through $\mathcal{D}_1$. Define $F_k(\cdot)$ as the cumulative distribution function $t_k(Z_i)$ conditional on $\mathcal{D}_1$, i.e.,
% \[
$F_k(s) = \mathbb{P}\left(t_k(Z_{N+1}) \le s\big|\mathcal{D}_1\right).$
% \]
Note that this is a random cumulative distribution function, random through $\mathcal{D}_1$.
Define $\widehat{Q}_{\alpha,k}$ as the $(1-\alpha)$-th quantile of $F_k(\cdot)$; the hat-notation is to stress the randomness of the quantile.
\begin{enumerate}[label=\bf(EO\arabic*)]
    \item For each $k\in[K]$, there exists $r^*, \gamma\in(0, 1]$ such that $F_k^{-1}(\cdot)$ is $\gamma$-H{\"o}lder continuous on $[\widehat{Q}_{\alpha,k} - r^*, \widehat{Q}_{\alpha,k} + r^*]$ with H{\"o}lder continuity constant $L_k$, i.e., for all $k\in[K],$ $q_1, q_2\in[\widehat{Q}_{\alpha,k} - r^*,\, \widehat{Q}_{\alpha,k} + r^*]$, $|F_k^{-1}(q_1) - F_k^{-1}(q_2)| \le L_k|q_1 - q_2|^{\gamma}.$
    \label{eq:Holder-VFCP}
%     \[
% ,\quad\mbox{for all}\quad .
%     \]
    \item For each $k\in[K]$, the set-valued map $t\mapsto\mathcal{F}_t^{(k)}$ is continuous and the function $t\mapsto\mbox{Width}(\mathcal{F}_t^{(k)})$ is Lipschitz continuous with Lipschitz constant $L_W$.\label{eq:width-continuity}
\end{enumerate}
Assumption~\ref{eq:Holder-VFCP} is necessary to ensure that the quantiles from empirical distribution of $t_k(Z_i), i\in\mathcal{I}_2$ are close to those from $F_k(\cdot)$. An analogue assumption is often made in the study of median/quantile estimation. In that setting, the assumption is often presented in terms of density at median/quantile bounded away from zero, which implies~\ref{eq:Holder-VFCP} with $\gamma = 1$; also, see Theorem 3.1 of~\cite{lei2018distribution} for a similar assumption. Assumption~\ref{eq:width-continuity} is used to imply that the widths of prediction sets are close when the quantiles are close. For some intuition, we verify this assumption in the case of fixed width and conformalized quantile regression methods described in Section~\ref{sec:EFCP}. For fixed width regression, $\mathcal{F}_t^{(k)} = \{(x, y):\,|y-\widehat{m}_k(x)| \le t\}$. The Lebesgue measure of $x$-cross-section of this set is $2t$ and hence $\mbox{Width}(\mathcal{F}_t^{(k)}) = 2t$, which implies~\ref{eq:width-continuity}. For conformalized quantile regression,~\eqref{eq:width-x-cross-section} readily proves assumption~\ref{eq:width-continuity}. It may not always be this easy to verify~\ref{eq:width-continuity}. For example, in the case of the density level sets in Section~\ref{sec:EFCP}, assumption~\ref{eq:width-continuity} might fail to hold true if the density estimates $\widehat{p}_{h_k}(\cdot)$ are piecewise constant functions (e.g., histograms). 
% It may also fail to hold in the variants of conformalized quantile regression mentioned in Table~1 of~\cite{gupta2019nested}.

It is important to note that assumption~\ref{eq:Holder-VFCP} depends on $\alpha\in[0, 1]$ (the required miscoverage probability). Also, the H{\"o}lder and Lipschitz continuity constants $L_{k}$ and $L_{W}$ in~\ref{eq:Holder-VFCP},~\ref{eq:width-continuity} are random variables in general, random through $\mathcal{D}_1$ and $\mathcal{D}_2$; see, e.g.,~\eqref{eq:width-x-cross-section}.

Under these two assumptions, we prove two versions of the empirical oracle inequality~\eqref{eq:empirical-oracle-inequality-VFCP}.
\begin{thm}\label{thm:empirical-oracle-inequality-general-VFCP}
Suppose $Z_1, \ldots, Z_N, Z_{N+1}$ are independent and identically distributed and assumption~\ref{eq:width-continuity} holds true. Fix $\delta\in[0, 1]$. If assumption~\ref{eq:Holder-VFCP} holds true with
\begin{equation}\label{eq:quantile-holder-domain-requirement}
r^* \ge \max\left\{\sqrt{\frac{\log(4K/\delta)}{2\min\{|\mathcal{I}_2|,|\mathcal{I}_3|\}}},\, \frac{2}{\min\{|\mathcal{I}_2|,|\mathcal{I}_3|\}}\right\},
\end{equation}
then with probability at least $1 - \delta$,
\begin{equation}\label{eq:proved-empirical-oracle-ineq-VFCP}
\mathrm{Width}(\widehat{C}_{N,\alpha}^{\mathrm{VFCP}}) ~\le~ \min_{1\le k\le K}\,\mathrm{Width}(\widehat{C}_k) + 3L_WL_{[K]}\left[\left(\frac{\log(4K/\delta)}{\min\{|\mathcal{I}_2|, |\mathcal{I}_3|\}}\right)^{\gamma/2} + \left(\frac{2}{\min\{|\mathcal{I}_2|, |\mathcal{I}_3|\}}\right)^{\gamma}\right].
\end{equation}
Here $L_{[K]} := \max_{k\in[K]}L_k$. Moreover, the result holds true with $\widehat{C}_k$ %on the right hand side
replaced with $\mathcal{F}_{\widehat{Q}_{\alpha,k}}^{(k)}$.
\end{thm}
The proof of Theorem~\ref{thm:empirical-oracle-inequality-general-VFCP} (in Section~\ref{appsec:proof-thm-empirical-oracle-inequality-general-VFCP} of the supplementary file) follows the strategy of first proving the closeness of quantiles $T_{\alpha,k}$ and $T^*_{\alpha,k}$ is uniform in $k$ using Massart's inequality. Then the width continuity assumed in~\ref{eq:width-continuity} implies the result. Similar to the coverage of EFCP method in Theorem~\ref{thm:coverage-guarantee-minimum-width}, Theorem~\ref{thm:empirical-oracle-inequality-general-VFCP} requires a finite number $K$ of methods and has a logarithmic dependence in the slack. Once again, similar to EFCP, $K < \infty$ can be relaxed if the class of nested sets $\mathcal{F}_t^{(k)}, t\in\mathcal{T}, k\ge1$ forms a VC class. Finally, Theorem~\ref{thm:empirical-oracle-inequality-general-VFCP} also proves an oracle inequality with $\widehat{C}_k$ replaced by $\mathcal{F}_{\widehat{Q}_{\alpha,k}}^{(k)}$. The set $\mathcal{F}_{\widehat{Q}_{\alpha,k}}^{(k)}$ is random through $\mathcal{D}_1$ only and is the asymptotic version of $\widehat{C}_k$ obtained by taking $|\mathcal{I}_2| = |\mathcal{D}_2| = \infty$. 
Under assumption~\ref{eq:Holder-VFCP}, the EFCP prediction set $\widehat{C}_{N,\alpha}^{\mathrm{EFCP}}$ also satisfies the oracle inequality with $\mathcal{F}_{\widehat{Q}_{\alpha,k}}^{(k)}$. This follows from the proof of Theorem~\ref{thm:empirical-oracle-inequality-general-VFCP}.

It is worth pointing out a subtle aspect of Theorem~\ref{thm:empirical-oracle-inequality-general-VFCP} which stems from the definition of $\mbox{Width}(\cdot)$. In the case of density level sets, width of any prediction set is the Lebesgue measure and hence is independent of the data on which the prediction set is computed. In the fixed width regression context too, width of the prediction set under consideration is just a function of the prediction set. However, in general regression setting, width can depend on the data and not just on the prediction set. For example, in the conformalized quantile regression setting, recall from~\eqref{eq:width-x-cross-section} the width of $\widehat{C}_k$.
% \[
% \mbox{Width}(\widehat{C}_k) = 2T_{\alpha,k} + \frac{1}{|\mathcal{I}_2|}\sum_{i\in\mathcal{I}_2} \bigl( \widehat{q}_{1-\alpha/2}^{(k)}(X_i) - \widehat{q}_{\alpha/2}^{(k)}(X_i) \bigr).
% \]
Hence, the width of $\widehat{C}_{N,\alpha}^{\mathrm{VFCP}}$ is given by
\[
\mbox{Width}(\widehat{C}_{N,\alpha}^{\mathrm{VFCP}}) = 2T_{\alpha,\widehat{k}}^* + \frac{1}{|\mathcal{I}_2|}\sum_{i\in\mathcal{I}_2} \bigl( \widehat{q}_{1-\alpha/2}^{(\widehat{k})}(X_i) - \widehat{q}_{\alpha/2}^{(\widehat{k})}(X_i) \bigr),
\]
with $T_{\alpha,\widehat{k}}^*$ defined in step 5 of Algorithm~\ref{alg:validity-first-conformal}. The implication is that even though the prediction set $\widehat{C}_{N,\alpha}^{\mathrm{VFCP}}$ is defined using $\mathcal{D}_3$, its width depends on both $\mathcal{D}_2$ and $\mathcal{D}_3$. In this case, comparison of widths is the same as comparison of quantiles $T_{\alpha,\widehat{k}}^*$ and $T_{\alpha,\widehat{k}}$ and hence, Theorem~\ref{thm:empirical-oracle-inequality-general-VFCP} implies approximate efficiency in terms of the population $x$-cross-section width. This does not hold true for general nested methods and other variants of conformalized quantile regression methods. 

As mentioned before, in the context of quantile/median estimation, assumption~\ref{eq:Holder-VFCP} is commonly used with $\gamma = 1$ and a fixed $r^* > 0$. In Theorem~\ref{thm:empirical-oracle-inequality-general-VFCP}, we allow for $r^*$ to converge to zero. If $r^* > 0$ (independent of $N$) and $\log(K)/\min\{|\mathcal{I}_2|, |\mathcal{I}_3|\} \to 0$ as $N\to\infty$, then the requirement~\eqref{eq:r-star-requirement} would be asymptotically satisfied. If assumption~\ref{eq:Holder-VFCP} is satisfied with $\gamma = 1$, then interestingly the slack in the width of $\widehat{C}_{N,\alpha}^{\mathrm{VFCP}}$ is of the same order as the slack in the coverage of $\widehat{C}_{N,\alpha}^{\mathrm{EFCP}}$.
%%%%%%%%%%%%%%%%%%%%%%%%%%%%%%%%%%%%%%%%%%%%%%%%%%%%%%%%%%%
%%%%%%%%%%%%%%%%%%%%%%%%%%%%%%%%%%%%%%%%%%%%%%%%%%%%%%%%%%%
\section{Comparison of EFCP and VFCP}\label{sec:comparison-EFCP-VFCP}
Algorithms~\ref{alg:efficiency-first-conformal} and~\ref{alg:validity-first-conformal} both aim to obtain close-to-smallest width among the constructed prediction sets. Respectively, these methods primarily target finite sample efficiency and finite sample validity. Algorithm~\ref{alg:efficiency-first-conformal} attains required coverage (or validity) only approximately. Algorithm~\ref{alg:validity-first-conformal} attains the smallest width (or efficiency) only approximately. Both algorithms asymptotically are efficient and valid. These comparisons are summarized in Table~\ref{tab:comparison-VFCP-EFCP}. {Out of the two methods, EFCP is the more favorable choice.}
\begin{table}[!h]
\centering
\caption{Comparison of efficiency first conformal prediction (EFCP) algorithm~\ref{alg:efficiency-first-conformal} and validity first conformal prediction (VFCP) algorithm~\ref{alg:validity-first-conformal} in terms of coverage and width slack as well as the required assumptions. Coverage slack represents the finite sample coverage guarantee minus the required coverage of $1 - \alpha$ and width slack represents finite sample width guarantee minus the minimum width $\min_{k\in[K]}\mbox{Width}(\widehat{C}_k)$. Theorem~\ref{thm:coverage-guarantee-minimum-width} and Eq.~\eqref{eq:empirical-oracle-inequality-EFCP} prove the properties for EFCP. Theorems~\ref{thm:coverage} and~\ref{thm:empirical-oracle-inequality-general-VFCP} prove the properties of VFCP.}
\setlength\extrarowheight{-1pt}
\begin{tabular}{ccccl}
\hline
 & \begin{tabular}[c]{@{}c@{}}Number\\ of Splits\end{tabular} & Coverage Slack & Width Slack & Assumptions \\\hline\hline
\begin{tabular}[c]{@{}c@{}}EFCP\\ (Algorithm~\ref{alg:efficiency-first-conformal})\end{tabular} & 2 & $\displaystyle-\mathfrak{C} \left({\frac{\log(K)}{|\mathcal{I}_2|}}\right)^{1/2}$ & 0 & i.i.d. data. \\\hline
\begin{tabular}[c]{@{}c@{}}VFCP\\ (Algorithm~\ref{alg:validity-first-conformal})\end{tabular} & 3 & $0$ & $\displaystyle\mathfrak{C}L_WL_{[K]}\left(\frac{\log(K/\delta)}{\min\{|\mathcal{I}_2|,|\mathcal{I}_3|\}}\right)^{\gamma/2}$ & \begin{tabular}[c]{@{}l@{}}i.i.d. data,\\ H{\"o}lder continuity~\ref{eq:Holder-VFCP},\\ Width continuity~\ref{eq:width-continuity}.\end{tabular}\\
\hline
\end{tabular}
\label{tab:comparison-VFCP-EFCP}
\end{table}

Now comparing EFCP and VFCP in terms of the assumptions required, EFCP is the clear winner. The finite sample efficiency and validity properties of EFCP only require the observations to be independent and identically distributed, while the analogues results for VFCP require certain regularity assumptions on the nested sets used and the training methods via the assumptions~\ref{eq:width-continuity} and~\ref{eq:Holder-VFCP}. It is relatively easy to verify if assumption~\ref{eq:width-continuity} holds true or not for a given class of nested sets, but it seems difficult to verify assumption~\ref{eq:Holder-VFCP} in practice (at least for complicated non-parametric training methods $\mathcal{A}_1, \ldots, \mathcal{A}_K$). Finally, from a sample usage aspect, EFCP only requires two splits while VFCP requires three splits. It is worth pointing out that one can replace $\alpha$ in EFCP Algorithm~\ref{alg:efficiency-first-conformal} with $\alpha - (\sqrt{\log(2K)/2} + 1/3)/\sqrt{|\mathcal{I}_2|}$ to attain finite sample $1-  \alpha$ coverage validity with EFCP. For a fair comparison of EFCP and VFCP, we did not incorporate this change in Algorithm~\ref{alg:efficiency-first-conformal}. Moreover, in practice, we notice that EFCP Algorithm~\ref{alg:efficiency-first-conformal} without any modifications already attains the required coverage; see the simulations in following subsections. Finally, we mention that assumption~\ref{eq:Holder-VFCP} depends on $\alpha\in(0, 1)$ and might lead to a large H{\"o}lder constant $L_k$ when $\alpha$ converges to zero. This leads to a worse empirical oracle inequality for VFCP. Such a dependence on $\alpha$ does not occur with EFCP. An empirical comparison of EFCP and VFCP in the context of conformalized quantile regression is given in supplementary Appendix~\ref{appsec:application-conformalized-quantile-regression}.

\subsection{Empirical vs Asymptotic Oracle Inequality}\label{subsec:empirical-asymptotic-oracle}
It should be stressed here that the efficiency discussed in this and the previous sections is via an empirical oracle inequality where we compare the outputted prediction region to the prediction regions that are random either through $\mathcal{D}_1$ or $\mathcal{D}_2$ or both. They may not satisfy an asymptotic oracle inequality in that if we take $C_k$ to be the asymptotic counterpart of $\widehat{C}_k$ in Algorithms~\ref{alg:efficiency-first-conformal} and~\ref{alg:validity-first-conformal}, then the widths of $\widehat{C}_{N,\alpha}^{\mathrm{EFCP}}$ and $\widehat{C}_{N,\alpha}^{\mathrm{VFCP}}$ need not be close to the minimum width of $C_k$. Of course, if the width of $\widehat{C}_k$ is uniformly (in $k$) close to width of $C_k$, then the empirical oracle inequalities proved in~\eqref{eq:empirical-oracle-inequality-EFCP} and Theorem~\ref{thm:empirical-oracle-inequality-general-VFCP} readily imply asymptotic oracle inequalities. 
% This will be the focus of the following sections.

In the following sections, we will consider Algorithms~\ref{alg:efficiency-first-conformal} and~\ref{alg:validity-first-conformal} in the regression context where the training methods correspond either to all linear predictors or ridge regressions with different penalty parameters. Because of this, we get $K = \infty$ which means we select a training method from infinitely many methods. This is possible because the set of all training methods is a VC class~\citep[Section 2.6]{van1996weak}. We will consider oracle versions of prediction sets. With $\widehat{C}_k, 1\le k\le K$ representing prediction sets as in step 2 of Algorithm~\ref{alg:efficiency-first-conformal}, one can construct non-random sets $C_k, 1\le k\le K$ such that $\widehat{C}_k$ and $C_k$ are close (in terms of symmetric difference) asymptotically. We will prove that
\begin{equation}\label{eq:asymptotic-oracle}
\max \bigl\{\mbox{Width}(\widehat{C}^{\mathrm{EFCP}}_{N,\alpha}), \mbox{Width}(\widehat{C}^{\mathrm{VFCP}}_{N,\alpha}) \bigr\} ~\le~ \min_{1\le k\le K}\mbox{Width}(C_k) ~+~ R_N,
\end{equation}
where $R_N$ converges (in probability) to zero as $\min\{|\mathcal{I}_j|, 1\le j\le 3\} \to \infty$. One may call~\eqref{eq:asymptotic-oracle} an ``asymptotic'' oracle inequality as opposed to the empirical oracle inequality which has $\widehat{C}_k$, in place of $C_k$ on the right hand side. 

\section{Fixed Width Prediction Sets in Regression}\label{sec:linear-prediction-fixed-width}
In this section, we consider Algorithms~\ref{alg:efficiency-first-conformal} and~\ref{alg:validity-first-conformal} in the regression setting and study oracle inequalities when the training methods are linear, including ridge regression.
{Formally, we consider two types of nested sets $\mathcal{F}_t^{(\theta)} := \{(x, y): y\in [\theta^{\top}x - t, \theta^{\top}x + t]\}, t\ge0$, for either $\theta\in\Theta\subseteq\mathbb{R}^d$ (with $\Theta$ a fixed, i.e., non-random set) or $\theta\in\widehat{\Theta}\subseteq\mathbb{R}^d$ (with $\widehat{\Theta}$ obtained from an independent sample). 
Although linear trainers are simplistic, their analysis can serve as a way to understand general aggregation~\citep{bunea2007aggregation,lecue2016performance}. In detail, suppose we have data $(X_i, Y_i), 1\le i\le N$. Split this into two parts $\mathcal{D}_1$ and $\mathcal{D}_2$ (with index sets $\mathcal{I}_1$ and $\mathcal{I}_2$). Construct $L$ conditional mean estimators $\widehat{\mu}_1(\cdot), \ldots, \widehat{\mu}_L(\cdot)$ using $\mathcal{D}_1$. For each $X_i$ with $i\in\mathcal{I}_2$, we  can construct a new set of covariates $Z_i = (\widehat{\mu}_1(X_i), \ldots, \widehat{\mu}_L(X_i))^{\top}\in\mathbb{R}^L$ and create data $(Z_i, Y_i), i\in\mathcal{I}_2$. Hence considering linear combinations such as $\theta^{\top}Z_i$ with $\theta\in\mathbb{R}^L$ for prediction corresponds to linear aggregation of $\widehat{\mu}_1(\cdot), \ldots, \widehat{\mu}_L(\cdot)$. Because $\widehat{\mu}_{\ell}(\cdot)$ are thought of as estimators of the conditional mean, one often imposes the constraint $e_{\ell}^{\top}\theta \ge 0, 1\le \ell\le L$ and $\mathbf{1}^{\top}\theta = 1$ where $e_{\ell}, 1\le \ell\le L$ are the canonical basis vectors in $\mathbb{R}^L$ and $\mathbf{1}\in\mathbb{R}^L$ is a vector of all $1$'s. Hence, we can take $\Theta = \{\theta\in\mathbb{R}^L:\,e_{\ell}^{\top}\theta \ge 0, 1\le \ell\le L$ and $\mathbf{1}^{\top}\theta = 1\}$ in Algorithm~\ref{alg:general-conformal}. Our results in this section do not impose any restrictions on the covariate distribution and hence apply to the setting of linear aggregation.

Given the comparison of EFCP and VFCP in Section~\ref{sec:comparison-EFCP-VFCP}, in this section, we restrict our discussion to the performance of EFCP alone. Similar results for VFCP are in Appendix~\ref{appsec:fixed-width-VFCP-analysis}.}

% The prediction problem has been well studied in the context of linear regression, where
% prediction sets are usually constructed under linearity and Gaussian assumptions.
% % ; see~\citet[Theorem 11.3.6]{degroot2012probability}. 
% % The Gaussian assumption can be relaxed by using, for example, quantile regression \citep{koenker2001quantile}. 
% These linear-model-based methods usually have reasonable finite sample performance. However, these prediction intervals are valid only when the linear (or other parametric) regression model is correctly specified. 
% Without any assumption on the underlying distribution, \textit{conformal prediction}~\citep{vovk2005algorithmic} is a method that guarantees valid average coverage, and this will be the basis of our framework.
\subsection{Best Linear Fixed Width Prediction Set}\label{sec:general-algorithm}
In this section, we will consider finding the best linear combination of features that leads to the smallest width prediction set (i.e., non-random $\Theta \subseteq \mathbb{R}^d$). 
% Clearly, in this case, no data is used in obtaining preliminary estimators. Formally, we consider nested sets $\mathcal{F}_t^{(\theta)} := \{(x, y): y\in [\theta^{\top}x - t, \theta^{\top}x + t]\}$, for every $\theta\in\mathbb{R}^p$. 
Because the sets $\mathcal{F}_t^{(\theta)}$ are fully specified without requiring any observations, we need no splitting of the data for EFCP and one splits the data into two parts for VFCP.
% It follows that for $z = (x, y)$,
% \begin{equation}\label{eq:conformal-score-fixed-width-linear}
% \begin{split}
% t_{\theta}(z) &:= \inf_t \Bigl\{t\ge0:\, z\in\mathcal{F}_t^{(\theta)}\Bigr\} = \inf_t \Bigl\{t\ge0:\,y\in[\theta^{\top}x - t, \theta^{\top}x + t]\Bigr\}  = |y - x^{\top}\theta|.
% \end{split}
% \end{equation}
% Equation~\eqref{eq:conformal-score-fixed-width-linear} implies 
It can be easily verified that for $z = (x, y)$, $t_{\theta}(z) = |y - x^{\top}\theta|$ which implies
that the conformal prediction set corresponding to $\theta$ based on one split $\mathcal{D}_1$ of the data is given by
\begin{equation}\label{eq:prediction-set-theta}
\widehat{C}_{\theta} ~:=~ \{z = (x, y)\in\mathbb{R}^p\times\mathbb{R}:\, |y - \theta^{\top}x| \le T_{\alpha,\theta}\},
\end{equation}
where $T_{\alpha,\theta}$ is the $\lceil(1 - \alpha)(1 + |\mathcal{I}_1|)\rceil$-th largest value of $|Y_i - X_i^{\top}\theta|, i\in\mathcal{I}_1$. The $x$-cross-section width of $\widehat{C}_{\theta}$ is given by $T_{\alpha,\theta}$ and we choose the coefficient vector $\widehat{\theta}$ that minimizes $T_{\alpha,\theta}$ over all $\theta\in\Theta$, for a bounded set $\Theta\subseteq\mathbb{R}^d$. The restriction of minimization to a bounded set $\Theta$ is for technical reasons. If the joint distribution of $(X_i, Y_i)$ is very ``well-behaved'', then $\Theta$ can be taken to be $\mathbb{R}^d$.
The complete algorithm is presented in Algorithm~\ref{alg:general-conformal}. It includes both EFCP and VFCP versions.
% \paragraph{Motivation: Linear Aggregation.} Using linear predictors $\theta^{\top}x$ as in~\eqref{eq:prediction-set-theta} might seem too simplistic for prediction. But the setting considered in this section can be used for linear aggregation of possible non-parametric methods. This is similar to how linear regression can be used for aggregating several initial estimators of the conditional mean. See, for example,

\begin{algorithm}[htbp]\label{alg:general-conformal}
\SetAlgoLined
\KwIn{Data $(X_i,Y_i),i\in[N]$, coverage probability $1 - \alpha \in (0,1)$, and a bounded set $\Theta$.}
\KwResult{Optimal fixed width prediction set based on linear function of $X_{N+1}$.}
  (Splitting) Randomly split $\{1, \dots, N\}$ into two subsets $\mathcal{I}_1,\mathcal{I}_2$. Set $\mathcal{D}_1 = \{(X_i, Y_i):i\in\mathcal{I}_1\}$ and $\mathcal{D}_2 = \{(X_i, Y_i):i\in\mathcal{I}_2\}.$ For EFCP, $\mathcal{I}_2 = \emptyset$ and for VFCP, $\mathcal{I}_2$ is non-empty.
  
  (Initial Prediction)\label{me:general-prediction} Define
   $T_{\alpha,\theta}$ as $\lceil(1+|\mathcal{I}_1|)(1-\alpha)\rceil$-th largest element of $|Y_i -\theta^\top X_i|, i \in \mathcal{I}_1.$
  
  (Selection) \label{me:general-selection} Select the vector  
  $$\widehat{\theta} ~:=~ \argmin_{\theta\in\Theta} T_{\alpha,\theta} ~=~ \argmin_{\theta\in\Theta}\mbox{Width}([\theta^{\top}X - T_{\alpha,\theta}, \theta^{\top}X + T_{\alpha,\theta}]).$$
  
   (Final Prediction for VFCP)\label{me:general-final-prediction} Find
   $T_{\alpha,\widehat{\theta}}^*$, $\lceil(1+|\mathcal{I}_2|)(1-\alpha)\rceil$-th largest element of $|Y_i -X_i^{\top}\widehat\theta|, i \in \mathcal{I}_2.$
  
  \Return $\widehat{C}_{\alpha}^{\texttt{EF-lin}} :=\{(x, y):\,|y-\widehat\theta^\top x| \le T_{\alpha,\widehat{\theta}}\}$ and $\widehat{C}_{\alpha}^{\texttt{VF-lin}} := \{(x, y):\, |y-\widehat{\theta}^{\top}x| \le T_{\alpha,\widehat{\theta}}^{*}\}$.

  \caption{Best Linear Fixed Width Prediction}
\end{algorithm}

% Algorithm \ref{alg:general-conformal} is a special case of Algorithm~\ref{alg:validity-first-conformal} without the training of methods based on a part of the data. 
% The bounded set $\Theta$ used in Algorithm~\ref{alg:general-conformal} can be taken to be a bounded Euclidean ball in practice. 
% One can also restrict the minimization space depending on the context in which case the optimality and oracle inequality would be within the minimization space.
% , and by Theorem \ref{thm:coverage}, without any assumptions on the algorithm or distribution of the data, it always achieves the target coverage rate.
It readily follows from Algorithm~\ref{alg:general-conformal} (also see~\eqref{eq:empirical-oracle-inequality-EFCP}) that
\begin{equation}\label{eq:empirical-general-conformal-EFCP}
\mbox{Width}(\widehat{C}_{\alpha}^{\texttt{EF-lin}}) ~\le~ \min_{\theta\in\Theta}\,\mbox{Width}(\widehat{C}_{\theta}),
\end{equation}
which is the empirical oracle inequality for $\widehat{C}_{\alpha}^{\texttt{EF-lin}}.$
% It readily follows from Theorem~\ref{thm:coverage} that
% \[
% \mathbb{P}\left((X_{N+1}, Y_{N+1}) \in \widehat{C}_{\alpha}^{\texttt{VF-lin}}\right) \ge 1 - \alpha,
% \]
% which proves 
The validity guarantee for $\widehat{C}_{\alpha}^{\texttt{VF-lin}}$ readily follows from Theorem~\ref{thm:coverage}. Recall, once again, that the probability here is with respect to $(X_{N+1}, Y_{N+1})$ and the full data $\mathcal{D}_1\cup\mathcal{D}_2$. As mentioned in the discussion of Theorems~\ref{thm:coverage-guarantee-minimum-width} and~\ref{thm:empirical-oracle-inequality-general-VFCP}, validity of $\widehat{C}_{\alpha}^{\texttt{EF-lin}}$ and efficiency of $\widehat{C}_{\alpha}^{\texttt{VF-lin}}$ follow from an analogue of Massart's inequality. With Massart's inequality, we got $\log(K)$ dependence, but the fact that linear spaces are VC allows us to use empirical process techniques to prove validity and efficiency for $\widehat{C}_{\alpha}^{\texttt{EF-lin}}$ and $\widehat{C}_{\alpha}^{\texttt{VF-lin}}$, respectively.
The validity guarantee for $\widehat{C}_{\alpha}^{\texttt{EF-lin}}$ does not require any more assumptions than just independent and identically distributed observations. The efficiency gurantee of $\widehat{C}_{\alpha}^{\texttt{VF-lin}}$ requires the analogue of assumption~\ref{eq:Holder-VFCP}. Note that in the current setting $\mbox{Width}(\mathcal{F}_t^{(\theta)}) = 2t$ which is a Lipschitz function in $t$ and hence satisfies~\ref{eq:width-continuity} with $L_W = 2$. {See Appendix~\ref{appsec:fixed-width-VFCP-analysis} for details on the efficiency guarantee of $\widehat{C}_{\alpha}^{\texttt{VF-lin}}$.}

{Let us first consider the validity guarantee of $\widehat{C}_{\alpha}^{\texttt{EF-lin}}$. We define some notation that will be used in this section.
For every $\theta\in\mathbb{R}^d$ and $t\ge0$, define
\begin{equation}\label{F}
 F_{\theta}(t):= \mathbb{P} \bigl( |Y- X^\top \theta | \leq t \bigr).
\end{equation}
If $\theta$ is random, then $F_\theta (t)$ is a random variable, i.e., the probability in the definition of $F_{\theta}(t)$ is not taken with respect to $\theta$, but just with respect to $(X, Y)$.
\begin{prop}\label{prop:validity-EFCP-fixed-width}
For any $\mathcal{I}\subseteq[N] = \{1, 2, \ldots, N\}$, set
\begin{equation}\label{eq:general-distribution-function-assumptions}
R (\mathcal{I}) ~:=~ \sup_{t\ge 0}\sup_{\theta \in \mathbb{R}^d}\, \biggl|\frac{1}{ |\mathcal{I}| } \sum_{i \in \mathcal{I}} \mathbbm{1}{\{ |Y_i - \theta^\top X_i | \leq t \} } -F_{\theta}(t) \biggr|.
\end{equation}
Then $\mathbb{P}\left((X_{N+1}, Y_{N+1})\in\widehat{C}_{\alpha}^{\texttt{EF-lin}}\big|\mathcal{D}_1\right) \ge 1 - \alpha - R([N]).$ Moreover, there exist universal constants $\mathfrak{C}_1, \mathfrak{C}_2 > 0$ such that for independent and identically distributed random vectors $(X_i, Y_i), 1\le i\le N$,
\begin{equation}\label{eq:VC-class-bound}
\mathbb{P}\left(R([N]) \le \mathfrak{C}_1\sqrt{\frac{d + \log(1/\delta)}{N}}\right) \ge 1 - \delta\quad\mbox{and}\quad \mathbb{E}[R([N])] \le \mathfrak{C}_1\sqrt{\frac{d}{N}}.
\end{equation}
\end{prop}
See Appendix~\ref{appsec:proof-of-prop-validity-EFCP-fixed-width} for a proof.
Combining~\eqref{eq:empirical-general-conformal-EFCP} and Proposition~\ref{prop:validity-EFCP-fixed-width}, we obtain empirical oracle inequality and validity guarantee for EFCP. To prove an asymptotic oracle inequality for EFCP, we need an assumption similar to~\ref{eq:Holder-VFCP}.}
% We now consider an assumption on the joint distribution of $(X, Y)$ similar to~\ref{eq:Holder-VFCP} in order to prove an oracle inequality for $\widehat{C}_{\alpha}^{\texttt{VF-lin}}$ and validity for $\widehat{C}_{\alpha}^{\texttt{EF-lin}}$.  Further, 
Define the quantile function $F_\theta^{-1}:[0,1] \rightarrow \mathbb{R}$ as
% \begin{equation}\label{def:quantile-function}
$F_\theta^{-1}(p):=\inf \{t \in \mathbb{R}: F_\theta(t) \geq p\}.$
% \end{equation} 
% Set
With this notation, the set
\begin{equation}\label{eq:def-quantile-theta}
% Clearly, $Q_{\alpha,\theta}$ depends also on $\alpha$, but we suppress that dependence for brevity.
% Definition~\eqref{eq:def-quantile-theta} implies that $\mathbb{P}(Y \in [X^{\top}\theta - Q_{\alpha,\theta}, X^{\top}\theta + Q_{\alpha,\theta}]) \ge 1 - \alpha$ for all $\theta\in\mathbb{R}^d$. 
% Hence, 
% $$
C_{\alpha,\theta}^{\texttt{orc-lin}} := \{(x, y):\,y\in[\theta^{\top}x - Q_{\alpha,\theta}, \theta^{\top}x + Q_{\alpha,\theta}]\},\quad\mbox{with }Q_{\alpha,\theta} ~:=~ F_{\theta}^{-1}(1 - \alpha),
\end{equation}
% $$
can be thought of as an asymptotic oracle of the conformal prediction set $\widehat{C}_{\theta}$ defined in~\eqref{eq:prediction-set-theta}. % step 3 of Algorithm~\ref{alg:general-conformal}.
Therefore, an asymptotic oracle inequality for EFCP sets in Algorithm~\ref{alg:general-conformal} would be
\begin{equation}\label{eq:asymptotic-oracle-general-conformal}
    \mbox{Width}(\widehat{C}_{\alpha}^{\texttt{EF-lin}}) ~=~ 2 T_{\alpha,\widehat{\theta}} ~\le~ 2 \min_{\theta\in\Theta} Q_{\alpha,\theta} + R_N,
\end{equation}
% and the empirical oracle inequality for EFCP sets in Algorithm~\ref{alg:general-conformal} would be
% \begin{equation}\label{eq:empirical-oracle-general-conformal}
% \mbox{Width}(\widehat{C}_{\alpha}^{\texttt{EF-lin}}) ~=~ 2 T_{\alpha,\widehat{\theta}} ~\le~ 2 \min_{\theta\in\Theta}T_{\alpha,\theta} + R_N',
% \end{equation}
for some rate $R_N$. 
% Recall $Q_{\alpha,\theta}$ defined in~\eqref{eq:def-quantile-theta}. 
% We already have~\eqref{eq:empirical-oracle-general-conformal} for $\widehat{C}_{\alpha}^{\texttt{EF-lin}}$ from~\eqref{eq:empirical-general-conformal-EFCP}. 
Now, consider the following analog of assumption~\ref{eq:Holder-VFCP}. 
{Let $\theta^*\in\Theta$ be any vector such that $Q_{\alpha,\theta^*} \le \min_{\theta\in\Theta} Q_{\alpha,\theta} + 1/N.$} 
\begin{enumerate}[label=\bf(A0)]
\item \label{assump:general-holder-EFCP} There exists $r^*, \gamma \in (0, 1]$, such that $F_{\theta^*}^{-1}(\cdot)$ is $\gamma$-H{\"o}lder continuous on $[1-\alpha-r^*,1-\alpha+r^*]$ with  H\"older continuity constant $L_{\theta^*}$, i.e., for all $\theta\in\Theta$, and all $q_1, q_2\in [1-\alpha-r^*,1-\alpha+r^*]$,
$
| F_{\theta^*}^{-1}(q_1) - F_{\theta^*}^{-1}(q_2) | \le L_{\theta^*}|q_1 - q_2|^{\gamma}.
% , \mbox{ for all } q_1, q_2 .
$
\end{enumerate}
% Assumption~\ref{assump:general-holder-EFCP} is a relaxation of the condition that the density of $|Y-X^{\top}\theta^*|$ is bounded away from zero in the neighborhood of $Q_{\alpha,\theta^*}$ and is needed to ensure that $T_{\alpha,\theta^*}$ is close to $Q_{\alpha,\theta^*}$. 
Note that Assumption~\ref{assump:general-holder-EFCP} depends on $\alpha\in(0, 1)$. If the density of $Y - X^{\top}\theta^*$ is bounded away from zero on $[F_\theta^{-1}\bigl( 1-\alpha - r^* \bigr),F_\theta^{-1}\bigl( 1-\alpha + r^* \bigr) ]$, then~\ref{assump:general-holder-EFCP} will hold true with $\gamma = 1$; see Proposition~\ref{prop:simple-sufficient-quantile-holder}. If $(X, Y)$ has a joint normal distribution, then Assumption~\ref{assump:general-holder-EFCP} holds true with $\gamma = 1$ and a constant $L_{\theta^*}$ satisfying $L_{\theta} \le C(1 + \|\theta\|_2)$ for some constant $C$ (Proposition \ref{prop:joint-normal}). {A simple sufficient condition for~\ref{assump:general-holder-EFCP} when $X$ is compactly supported is that the conditional distribution of $Y$ given $X$ is bounded away from zero on bounded sets; see Proposition~\ref{prop:simple-sufficient-quantile-holder-conditional}.
Under~\ref{assump:general-holder-EFCP}, we have the following asymptotic oracle inequality for EFCP.
\begin{thm}\label{thm:general-oracle-final}
Fix any $\delta\in(0, 1)$. If Assumption~\ref{assump:general-holder-EFCP} holds true with $r^* \ge \sqrt{\log(2/\delta)/(2N)}$, then with probability $1 - \delta$,
\begin{equation}\label{eq:general-conformal-oracle-EFCP}
% T_{\alpha,\widehat{\theta}}^* 
\mathrm{Width}(\widehat{C}_{\alpha}^{\texttt{EF-lin}})
~\leq~ \min_{\theta \in \Theta }\,\mathrm{Width}(C_{\alpha,\theta}^{\texttt{orc-lin}}) ~+~  \frac{2}{N} +  2L_{\theta^*}\left(\frac{\log(2/\delta)}{2N}\right)^{\gamma/2}.
\end{equation}
% and
% \begin{equation}\label{eq:coverage-EF-lin}
% \mathbb{P}\left((X_{N+1}, Y_{N+1}) \in \widehat{C}_{\alpha}^{\texttt{EF-lin}}\,\bigg|\,\mathcal{D}_1\right) ~\ge~ 1 - \alpha - R([N]).
% \end{equation}
\end{thm}
}
We defer the proof of Theorem~\ref{thm:general-oracle-final} to Appendix \ref{sec:Proof of Theorem general-oracle-final} of the supplementary. {It is interesting to note that the slack of the asymptotic oracle inequality does not depend on the ``size'' of $\Theta$.}

Algorithm~\ref{alg:general-conformal} does not require any estimation and only uses two splits of the data (for VFCP). These advantages are countered with the disadvantages of choosing $\Theta$ and minimization over a multidimensional space $\Theta\subseteq\mathbb{R}^d$. For this reason, in practice, one may consider the version of Algorithm~\ref{alg:validity-first-conformal} when the training methods correspond to ridge regressions with different penalty parameters. This reduces $\Theta\subseteq\mathbb{R}^d$ to $\{\widehat{\beta}_{\lambda}:\,\lambda \ge C\}$ for some constant $C\in\mathbb{R}$, where $\widehat{\beta}_{\lambda}$ is the ridge regression estimator with penalty parameter $\lambda$ computed on an independent dataset. {One can prove asymptotic oracle inequalities for this setting as well, but these results require strong assumptions such as bounded moments and non-singularity of the covariance matrix of covariates. Although these assumptions are much weaker in the linear model theory, they are nevertheless stronger in the application to aggregation. Hence, we present the ridge regression results in Appendix~\ref{appsec:Best-ridge-fixed-width} and consider the application to aggregation here.}
%%%%%%%%%%%%%%%%%%%%%%%%%%%%%%%%%%%%%%%%%%%%%%%%%%%%
%%%%%%%%%%%%%%%%%%%%%%%%%%%%%%%%%%%%%%%%%%%%%%%%%%%%
\subsection{Aggregation and Conditional Coverage}\label{subsec:aggregation-conditional-coverage}
{
In this section, we consider the application of results in Section~\ref{sec:general-algorithm} to aggregation and in this context, we also show (under some regularity conditions) that $\widehat{C}_{\alpha}^{\texttt{EF-lin}}$ is asymptotically conditionally valid. As described in the beginning of Section~\ref{sec:linear-prediction-fixed-width}, let $\mathcal{D}_1, \mathcal{D}_2$ be two random parts of the data $\mathcal{D} = \{(X_i, Y_i), 1\le i\le N\}$; let the corresponding indices to $\mathcal{I}_1, \mathcal{I}_2$. Let $\widehat{\mu}_1(\cdot), \ldots, \widehat{\mu}_L(\cdot)$ be estimators of the conditional mean $\mu^*(x) = \mathbb{E}[Y|X = x]$ based on $\mathcal{D}_1$. With $e_{\ell}, 1\le \ell \le L$ representing the canonical basis of $\mathbb{R}^L$, take
\begin{equation}\label{eq:simplex-definition}
\Theta = \{\theta\in\mathbb{R}^L:\,e_{\ell}^{\top}\theta \ge 0, 1\le \ell \le L, \mathbf{1}^{\top}\theta = 1\},
\end{equation}
and set $\widehat{C}_{\alpha}^{\texttt{agg}}$ to be EFCP prediction set returned by Algorithm~\ref{alg:general-conformal} with $\{(Z_i, Y_i), i\in\mathcal{I}_2\}$:
\[
\widehat{C}_{\alpha}^{\texttt{agg}} := \left\{(x, y):\, \left|y - \widehat{\mu}_{\widehat{\theta}}(x)\right| \le T_{\alpha,\widehat{\theta}}\right\},\quad\mbox{with}\quad \widehat{\mu}_{\widehat{\theta}}(x) := \sum_{\ell = 1}^L \widehat{\theta}_{\ell}\widehat{\mu}_{\ell}(x),
\]
where $\widehat{\theta} = (\widehat{\theta}_1, \ldots, \widehat{\theta}_L)$ is the width minimizing vector in $\Theta$.
These types of fixed width prediction sets are optimal, if the true data-generating process for $(X, Y)$ is $Y = \mu_0(X) + \xi$ where $\xi$ is a mean zero symmetric random variable that is independent of $X$~\citep{lei2018distribution}. Proposition~\ref{prop:validity-EFCP-fixed-width} already implies asymptotic validity of $\widehat{C}_{\alpha}^{\texttt{agg}}$ if $L = o(N)$; this requirement could be improved to $\log(L) = o(n)$ because $\Theta$ is a simplex and much smaller than $\mathbb{R}^{L}$. We do not pursue such a generalization here.

For any function $g:\mathcal{X}\to\mathbb{R}$, define $\|g\|_{\infty} := \sup_{x\in\mathcal{X}}|g(x)|$ and $\|g\|_2^2 := \int g^2(x)dP_X(x)$. (Here $P_X$ is the probability measure of covariates on $\mathcal{X}$). To study the conditional coverage and efficiency properties, consider the following assumptions:
\begin{enumerate}[label=\bf (A\arabic*)]
    \item $(X_i, Y_i), 1\le i\le N$ are independent and identically distributed random vectors from the model $Y = \mu_0(X) + \xi,$ where $\xi$ is a symmetric unimodal distribution with mean zero.\label{assump:symmetric-error-dist}
    \item There exists $r, \kappa > 0$ such that the Lebesgue density $f$ of $\xi$ is lower bounded by $r > 0$ on $[q_{\alpha} - \kappa, q_{\alpha} + \kappa]$, where $q_{\alpha}$ is the $(1-\alpha)$-th quantile of $|\xi|$. Moreover, $f$ is uniformly bounded by $M \in (0, \infty)$ and is continuously differentiable with the derivative bounded by $M.$\label{assump:lower-bound-differentiable}
    \item There exists $r_{n,2}, r_{n,\infty}, \eta_{n,2}, \eta_{n,\infty} \ge 0$ such that\label{assump:rates-of-convergence}
    \begin{equation}\label{eq:assump-conditional-coverage}
    \mathbb{P}\left(\max_{1\le \ell \le L}\|\widehat{\mu}_{\ell} - \mu_0\|_{\infty} \ge r_{n,\infty}\right) \le \eta_{n,\infty}\quad\mbox{and}\quad \mathbb{P}\left(\inf_{\theta\in\Theta}\|\widehat{\mu}_{\theta} - \mu_0\|_2^2 \ge r_{n,2}^2\right) \le \eta_{n,2},
    \end{equation}
\end{enumerate}
Assumptions~\ref{assump:symmetric-error-dist} and~\ref{assump:lower-bound-differentiable} are taken from~\cite{lei2018distribution}. They are restrictive but serve as suitable assumptions to study the fixed width prediction sets based on non-parametric regression estimators. Under these assumptions,~\cite{lei2018distribution} shows that $C_{\alpha}^{\texttt{orac}} = \{(x, y):\,|y - \mu_0(x)| \le q_{\alpha}\}$ is the optimal (i.e., the smallest width) prediction set with conditional coverage. Assumption~\ref{assump:rates-of-convergence} is an extension of the assumptions A2 and A4 of~\cite{lei2018distribution}. There are numerous estimators that satisfy these conditions with a minimax rate of convergence. See, for example,~\cite{gyorfi2002distribution,brown2008robust}. Note that under~\ref{assump:symmetric-error-dist}, $\mu_0(\cdot)$ can be regarded as a conditional location measure such as conditional median, conditional $M$-estimator which allows for the construction of estimators with sub-Gaussian tail behavior. Hence, the requirement of uniform convergence of $1\le \ell \le L$ in~\ref{assump:rates-of-convergence} is not a strong requirement. 
\begin{thm}\label{thm:width-efficiency-conditional-coverage}
    Fix any $\delta\in(0, 1)$. Under assumptions~\ref{assump:symmetric-error-dist} and~\ref{assump:lower-bound-differentiable}, with probability at least $1 - \delta$, 
    \begin{equation}\label{eq:oracle-width-efficiency-aggregation}
    \mathrm{Width}(\widehat{C}_{\alpha}^{\texttt{agg}}) \le \mathrm{Width}(C_{\alpha}^{\texttt{orac}}) + \frac{M}{4r}\inf_{\theta\in\Theta}\|\widehat{\mu}_{\theta} - \mu_0\|_2^2 + \sqrt{\frac{\log(2/\delta)}{2|\mathcal{I}_2|}} + \frac{2}{|\mathcal{I}_2|},
    \end{equation}
    whenever $\kappa > (M/4r)\inf_{\theta\in\Theta}\|\widehat{\mu}_{\theta} - \mu_0\|_2^2 + \sqrt{\log(2/\delta)/(2|\mathcal{I}_2|)} + 2/|\mathcal{I}_2|$.

    Furthermore,
    conditional on $\mathcal{D}_1$, with probability at least $1 - \delta$, for all $x\in\mathcal{X}$,
    \begin{align*}
    \mathbb{P}\left((X_{N+1}, Y_{N+1})\in\widehat{C}_{\alpha}^{\texttt{agg}}\bigg|X_{N+1} = x, \mathcal{D}_1\right) &\ge 1 - \alpha- \sqrt{\frac{\log(2/\delta)}{2|\mathcal{I}_1|}} - \frac{(1 + M/2)}{|\mathcal{I}_2|}\\ 
    &\quad- 2M\max_{1\le \ell \le L}\|\widehat{\mu}_{\ell} - \mu_0\|_{\infty} - (M/2)\inf_{\theta\in\Theta}\|\widehat{\mu}_{\theta} - \theta^*\|_2^2 .
    \end{align*}
    Under assumption~\ref{assump:rates-of-convergence}, there exists a universal constant $\mathfrak{C}$,
    \[
    \mathbb{P}\left((X_{N+1}, Y_{N+1})\in\widehat{C}_{\alpha}^{\texttt{agg}}\bigg|X_{N+1} = x\right) \ge 1 - \alpha - (2Mr_{n,\infty} + Mr_{n,2}^2/2) - \frac{\mathfrak{C} + M/2}{|\mathcal{I}_2|^{1/2}} - \eta_{n,2} - \eta_{2,\infty}.
    \]
\end{thm}
Inequality~\eqref{eq:oracle-width-efficiency-aggregation} proves that the width of the aggregation prediction set asymptotically matches that of the optimal prediction set. Note that, for this result, we only need the existence of a linear combination estimator that is consistent and allows for some inconsistent estimators in the collection $\{\widehat{\mu}_{\ell}:1\le \ell \le L\}$. The last inequality of Theorem~\ref{thm:width-efficiency-conditional-coverage} proves that consistency of all non-parametric estimators implies asymptotic conditional coverage validity of the aggregation prediction set.   
}
\subsection{Simulations: EFCP and VFCP with Ridge Regression}\label{sec:simulations}
We now present some simulations to compare the performance of EFCP and VFCP when the training methods corresponds to ridge regression with different penalty parameters. Because ridge regression is a linear prediction methodology, we also consider comparison with two other linear predictors based prediction regions.
% Before describing these, we 
Define $X=(X_1^\top, \dots, X_N^\top)^\top$ and $Y=(Y_1, \dots, Y_N)^\top$.
\begin{enumerate}
    \item \textbf{Linear}: Under the classical linear model assumptions, wherein the response vector $Y$ conditional on the covariate matrix $X$ is generated from a multivariate Gaussian with mean zero and a scaled identity covariance $\sigma^2I_n$, one can construct a finite sample valid prediction region using normality:
    % Motivated by Bayes prediction interval \citep{vovk2014efficiency}, where it is assumed that given $X_1, \dots, X_N$, the response $Y_1, \dots, Y_N$ are generated by a linear model $Y_i=X_i^\top \beta + \epsilon_i,$ where each $\epsilon_i$ is normally distributed as $\Normal(0,\sigma^2)$ and $\epsilon_1, \dots, \epsilon_N$ are independent, then the prediction interval is
    $
    \widehat{C}_{ \alpha}^{\text {linear }}:= \{(x, y):\,|y - x^\top (X^\top X)^{+} X^\top Y| \le \widehat{\sigma}z_{\alpha/2}\sqrt{1+g_N(x)} \},
    $
    where $g_N(x)=x (X^\top X)^{+} x$, $z_{\alpha/2}$ is the $(1-\alpha/2)$-quantile of the standard normal distribution $\Normal(0,1)$ and $\widehat{\sigma}^2=\| Y-X^\top (X^\top X)^{+} X^\top Y \|^2_2/N$. 
    % This prediction interval is valid only under linearity and homoscedastic Gaussian error assumption. This interval need not have approximate validity when the errors are homoscedastic but not Gaussian. 
    See~\cite{vovk2014efficiency} for a discussion and comparison with the full conformal prediction interval.
    \item \textbf{Naive}: For any point prediction method, one can also consider a naive prediction interval using the residuals on the full data.
    % Note that, in comparison, the split conformal method uses residuals on a ``calibration'' data while the point prediction method is trained on the training data. 
    Let $T_\alpha$ denote the $\lceil(1-\alpha)(N+1)\rceil$-th largest value of $\left|Y_{i}-X_i^\top (X^\top X)^{+} X^\top Y \right|, i = 1,\ldots, N$.
    With this notation, the ``naive" conformal prediction interval as defined in~\cite{barber2019predictive} is
    $
    \widehat{C}_{ \alpha}^{\text {naive }}:=\{(x, y):\, |y - x^\top (X^\top X)^{+} X^\top Y| \le T_\alpha\}.
    $
    If the point prediction algorithm is consistent to the true conditional expectation and the true distribution of the errors is homoscedastic as well as symmetric around zero, then the naive prediction interval has an \emph{asymptotic} valid coverage guarantee.
%     \item Conformal: Split tuning-free conformalized ridge regression described in Algorithm \ref{alg:conformal-ridge}.
%     \item CV$^*$: Instead of minimizing the quantiles in the Selection step \ref{me:selection} of the algorithm, this method chooses $\widehat{\lambda}_\alpha$ to minimize the test error: $\widehat{\lambda}_\alpha:=\argmin_{\lambda} \sum_{i \in \mathcal{I}_2} R_i(\lambda)/|\mathcal{I}_2|$. Everything else is the same as Algorithm \ref{alg:conformal-ridge}.
%      \item CV-5-fold: Split conformal prediction \citep{lei2018distribution} with tuning parameter from 5-fold cross validation, i.e. randomly split $\{1, \dots, N \}$ into two equal-sized subsets $\mathcal{I}_1,\mathcal{I}_2$, then
% \begin{equation}\label{eq:cv}
% \widehat{C}_{N, \alpha}^{l-\text {fold }}\left(X_{N+1}\right):=X_{N+1}^\top (X^\top X)^{+} X^\top Y \pm \widehat{q}_{|\mathcal{I}_2|, \alpha}^{+}\left\{\left|Y_{i}-\widehat{\mu_l}\left(X_{i}\right)\right|, i \in \mathcal{I}_2\right\},
%   \end{equation}
% with $\lambda_l$ from $l-$fold cross validation on $\mathcal{I}_1$ and $l=5$.
% \item CV-10-fold: Split conformal prediction \citep{lei2018distribution} with tuning parameter from 10-fold cross validation, i.e. plug in $l=10$ in \eqref{eq:cv}.
%     \item CV-loo: Split conformal prediction \citep{lei2018distribution} with tuning parameter from leave-one-out cross validation, i.e. plug in $l=N/2$ in \eqref{eq:cv}.
\end{enumerate}

In our simulations below, the target coverage level is $1-\alpha=0.9.$ To demonstrate the validity and efficiency, we consider both linear and nonlinear models and repeat the experiment at each dimension $d=10, \ldots, 300$ with i.i.d. data points. The result is averaged over 100 trials with each trial having an independent draw of a sample of size $N=200$ and the test sample of size 100. 
%%% Removed for arxiv
% In Appendix~\ref{appsec:more-simulations}, we provide comparisons to other methods and in more data generating models. 
In the following, we will 
% only 
show the performance in two models: one linear and the other non-linear. In both these models, we consider heavy-tailed data.
In the implementation of EFCP and VFCP with ridge regression, we minimized over a grid of 100 $\lambda$ values from $0$ to $200$ and did not attempt to obtain the global minimum as in Algorithm~\ref{alg:conformal-ridge}.
% Two scenarios are experimented for each study.

\paragraph{Linear Model:} Following assumption \ref{assump:finite-X}, each row of  $X$ is generated independently following a multivariate student distribution $t_d(\nu, \Sigma)$ with $d \times d$ covariance matrix $\Sigma$ satisfying $\Sigma_{ij}=\rho$ if $i \neq j$ (equi correlation) and $\Sigma_{jj}=1$.
$Y$ is generated from 
\begin{equation}\label{eq:linear_fm}
 Y=X \beta + \xi, \text{ where }\beta_j=1+(j-1)\mbox{ mod }5, 1 \leq j \leq d,
\end{equation}
and the coordinates of $\xi$ are independently generated from $t(\nu) \times (1+ \sqrt{ X_1^2+X_2^2})$, conditional on $X$. Note that irrespective of the dimension, the conditional variance of $\xi$ depends only on the first two covariates. The degrees of freedom is $\nu=3$ in the first experiment and $\nu=5$ for the second one. Note that~\eqref{eq:linear_fm} is a well-specified model but with heteroscedastic noise. Because of heteroscedasticity, $\widehat{C}_{\alpha}^{\mathrm{linear}}$ does not have a finite sample valid coverage guarantee.
 
% \paragraph{Study 2}Linear: Following assumption \ref{assump:finite-X-indep}, we first generate a matrix $Z$ that has the same size as $X$, where each entry of $Z$ follows a student distribution $t(\nu)$ with $\nu$ being the degrees of freedom. Then $X$ is obtained by $X=Z \Sigma$ with the $d \times d$ matrix $\Sigma$ satisfying $\Sigma_{ij}=\rho$ if $i \neq j$ (equi correlation) and $\Sigma_{jj}=1$. $Y$ is generated from
% \begin{equation}\label{eq:linear_structure}
%  Y=X \beta + \epsilon, \text{ where }\beta_j=1+(j-1)mod 5, 1 \leq j \leq d,
%   \end{equation}
%   and each noise component $\epsilon_i$ follows $t(\nu) \times (1+ \sqrt{ X_1^2+X_2^2})$ independently. The degrees of freedom $\nu=3$ in the first experiment and $\nu=5$ for the second one.

\paragraph{Nonlinear Model:} Following assumption \ref{assump:finite-X}, each row of  $X$ is generated independently following a multivariate student distribution $t_d(\nu, \Sigma)$ with $d \times d$ covariance matrix $\Sigma$ satisfying $\Sigma_{ij}=0.5$ if $i \neq j$ (equi correlation) and $\Sigma_{jj}=1$.
$Y$ is generated from 
\begin{equation}\label{eq:nonlinear_pois_fm2}
Y \sim \mathrm{Pois}\left(\sin ^{2}(X_1)+\cos^4(X_2)+0.01\right)+0.03 X_1 \epsilon_{1}+25 \mathbbm{1}_{\{u<0.01\}} \epsilon_{2},
\end{equation}
and each noise component of $\epsilon_k, k=1,2$, follows $t(\nu) \times (1+ \sqrt{ X_1^{2k}+X_2^{2k}})$ independently. Note that irrespective of the dimension $d$, the conditional distribution of the response only depends on the first two covariates. Similar to the linear case, we consider the degrees of freedom of $\nu=3$ in the first experiment and $\nu=5$ in the second one. 

The results for both linear and non-linear models with $\nu = 3, 5$ are presented in Figure~\ref{fig:combined_fm_rep_10_ggplot}. The columns labelled ``3rd moment'' represent the cases with $\nu = 3$ and the columns labelled ``5th moment'' represent the cases with $\nu = 5$. ``Linear'' and ``Non Linear'' represent respectively the linear~\eqref{eq:linear_fm} and non-linear~\eqref{eq:nonlinear_pois_fm2} models stated above.
% \paragraph{Study 4}Nonlinear: Following assumption \ref{assump:finite-X-indep}, we first generate a matrix $Z$ that has the same size as $X$, where each entry of $Z$ follows a student distribution $t(\nu)$ with $\nu$ being the degrees of freedom. Then $X$ is obtained by $X=Z \Sigma$ with the $d \times d$ matrix $\Sigma$ satisfying $\Sigma_{ij}=0.5$ if $i \neq j$ (equi correlation) and $\Sigma_{jj}=1$. $Y$ is generated from 
% \begin{equation}\label{eq:nonlinear_pois_structure}
%  Y \sim \mathrm{Pois}\left(\sin ^{2}(X_1)+\cos^4(X_2)+0.01\right)+0.03 X_1 \epsilon_{1}+25 \mathbbm{1}_{\{u<0.01\}} \epsilon_{2},
%   \end{equation}
%  and each noise component of $\epsilon_k, k=1,2$, follows $t(\nu) \times (1+ \sqrt{ X_1^{2k}+X_2^{2k}})$ independently. The degrees of freedom $\nu=3$ in the first experiment and $\nu=5$ for the second one.

% Figures  \ref{fig:linear_fm_rep_10_ggplot} and \ref{fig:linear_structure_rep_10_ggplot} show the results for the linear model in Study 1 and Study 2. Table \ref{tab:time} shows the computation efficiency for different methods for Study 1, where we run the simulation on a 2 GHz Quad-Core Intel Core i5 laptop.
\begin{figure}[htbp]
\centering
\includegraphics[width=1\linewidth,height=4in]{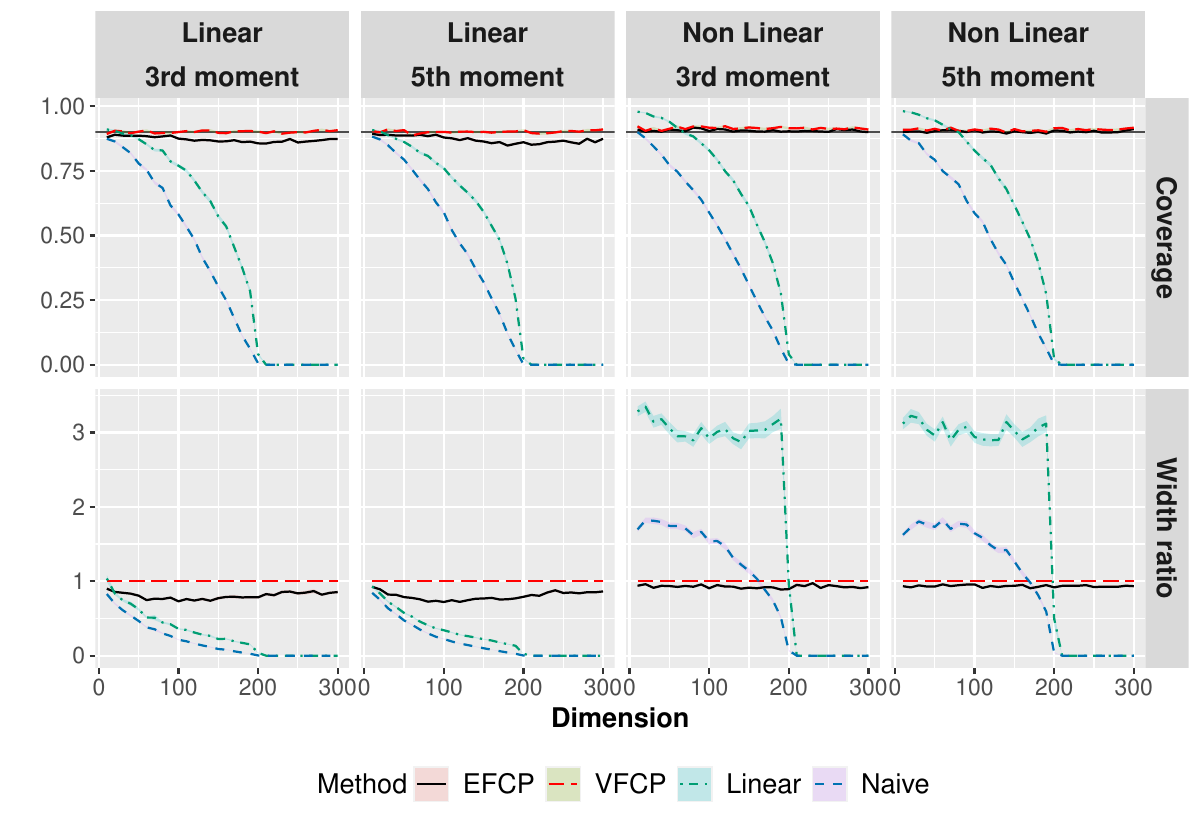}
\caption{Simulation results, showing the mean coverage and mean width ratio of EFCP, VFCP, Naive, and Linear over 100 independent trials. The $x$-axis shows the dimension of the model changing from $10$ to $300$ while the total sample size is fixed at 200. The second row of the plot ``Width ratio'' shows the ratio of the width of each method prediction interval to that of the VFCP's prediction interval. The standard error of coverage and width ratio is shown as a colored band around the mean coverage and mean width ratio. Given the large number of replications, the standard error is too small for most methods except for Linear. In all the models, the covariates and errors are heavy-tailed. EFCP and VFCP methods attain validity while Naive and Linear undercover the true response as expected.}
\label{fig:combined_fm_rep_10_ggplot}
\end{figure}
Below, we note some observations from our simulations:
\begin{itemize}
    \item VFCP has valid coverage (near perfect at $0.9$), while EFCP slightly undercovers attaining at worst coverage of $0.88$ across all dimensions. For smaller dimensions $(\le 100)$, the coverage of EFCP is also near perfect. Neither the Linear nor the Naive methods give valid coverage, and this is expected because the two settings are both heteroskedastic. Further, the coverage for these methods drop significantly as dimension increases.  
    \item When $d \geq N$, we note that the Linear method and the Naive method prediction will always have zero length intervals since all $N$ residuals will be equal to zero. Consequently, in such a situation, these two methods will have coverage equal to zero.
    \item Compared with VFCP, EFCP has an improved efficiency resulting in 10--30\% smaller widths across the dimensions. On average, EFCP improves the width of VFCP by 20\%. 
    \item Before computing the ratio of the widths, we find that the standard error of the widths of EFCP is, on average, 30\% smaller than that of VFCP over 100 repetitions.
\end{itemize}

\section{Summary and Future Directions}\label{sec:conclusions}
In this paper, we have proposed two selection algorithms for conformal prediction regions to obtain the smallest prediction set in practice; these are called ``efficiency first'' and ``validity first'' conformal prediction algorithms. We have studied the width and coverage properties of both these algorithms. The efficiency first method attains finite sample efficiency but only has an approximate coverage validity guarantee. The validity first method attains finite sample coverage validity but only attains approximate efficiency. Both methods are asymptotically efficient and asymptotically valid. In both these methods, we have made use of the split conformal method for the construction of prediction regions. The efficiency, as well as approximate validity of the efficiency first conformal method, applies to any machine learning technique; it is as general as the split conformal method itself. {Comparison of both methods shows that the ``efficiency first'' method is the preferable one in practice. We also considered aggregation of non-parametric regression methods for the smallest width prediction set using absolute residual as the conformal score.} 
% In the context of linear predictors and ridge regression, we show that our methods yield confidence regions that asymptotically match the smallest prediction regions based on the unknown data distribution.
% In this paper, we have introduced a method that chooses the smallest prediction set by tuning the hyper-parameter while retaining valid coverage for predictive inference. In particular, we studied the case when conformal prediction acts as a wrapper around ridge regression where we allow for negative penalty and proved that in terms of width, our method converges to the smallest prediction set at a very fast rate. All the results are derived without requiring the underlying model to be linear and under minimal assumptions on the data generating distribution, we allow for the dimension $d$ to grow with the sample size $N$ in such a way that $d=o(N)$. %Our results hold under very weak moment conditions on the covariates $X$ and response variable $Y$.

Throughout the paper, we have focused on minimizing the width of the prediction sets. The quantity $\mbox{Width}(\cdot)$ in our algorithms can be replaced by \emph{any} other property of the prediction sets. This could be useful when the domain $\mathcal{Z}$ of the observations is more complicated, e.g., the space of images or the space of functions~\citep{lei2015conformal,bates2021distribution}. 
% Some practical examples can be found in~\cite{bates2021distribution}, where one adaptation of our method to the multi-label classification example in Section 5.2 there would be by replacing the width of the prediction set by the proportion of false positives.% In the multi-label classification example of~\citet[Section 5.2]{bates2021distribution}, one can  replace width of the prediction set by the proportion of false positives.

% Moving forward, several interesting questions remain in the direction of our work. Firstly, we have developed EFCP and VFCP building on the split conformal method. One can ask whether such a program can be built on other conformal procedures such as the full conformal method, the jackknife+ or the CV+ (of~\cite{barber2019predictive}). The main reason for why split conformal method works lies in its conditional coverage guarantees (conditional on the training data). Such training data conditional coverage guarantees have been recently investigated in~\cite{bian2022training} and the authors have shown that the full conformal as well as the jackknife+ methods do not necessarily satisfy the conditional coverage guarantees, but the CV+ method satisfies. In principle, hence, one can build EFCP and VFCP procedures with the ingredient prediction sets coming from the CV+ procedure. 

% The proof of Theorem~\ref{thm:coverage-guarantee-minimum-width} 
Theorem~\ref{thm:multiple-splits-conformal}
implies that one can repeat the split conformal method several times on the data and choose the smallest prediction set to obtain an approximately valid prediction region. It would be interesting to compare the width of such a prediction region with other conformal methods such as jackknife+ or cross-conformal~\citep{barber2019predictive}. Note that by repeating the split conformal procedure multiple times, one uses the full data for training as well as prediction. 

\textit{Conflict of interest:} The authors report there are no competing interests to declare.

\bibliographystyle{plainnat}
\bibliography{ref}

\newpage
% \appendix
% \newpage
\setcounter{section}{0}
\setcounter{equation}{0}
\setcounter{figure}{0}
\renewcommand{\thesection}{A.\arabic{section}}
\renewcommand{\theequation}{E.\arabic{equation}}
\renewcommand{\thefigure}{A.\arabic{figure}}
% \tableofcontents
% \titlelabel{\thetitle: }
% \cftsetindents{section}{1em}{2.5em}
% \cftsetindents{subsection}{1.5em}{3em}
\setcounter{page}{1}
  \begin{center}
  \Large {\bf Supplement to ``Selection and Aggregation of Conformal Prediction Sets''}
  \end{center}
       
\begin{abstract}
This supplement contains the proofs to all the main results in the paper and some supporting lemmas. 
\end{abstract}

\section{Analysis of VFCP for Fixed Width Prediction Sets}\label{appsec:fixed-width-VFCP-analysis}
In this section, we consider oracle inequalities for the width of VFCP sets for the best linear fixed width nested sequence. 
Recall the asymptotic oracle prediction sets $C_{\alpha, \theta}^{\texttt{orc-lin}}$ from~\eqref{eq:def-quantile-theta} and that an asymptotic oracle inequality for VFCP sets in Algorithm~\ref{alg:general-conformal} would be
\begin{equation}\label{eq:asymptotic-oracle-general-conformal-VFCP}
   \mbox{Width}(\widehat{C}_{\alpha}^{\texttt{VF-lin}}) ~=~ 2T_{\alpha,\widehat{\theta}}^* ~\le~ 2 \min_{\theta\in\Theta} Q_{\alpha,\theta} + R_N,
\end{equation}
and the empirical oracle inequality for VFCP sets in Algorithm~\ref{alg:general-conformal} would be
\begin{equation}\label{eq:empirical-oracle-general-conformal-VFCP}
\mbox{Width}(\widehat{C}_{\alpha}^{\texttt{VF-lin}}) ~=~ 2T_{\alpha,\widehat{\theta}}^* ~\le~ 2 \min_{\theta\in\Theta}T_{\alpha,\theta} + R_N',
\end{equation}
for some rates $R_N, R_N'$. 
Recall $Q_{\alpha,\theta}$ defined in~\eqref{eq:def-quantile-theta}.
Consider the analog of assumption~\ref{eq:Holder-VFCP}.
\begin{enumerate}[label=\bf(A0')]
\item \label{assump:general-holder-VFCP} For all $ \theta \in \Theta$, there exists $r^*, \gamma \in (0, 1]$, such that $F_{\theta}^{-1}(\cdot)$ is $\gamma$-H{\"o}lder continuous on $[1-\alpha-r^*,1-\alpha+r^*]$ with  H\"older continuity constant $L_{\theta}$, i.e., for all $\theta\in\Theta$, and all $q_1, q_2\in [1-\alpha-r^*,1-\alpha+r^*]$,
$
| F_{\theta}^{-1}(q_1) - F_{\theta}^{-1}(q_2) | \le L_{\theta}|q_1 - q_2|^{\gamma}.
% , \mbox{ for all } q_1, q_2 .
$
\end{enumerate}
Assumption~\ref{assump:general-holder-VFCP} is an extension of~\ref{assump:general-holder-EFCP} and is clearly more restrictive. Note that the comments on the verification of~\ref{assump:general-holder-EFCP} also apply to verifying~\ref{assump:general-holder-VFCP}.
% Assumption~\ref{assump:general-holder} is a rewording of the condition that the density of $|Y-X^{\top}\theta|$ is bounded away from zero in the neighborhood of $Q_{\alpha,\theta}$ and is needed to ensure that $T_{\alpha,\theta}$ (and $T_{\alpha,\theta}^*$) is close to $Q_{\alpha,\theta}$ for each $\theta\in\Theta$. Note that Assumption~\ref{assump:general-holder} depends on $\alpha\in(0, 1)$. If the density of $Y - X^{\top}\theta$ is bounded away from zero on $[F_\theta^{-1}\bigl( Q_{\alpha,\theta} - r^* \bigr),F_\theta^{-1}\bigl( Q_{\alpha,\theta} + r^* \bigr) ]$, then~\ref{assump:general-holder} will hold true with $\gamma = 1$; see Proposition~\ref{prop:simple-sufficient-quantile-holder}. If $(X, Y)$ has a joint normal distribution, then Assumption~\ref{assump:general-holder} holds true with $\gamma = 1$ and a constant $L_\theta$ satisfying $L_{\theta} \le C(1 + \|\theta\|_2)$ for some constant $C$ (Proposition \ref{prop:joint-normal}).
\begin{thm}\label{thm:general-oracle-VFCP}
% Recall $T_{\alpha,\widehat{\theta}}^*$ obtained in the final prediction step \ref{me:general-final-prediction} of Algorithm \ref{alg:general-conformal}, assuming we have bounds on
Recall $R(\mathcal{I})$ from~\eqref{eq:general-distribution-function-assumptions}.
Suppose Algorithm~\ref{alg:general-conformal} is used with a bounded set $\Theta\subseteq\mathbb{R}^d$ and suppose assumption~\ref{assump:general-holder-VFCP} holds true with $r^* \geq \max \{ R (\mathcal{I}_1 ), R (\mathcal{I}_2 ), 2|\mathcal{I}_1|^{-1}, 2|\mathcal{I}_2|^{-1} \} $. Then setting $L_{\Theta} := \sup_{\theta\in\Theta}L_{\theta}$, we obtain
\begin{equation}\label{eq:general-final_quantile_strong}
% T_{\alpha,\widehat{\theta}}^* 
\mathrm{Width}(\widehat{C}_{\alpha}^{\texttt{VF-lin}})
~\leq~ \min_{\theta \in \Theta }\,\mathrm{Width}(C_{\alpha,\theta}^{\texttt{orc-lin}}) ~+~  4 L_{\Theta}   \left( R (\mathcal{I}_1 )^\gamma+R (\mathcal{I}_2 )^\gamma + \frac{2^{\gamma}}{|\mathcal{I}_1|^{\gamma}} + \frac{2^{\gamma}}{|\mathcal{I}_2|^{\gamma}}\right).
\end{equation}
% If the constant $L$ in assumption \ref{assump:general-holder} depends on $\theta$, then 
% \begin{equation}
% {q}^*_{\widehat{\theta}} \le \min_{\theta} Q_{\alpha,\theta} + \mathfrak{C} L_{\theta^*} \biggl( R(D_1)^{\gamma} + R(\mathcal{D}_2)^{\gamma} \biggr),
% \end{equation}
% where $\theta^*$ is the minimizer of $Q_{\alpha,\theta}$ over all $\theta\in\mathbb{R}^d$.
% Similarly, if we know that $\|\theta^*\|_{\Sigma} \le K$, then
% \begin{equation}
% {q}^*_{\widehat{\theta}} \le \min_{\theta} Q_{\alpha,\theta} + \mathfrak{C} \max_{\|\theta\|_{\Sigma} \le K}L_{\theta} \biggl( R(\mathcal{D}_1)^{\gamma} + R(\mathcal{D}_2)^{\gamma} \biggr).
% \end{equation}
The conclusion continues to hold true if $C_{\alpha,\theta}^{\texttt{orc-lin}}$ is replaced with $\widehat{C}_{\theta}$ on the right hand side. 
\end{thm}
We defer the proof of Theorem~\ref{thm:general-oracle-VFCP} to Appendix \ref{sec:Proof of Theorem general-oracle-final} of the supplementary. Inequality~\eqref{eq:general-final_quantile_strong} is an asymptotic oracle inequality. 
% If $\theta^*$ is the minimizer of $Q_{\alpha,\theta}$ over all $\theta\in\mathbb{R}^d$, then
% \[
% T_{\alpha,\widehat{\theta}}^* ~\le~ \min_{\theta\in\mathbb{R}^d}\,Q_{\alpha,\theta} ~+~ \mathfrak{C}L_{\theta^*}\left\{ R (\mathcal{D}_1 )^\gamma+R (\mathcal{D}_2 )^\gamma \right\}.
% \]
Further, because the conclusion of Theorem~\ref{thm:general-oracle-VFCP} holds true when $C_{\alpha,\theta}^{\texttt{orc-lin}}$ is replaced by $\widehat{C}_{\theta}$, we also obtain an empirical oracle inequality for $\widehat{C}_{\alpha}^{\texttt{VF-lin}}$. However, the rate of convergence in the oracle inequalities is not explicit. They depend on $R(\mathcal{I}_1)$ and $R(\mathcal{I}_2)$. Combined with Proposition~\ref{prop:R1-bound}, Theorem~\ref{thm:general-oracle-VFCP} does imply an explicit rate for the slack in the oracle inequalities and the rate depends on the ``size'' of $\Theta$. This is in contrast to how the slack behaves for EFCP.
%%%%%%%%%%%%%%%%%%%%%%%%%%%%%%%%%%%%%%%%%%%%%%%%%%%%%%%%%%%%
%%%%%%%%%%%%%%%%%%%%%%%%%%%%%%%%%%%%%%%%%%%%%%%%%%%%%%%%%%%%
\section{Best Ridge Regression Fixed Width Prediction Set}\label{appsec:Best-ridge-fixed-width}
In this section, we consider EFCP and VFCP with ridge regression training methods and prove asymptotic oracle inequalities. Although we reduce the set $\Theta$ of linear combinations, we will require more assumptions than in Section~\ref{sec:general-algorithm} to prove an asymptotic oracle inequality. This is mostly because we need to show the ridge regression estimators converge to ``population'' ridge regression functionals.

% , which is a special instance of Algorithm \ref{alg:general-conformal}. Theorem~\ref{thm:coverage}, it follows that the prediction set returned by Algorithm~\ref{alg:general-conformal} achieves the target coverage rate. 
For any $\lambda\in\mathbb{R}$, define $\beta^*_{\lambda}$ as the ``population'' ridge regression functional given by
\begin{equation}\label{eq:population-ridge-regression}
    \beta^*_{\lambda}:= \argmin_{ \theta  \in\mathbb{R}^d}\, \E [ ( Y - X^{\top}\theta )^2 ] +\lambda \| \theta \|_2^2.
\end{equation}
With $\{1,2, \ldots, N\}$ randomly split into three parts $\mathcal{I}_1, \mathcal{I}_2,$ and $\mathcal{I}_3$, we have the ridge regression estimator from $(X_i, Y_i), i\in\mathcal{I}_1$ as
\begin{equation}\label{eq:empirical-ridge-regression}
\widehat{\beta}_{\lambda} = \argmin_{\theta\in\mathbb{R}^d}\,\frac{1}{|\mathcal{I}_1|}\sum_{i\in\mathcal{I}_1}(Y_i - X_i^{\top}\theta)^2 + \lambda\|\theta\|^2.
\end{equation}
The EFCP, VFCP algorithms with ridge regression training methods are presented in Algorithm \ref{alg:conformal-ridge}.

\begin{algorithm}[htbp]
\SetAlgoLined
\KwIn{Data $(X_i,Y_i),i=1,\dots,N$, coverage level $1 - \alpha \in (0,1)$, and a parameter $\kappa < 1$.}
\KwResult{Optimal fixed width prediction set based on ridge regression.}
  (Splitting) Randomly split $\{1, \dots, N\}$ into three subsets $\mathcal{I}_1,\mathcal{I}_2, \mathcal{I}_3$. For EFCP, $\mathcal{I}_3 = \emptyset$. Set $\mathcal{D}_1 = \{(X_i, Y_i):i\in\mathcal{I}_1\}, \mathcal{D}_2 = \{(X_i, Y_i):i\in\mathcal{I}_2\}$, and $\mathcal{D}_3 = \{(X_i, Y_i):i\in\mathcal{I}_3\}$. 
  
  (Estimation) \label{me:estimation} Set $\widehat{\Sigma}_{1} := |\mathcal{I}_1|^{-1}\sum_{i\in\mathcal{I}_1} X_iX_i^{\top}$. For $\lambda \ge - \kappa\lambda_{\min}(\widehat{\Sigma}_{1}),$ calculate $\widehat{\beta}_{\lambda}$ in~\eqref{eq:empirical-ridge-regression}.
  
  (Initial Prediction)\label{me:prediction} Define
   $T_{\alpha,\lambda}:= \lceil(1-\alpha)(1+|\mathcal{I}_2|)\rceil \text{-th quantile of }|Y_i - X_i^{\top}\widehat{\beta}_{\lambda}|,\;i\in\mathcal{I}_2.$
  
  (Selection) \label{me:selection} Select the parameter as $$
  \widehat{\lambda}_{\alpha}: = \argmin_{\lambda \ge -\kappa \lambda_{\min}(\widehat{\Sigma}_1)}\, T_{\alpha,\lambda}.$$
  
   (Final Prediction for VFCP)\label{me:final-prediction} Define
   $T_{\alpha,\widehat{\lambda}_{\alpha}}^*:= 
   \lceil(1-\alpha)(1+|\mathcal{I}_3|)\rceil\text{-th quantile of } |Y_i - X_i^{\top}\widehat{\beta}_{\widehat{\lambda}_{\alpha}}|,\;i\in\mathcal{I}_3.$
  
  \Return $\widehat{C}_{\alpha}^{\texttt{EF-ridge}} :=\{(x, y):\,| y-x^{\top}\widehat{\beta}_{\widehat{\lambda}_{\alpha}}|\le T_{\alpha,\widehat{\lambda}_{\alpha}}\}$ and $\widehat{C}_{\alpha}^{\texttt{VF-ridge}} :=\{(x, y):\, |y-x^{\top}\widehat{\beta}_{\widehat{\lambda}_{\alpha}}| \le T_{\alpha,\widehat{\lambda}_{\alpha}}^*\}$.

  \caption{Best Ridge Regression Fixed Width Prediction Set}
  \label{alg:conformal-ridge}
\end{algorithm}
The parameter $\kappa$ in Algorithm~\ref{alg:conformal-ridge} is used to ensure that the empirical ridge regression problem in~\eqref{eq:empirical-ridge-regression} is convex. In contrast to the most commonly used setting, we allow the penalty parameter in ridge regression to be negative; this is motivated by considerations in~\cite{kobak2020optimal}. Theorems~\ref{thm:coverage} and~\ref{thm:general-oracle-final} imply that the prediction intervals $\widehat{C}_{\alpha}^{\texttt{EF-ridge}}$ and $\widehat{C}_{\alpha}^{\texttt{VF-ridge}}$ returned by Algorithm~\ref{alg:conformal-ridge} satisfy
\[
\mathbb{P}\left((X_{N+1},Y_{N+1})\in\widehat{C}_{\alpha}^{\texttt{EF-ridge}}\right) \ge 1 - \alpha - \mathfrak{C}\sqrt{\frac{d}{N-|\mathcal{I}_1|}},\quad\mbox{and}\quad \mathbb{P}\left((X_{N+1},Y_{N+1})\in\widehat{C}_{\alpha}^{\texttt{VF-ridge}}\right) \ge 1 - \alpha.
\]
Additionally, Theorem~\ref{thm:general-oracle-final} implies two versions of empirical oracle inequalities. We will use the notation
\[
\widehat{C}_{\theta} = \{(x, y):\,|y - x^{\top}\theta| \le T_{\alpha,\theta}\},\quad\mbox{and}\quad {C}_{\alpha,\theta}^{\texttt{orc-lin}} := \{(x, y):\,|y - x^{\top}\theta| \le Q_{\alpha,\theta}\}
\]
from Section~\ref{sec:general-algorithm} to describe the implications of Theorem~\ref{thm:general-oracle-final}. Suppose the assumption of i.i.d. observations and~\ref{assump:general-holder-VFCP} holds true. Set $L_{\Lambda_1} = \sup\{L_{\widehat{\beta}_{\lambda}}:\, \lambda \ge -\kappa\lambda_{\min}(\widehat{\Sigma}_1)\}$, then Theorems~\ref{thm:general-oracle-final} and~\ref{thm:general-oracle-VFCP} along with~\eqref{eq:VC-class-bound} implies that with probability at least $1 - \delta$,
\begin{equation}\label{eq:empirical-oracle-inequality-ridge}
    \begin{split}
        \mbox{Width}(\widehat{C}_{\alpha}^{\texttt{EF-ridge}}) ~&\le~ \min_{\lambda \ge -\kappa\lambda_{\min}(\widehat{\Sigma}_1)}\mbox{Width}({C}_{\alpha,\widehat{\beta}_{\lambda}}^{\texttt{orc-lin}}) + \mathfrak{C}L_{\Lambda_1}\left(\frac{\log(2/\delta)}{|\mathcal{I}_2|}\right)^{\gamma/2},\\
        \mbox{Width}(\widehat{C}_{\alpha}^{\texttt{VF-ridge}}) ~&\le~ \min_{\lambda \ge -\kappa\lambda_{\min}(\widehat{\Sigma}_1)}\mbox{Width}({C}_{\alpha,\widehat{\beta}_{\lambda}}^{\texttt{orc-lin}}) + \mathfrak{C}L_{\Lambda_1}\left(\frac{d + \log(1/\delta)}{\min\{|\mathcal{I}_2|, |\mathcal{I}_3|\}}\right)^{\gamma/2}.\\
    \end{split}
\end{equation}
The first inequality follows from~\eqref{eq:general-conformal-oracle-EFCP} and the second follows from~\eqref{eq:general-final_quantile_strong}. The same inequalities hold true with ${C}_{\alpha,\widehat{\beta}_{\lambda}}^{\texttt{orc-lin}}$ replaced by $\widehat{C}_{\widehat{\beta}_{\lambda}}.$ For the inequalities in~\eqref{eq:empirical-oracle-inequality-ridge}, we need $r^*$ in~\ref{assump:general-holder-EFCP} and~\ref{assump:general-holder-VFCP} to satisfy (respectively) $$r^* \ge \mathfrak{C}\sqrt{(\log(2/\delta))/|\mathcal{I}_2|} \quad\mbox{and}\quad r^* \ge \mathfrak{C}\sqrt{(d + \log(2/\delta))/\min\{|\mathcal{I}_2|, |\mathcal{I}_3|\}}.$$ 
% It might be of interest to note that for inequalities in~\eqref{eq:empirical-oracle-inequality-ridge}, we only need assumption~\ref{assump:general-holder} to hold true for $\Theta = \{\widehat{\beta}_{\lambda}:\,\lambda \ge -\kappa\lambda_{\min}(\widehat{\Sigma}_1)\}.$

Note that~\eqref{eq:empirical-oracle-inequality-ridge} refers to empirical oracle inequalities because $\widehat{C}_{\alpha,\widehat{\beta}_{\lambda}}^{\texttt{orc-lin}}$ still depends on the data through $\widehat{\beta}_{\lambda}$. The asymptotic version of $\widehat{C}_{\widehat{\beta}_{\lambda}}$ and $C_{\alpha,\widehat{\beta}_{\lambda}}^{\texttt{orc-lin}}$ is given by $C_{\alpha,\beta_{\lambda}^*}^{\texttt{orc-lin}}$. Hence, the asymptotic oracle inequality for Algorithm~\ref{alg:conformal-ridge} is given by
\begin{equation}\label{eq:asymptotic-oracle-to-prove-ridge}
\max\{\mbox{Width}(\widehat{C}_{\alpha}^{\texttt{EF-ridge}}), \mbox{Width}(\widehat{C}_{\alpha}^{\texttt{VF-ridge}})\} ~\le~ \min_{\lambda \ge -\kappa\lambda_{\min}(\widehat{\Sigma}_1)}\mbox{Width}(C_{\alpha,\beta^*_{\lambda}}^{\texttt{orc-lin}}) + R_N,
\end{equation}
for some $R_N$ converging to zero. Following~\eqref{eq:empirical-oracle-inequality-ridge}, it suffices to show that
\begin{equation}\label{eq:suffices-to-prove-asymp-oracle-ridge}
\min_{\lambda \ge -\kappa\lambda_{\min}(\widehat{\Sigma}_1)}\mbox{Width}({C}_{\alpha,\widehat{\beta}_{\lambda}}^{\texttt{orc-lin}}) ~\le~ \min_{\lambda \ge -\kappa\lambda_{\min}(\widehat{\Sigma}_1)}\mbox{Width}(C_{\alpha,\beta^*_{\lambda}}^{\texttt{orc-lin}}) + R_N',
\end{equation}
for some $R_N'$ converging to zero. Then the asymptotic oracle inequality~\eqref{eq:asymptotic-oracle-to-prove-ridge} holds true with $R_N = R_N' + O((d + \log(1/\delta))/\min\{|\mathcal{I}_2|,|\mathcal{I}_3|\})^{\gamma/2}$.
Proving~\eqref{eq:suffices-to-prove-asymp-oracle-ridge} involves multiple steps:
\begin{enumerate}
    \item Prove that $\widehat{\beta}_{\lambda} - \beta^*_{\lambda}$ converges to zero uniformly over $\lambda \ge - \kappa\lambda_{\min}(\widehat{\Sigma}_1)$.
    \item Prove that $F_{\widehat{\beta}_{\lambda}}(\cdot)$ is close to $F_{\beta^*_{\lambda}}(\cdot)$ uniformly in $\lambda \ge - \kappa\lambda_{\min}(\widehat{\Sigma}_1)$; recall $F_{\theta}(\cdot)$ from~\eqref{F}.
    \item Prove that $T_{\alpha,\lambda}$ is close to $Q_{\alpha,\beta^*_{\lambda}}$ uniformly over $\lambda \ge - \kappa\lambda_{\min}(\widehat{\Sigma}_1)$.
\end{enumerate}
The first two steps above require several moment conditions which is what distinguishes the results that follow from the ones above. We will first list the assumptions required to complete these steps. Most commonly found results on rate of convergence of ridge regression estimator assume sub-Gaussian moment assumptions on the response and covariates. For generality, we also consider response and covariate distributions with only a fixed number of finite moments. Hence the resulting bounds are of independent interest in the study of ridge regression with increasing dimension.

\begin{enumerate}[label=\bf(A\arabic*)]
\setcounter{enumi}{0}
\item \label{assump:DGP}The observations $(X_i, Y_i) \in \mathbb{R}^d \times \mathbb{R}, 1 \leq i \leq N$ are independent and identically distributed.% (i.i.d.).
\item \label{assump:finite-Y} There exists some $q_y \ge 2$ and a constant $K_{y} \in(0, \infty)$ such that
$
\left(\mathbb{E}\left[\left|Y_{i} \right|^{q_y}\right]\right)^{1 / q_y} \leq K_{y}<\infty.
% , \quad \text { for all } \quad 1 \leq i \leq N .
$
\item \label{assump:subGauss} Random vectors $X_1, \ldots, X_N$ are i.i.d., and for $\Sigma := \mathbb{E}[X_1X_1^{\top}]$, there exists a constant $K_{x} \in(0, \infty)$ such that for all $u\in S^{d-1}$,
$\E [\exp ({ K_{x}^{-2}}{|u^{\top} \Sigma^{-1 / 2} X_{i}|^{2}})] \leq 2.$ % \text{ for all }1 \leq i \leq N \text{ and }u \in S^{d-1}.$$ 
\item \label{assump:finite-X} Random vectors $X_1, \ldots, X_N$ are i.i.d., and for $\Sigma := \mathbb{E}[X_1X_1^{\top}]$, there exists some constant $q_x \ge 2$, and a constant $K_{x} \in (0, \infty)$ such that for all $u\in S^{d-1}$,
$
\mathbb{E}\bigl|u^{\top} \Sigma^{-1 / 2} X_i\bigr|^{q_x} \leq K_x^{q_x}.
% , \text{ for all }  1\le i\le N\mbox{ and }u \in S^{d-1}.
$
\item \label{assump:finite-X-indep} Random vectors $X_{1}, \ldots, X_{N}$ satisfy \ref{assump:finite-X} and also a structural condition that for any $1 \le i \le N$, there exists a random vector $Z_i$ such that $X_i = \Sigma^{1/2} Z_i$, and that $Z_{ij}, 1 \leq j \leq d$, are independent random variables. 
\item \label{assump:density} For
% $\Sigma := \mathbb{E}[X_1X_1^{\top}] $ and 
all $\lambda \ge -\kappa \lambda_{\min}(\widehat{\Sigma}_1)$, $ Y_1-X_1^{\top}\beta^{*}_{\lambda} $ has a Lebesgue density bounded by $ \Psi \in (0,\infty).$
\item \label{assump:density_conditional} The conditional distribution of $Y_1$ given $X_1$ has a Lebesgue density bounded by $\psi \in (0,\infty)$.
%where the quantile function $F^{-1}$ is defined in \eqref{def:quantile}.
% \item \label{assump:general-holder} For all $ \lambda \ge - \kappa \lambda_{\min}(\Sigma)$ for some $\kappa \le 1$, there exists some $r^*, \gamma$ both in $(0, 1]$, such that with probability 1, $F_{\beta^*_{\lambda}}^{-1}(\cdot)$ is $\gamma$-H{\"o}lder on $[q_\alpha(\lambda)-r^*,q_\alpha(\lambda)+r^*]$ continuous with  H\"older continuity constant $L$, i.e.,
% $$
% | F_{\beta^*_{\lambda} }^{-1}(q_1) - F_{\beta^*_{\lambda}}^{-1}(q_2) | \le L_{\lambda}|q_1 - q_2|^{\gamma}, \mbox{ for all } q_1, q_2 \in [q_\alpha(\lambda)-r^*,q_\alpha(\lambda)+r^*].
% $$
% where the quantile function $F^{-1}$ is defined in \eqref{def:quantile} and $\gamma \in [0,1]$.
% Further, there exists some constant $q_x \geq 2$ and a constant $K_{x} \in (0, \infty)$ such that
% $$
% \mathbb{E}\bigl|a^{\top} Z_i\bigr|^{q_x} \leq K_x^{q_x}, \text{ for all }  a \in \mathbb{R}^{d},\text{ with } \|a\|_2=1,
% $$
\end{enumerate}
Assumptions~\ref{assump:DGP}--\ref{assump:finite-X-indep} are standard moment assumptions on features and covariates. Assumption~\ref{assump:density_conditional} is stronger than~\ref{assump:density} and holds true under a well-specified linear model $Y = \theta_0^{\top}X + \epsilon$ with a bounded conditional density of $\epsilon$ given $X$.

We now prove~\eqref{eq:suffices-to-prove-asymp-oracle-ridge}. We start with a deterministic statement which serves as a road map for results under the stated assumptions. For notational convenience, set
% \[
$\Lambda_1 := \{\lambda\in\mathbb{R}:\, \lambda \ge -\kappa\lambda_{\min}(\widehat{\Sigma}_1)\}.
$
% \]
\begin{prop}\label{prop:oracle-final}
% Recall $\widehat{q}_\alpha^{\mathrm{final}}$ obtained in the final prediction step \ref{me:final-prediction} of Algorithm \ref{alg:conformal-ridge} and $\widehat{\beta}_\lambda$ obtained in the estimation step \ref{me:estimation} of Algorithm \ref{alg:conformal-ridge} using $(X_i, Y_i), i \in \mathcal{I}_1$, assuming we have bounds on 
Set
% \begin{equation}\label{eq:distribution-function-assumptions}
% \begin{split}  
% \mathcal{E}_1(\mathcal{I}) ~&:=~ \sup_{t\ge 0}\sup_{\lambda\in\Lambda_1}\, \biggl|\frac{1}{ |\mathcal{I}| } \sum_{i \in \mathcal{I}} \mathbbm{1}{\{ |Y_i - X_i^\top \widehat{\beta}_\lambda | \leq t \} } -F_{\widehat{\beta}_{\lambda}}(t) \biggr|,\quad\mbox{for any}\quad \mathcal{I}\subseteq[N],\\
$
\mathcal{E}(\mathcal{I}_1):=  \sup_{t \ge 0}\,\sup_{\lambda\in\Lambda_1}\,  |  F_{\widehat{\beta}_{\lambda}} (t) - F_{ \beta_{\lambda}^*} (t) |.
$
%  \end{split}
% \end{equation}
Suppose assumption~\ref{assump:general-holder-VFCP} holds true with 
\begin{equation}\label{eq:r-star-requirement}
r^* ~\geq~ \mathcal{E}(\mathcal{I}_1).
\end{equation}
Then setting $L_{\Lambda_1} := \sup\{L_{\lambda}:\,\lambda\in\Lambda_1\}$ for the parameter $\kappa < 1$ in Algorithm~\ref{alg:conformal-ridge}, we obtain
\begin{equation}\label{eq:final_quantile}
\min_{\lambda \ge -\kappa\lambda_{\min}(\widehat{\Sigma}_1)}\mathrm{Width}({C}_{\alpha,\widehat{\beta}_{\lambda}}^{\texttt{orc-lin}}) ~\le~ \min_{\lambda \ge -\kappa\lambda_{\min}(\widehat{\Sigma}_1)}\mathrm{Width}(C_{\alpha,\beta^*_{\lambda}}^{\texttt{orc-lin}})
% T^*_{\alpha,\widehat{\lambda}_{\alpha}} \leq \min_{\lambda\in\Lambda_1} Q_{\alpha, \beta^*_\lambda}  
~+~  L_{\Lambda_1}\mathcal{E}^{\gamma}(\mathcal{I}_1).
% \left(\mathcal{E}_1 (\mathcal{I}_2)^{\gamma}+\mathcal{E}_1 (\mathcal{I}_3)^{\gamma}+ \mathcal{E}_2  (\mathcal{I}_1)^{\gamma} +\frac{2^\gamma}{|\mathcal{I}_2|^{\gamma}}+\frac{2^{\gamma}}{|\mathcal{I}_3|^{\gamma}} \right).
\end{equation}
% Moreover, the conclusion continues to hold true if we replace $Q_{\alpha, \beta^*_\lambda}$ with $T^*_{\alpha, \lambda}$.
\end{prop}
Proposition~\ref{prop:oracle-final} follows from Proposition~\ref{prop:holder} (in Section~\ref{appsec:useful-propositions} of the supplementary file). 
% In an intuitive way, the oracle inequality can be broken into two parts. Firstly, recall from the empirical oracle inequality in Theorem~\ref{thm:final_holder_ecdf} that $T^*_{\alpha,\widehat{\lambda}_{\alpha}}$ is known to be close to $\min\{Q_{\alpha,\widehat{\beta}_{\lambda}}:\,\lambda\in\Lambda_1\}$. Also, recall that this only requires a bound on $\mathcal{E}_1(\mathcal{I}_2)$ and $\mathcal{E}_1(\mathcal{I}_3)$. With this empirical inequality hand, inequality~\eqref{eq:final_quantile} can be obtained by proving a bound on the difference between $Q_{\alpha,\beta^*_{\lambda}}$ and $Q_{\alpha,\widehat{\beta}_{\lambda}}$ uniformly over $\lambda\in\Lambda_1$. This is exactly where $\mathcal{E}(\mathcal{I}_1)$ comes from. 
Recall from~\eqref{eq:def-quantile-theta} that $Q_{\alpha,\widehat{\beta}_{\lambda}}$ and $Q_{\alpha,\beta^*_{\lambda}}$ are quantiles of $F_{\widehat{\beta}_{\lambda}}(\cdot)$ and $F_{\beta^*_{\lambda}}(\cdot)$, respectively. Proposition~\ref{prop:oracle-final} essentially follows from the fact that quantiles of close distributions are close when one of the distributions has density bounded away from zero. All the assumptions~\ref{assump:DGP}--\ref{assump:density_conditional} listed above are used in bounding $\mathcal{E}(\mathcal{I}_1)$. It might of interest to note that for Proposition~\ref{prop:oracle-final} we only need assumption~\ref{assump:general-holder-VFCP} for $\Theta = \{\beta^*_{\lambda}:\, \lambda\in\Lambda_1\}.$

Regarding the requirement~\eqref{eq:r-star-requirement} on $r^*$, we will show in the following that $\mathcal{E}(\mathcal{I}_1)$ converges to zero as $|\mathcal{I}_1| \to \infty$. This also requires some control on the growth of the dimension $d$ of the covariates. For many distributions of $(X, Y),$ assumption~\ref{assump:general-holder-VFCP} holds true with a fixed positive number $r^*$ and hence the requirement~\eqref{eq:r-star-requirement} holds true with a large enough number of observations in each split of the data. For this reason and to reduce the notational burden, we will not mention this requirement on $r^*$ in the following results.

% We will now proceed to find explicit rate bounds in the asymptotic oracle inequality~\eqref{eq:final_quantile}. Firstly, note that $\mathcal{E}_1(\mathcal{I}) \le R(\mathcal{I})$ for any $\mathcal{I}\subseteq[N]$; recall $R(\mathcal{I})$ from~\eqref{eq:general-distribution-function-assumptions}. Under the i.i.d. assumption~\ref{assump:DGP}, there exists a universal constant $\mathfrak{C}$ such that with probability at least $1 - \delta$,
% \[
% \mathcal{E}_1(\mathcal{I}) \le R(\mathcal{I}) \le \mathfrak{C}\sqrt{\frac{d + \log(2/\delta)}{|\mathcal{I}|}}.
% \]
% See Proposition~\ref{prop:R1-bound}. 
% Here $\mathfrak{C}$ is a universal constant. 
% With this bound, inequality~\eqref{eq:final_quantile} can be simplified as
% \begin{equation}\label{eq:simplified-bound-asymp-oracle-inequality}
% T^*_{\alpha,\widehat{\lambda}_{\alpha}} \le \min_{\lambda\in\Lambda_1}Q_{\alpha,\beta^*_{\lambda}} + \mathfrak{C} L_{\Lambda_1}\left(\left(\frac{d + \log(2/\delta)}{\min\{|\mathcal{I}_2|, |\mathcal{I}_3|\}}\right)^{\gamma/2} + \mathcal{E}(\mathcal{I}_1)^{\gamma}\right).
% \end{equation}
% Hence, it remains to 
We will now
bound $\mathcal{E}(\mathcal{I}_1)$ for which we make use of assumptions~\ref{assump:subGauss},~\ref{assump:finite-X},~\ref{assump:density}, and~\ref{assump:density_conditional}. We now state a simple deterministic inequality that bounds $|  F_{\widehat{\beta}_{\lambda}} (t) - F_{ \beta_{\lambda}^*} (t) |$ in terms of the estimation error $\|\widehat{\beta}_{\lambda} - \beta^*_{\lambda}\|_{\Sigma}$; recall the definition $\|\theta\|_{\Sigma} = \sqrt{\theta^{\top}\Sigma\theta}$.
\begin{lemma}(Bounds for $\mathcal{E}(\mathcal{I}_1)$)\label{lem:R2-bound}
For any $\lambda \in \Lambda_1$, we have
\begin{equation}\label{eq:cdf-beta-hat-beta-bound}
\sup_{t \ge 0} \,| F_{\widehat{\beta}_{\lambda}} (t) - F_{ \beta_{\lambda}^*} (t) | ~\le~ 
\mathfrak{C}
\begin{cases}
\Psi K_{x}  \| \widehat{\beta}_{\lambda} - \beta_{\lambda}^* \|_{\Sigma} \log \bigl(1/ ( \Psi K_{x} \|  \widehat{\beta}_{\lambda} - \beta_{\lambda}^* \|_{\Sigma}) \bigr), &\mbox{under } \ref{assump:density},~\ref{assump:subGauss},\\
(\Psi K_{x}\| \widehat{\beta}_{\lambda} - \beta_{\lambda}^* \|_{\Sigma} )^{q_x/(1+q_x)}, &\mbox{under } \ref{assump:density},~\ref{assump:finite-X},\\
\psi \| \widehat{\beta}_{\lambda} - \beta_{\lambda}^* \|_{\Sigma}, &\mbox{under }\ref{assump:density_conditional}.
\end{cases}
\end{equation}
% \begin{outline}
% \1  Under the density assumption \ref{assump:density} and assume sub-gaussianity on the normalized $X$, namely \ref{assump:subGauss}, there exists a universal constant $\mathfrak{C}$ such that for all $\lambda > -\lambda_{\min}(\Sigma)$,
% \begin{equation}\label{eq:cdf-bound}
%  \sup_{t \ge 0} \,| F_{\widehat{\beta}_{\lambda}} (t) - F_{ \beta_{\lambda}^*} (t) | ~\leq~ \mathfrak{C} \Psi K_{x}  \| \widehat{\beta}_{\lambda} - \beta_{\lambda}^* \|_{\Sigma} \log \bigl(1/ ( \Psi K_{x} \|  \widehat{\beta}_{\lambda} - \beta_{\lambda}^* \|_{\Sigma}) \bigr).
% \end{equation}

% \1 Under the density assumption \ref{assump:density} and the finite-moment condition \ref{assump:finite-X}, there exists a universal constant $\mathfrak{C}$ such that for all $\lambda > -\Sigma_{\min}(\Sigma)$,
% \begin{equation}\label{eq:cdf-fm-bound}
% \, \sup_{t \ge 0} |F_{\widehat{\beta}_{\lambda}}(t) - F_{\beta_{\lambda}^*}(t)| ~\leq~ \mathfrak{C} (\Psi K_{x}\| \widehat{\beta}_{\lambda} - \beta_{\lambda}^* \|_{\Sigma} )^{q_x/(1+q_x)}.
% \end{equation}

% \1 Under the strong density assumption \ref{assump:density_conditional}, it holds for all $\lambda > -\Sigma_{\min}(\Sigma)$ that
% \begin{equation}\label{eq:cdf-tighter-bound} 
%  \sup_{t \ge 0} \,| F_{\widehat{\beta}_{\lambda}} (t) - F_{ \beta_{\lambda}^*} (t) | ~\leq~ \psi \| \widehat{\beta}_{\lambda} - \beta_{\lambda}^* \|_{\Sigma}.
% \end{equation}
% \end{outline}
\end{lemma}
The proof of Lemma~\ref{lem:R2-bound} is provided in Appendix \ref{sec:Proof of Lemma R2-bound} of the supplementary. Hence, to obtain an explicit rate in the oracle inequality~\eqref{eq:final_quantile}, it suffices to prove bounds for the estimation error of $\widehat{\beta}_{\lambda}$.

We will further reduce the problem of bounding $\|\widehat{\beta}_{\lambda} - \beta^*_{\lambda}\|_{\Sigma}$ to the problem of bounding averages of random vectors and random matrices via a simple deterministic inequality. This is similar to the deterministic inequality results of~\cite{kuchibhotla2018model}. Recall the definition of $\widehat{\beta}_{\lambda}$ and $\beta^*_{\lambda}$ from~\eqref{eq:empirical-ridge-regression} and~\eqref{eq:population-ridge-regression}, respectively. 
\begin{lemma}\label{lem:ridge}(Deterministic Inequality for Ridge Regression)
% Define ridge regression estimator and target as
% \begin{align*}
%     \widehat{\beta}_{\lambda} &:= \argmin_{\theta \in\mathbb{R}^d}\, \frac{1}{n} \sum_{i=1}^n (Y_i - X_i^\top \theta )_2^2 +\lambda \|\theta\|_2^2,\\
%     \beta^*_{\lambda} &:= \argmin_{ \theta  \in\mathbb{R}^d}\, \E [ ( Y - X^{\top}\theta )_2^2 ] +\lambda \| \theta \|_2^2.
% \end{align*} 
Set
\begin{equation}\label{eq:gammahat_def}
\Sigma := \mathbb{E}[XX^{\top}],\quad\Gamma := \mathbb{E}[XY],\quad\widehat{\Sigma}_1 := \frac{1}{|\mathcal{I}_1|}\sum_{i\in\mathcal{I}_1} X_i X_i^{\top},\quad\mbox{and}\quad \widehat{\Gamma}_1 := \frac{1}{|\mathcal{I}_1|}\sum_{i\in\mathcal{I}_1} X_i Y_i.
\end{equation}
Define
\begin{equation}\label{eq:D_Sigma_def}
\mathcal{D}_{\Sigma} := \|\Sigma^{-1/2}\widehat{\Sigma}_1\Sigma^{-1/2} - I_d\|_{\mathrm{op}}    .
\end{equation}
Then for any $c\le 1$, we have
\begin{equation}\label{eq:oracle-ridge}
\sup_{\lambda \ge -c\lambda_{\min}(\Sigma)}\|\widehat{\beta}_{\lambda} - \beta^*_{\lambda}\|_{\Sigma} ~\le~ \frac{1}{(1 - c - \mathcal{D}_{\Sigma})_+}\left[\|\Sigma^{-1/2}(\widehat{\Gamma}_1 - \Gamma)\|_2 + \mathcal{D}_{\Sigma}\frac{\|\Sigma^{-1/2}\Gamma\|_2}{1 - c}\right].
\end{equation}
% In particular, for $c=1/2$, equation \eqref{eq:oracle-ridge} becomes
% \begin{equation}
% \sup_{\lambda \ge -\lambda_{\min}(\Sigma)/2}\|\widehat{\beta}_{\lambda} - \beta^*_{\lambda}\|_{\Sigma(\lambda)} ~\le~ \frac{ \sqrt{2} }{(1 - 2\mathcal{D}_{\Sigma})_+}\left[\|\Sigma^{-1/2}(\widehat{\Gamma} - \Gamma)\|_2 + 4 \mathcal{D}_{\Sigma}\|\Sigma^{-1/2}\Gamma\|_2\right].
% \end{equation}
If assumption \ref{assump:finite-Y} holds true, then $\| \Sigma^{-1/2} \Gamma \|_2 \leq K_y.$
% \begin{equation}\label{eq:covariance}
% \end{equation}
\end{lemma}
The proof is provided in Appendix \ref{sec:Proof of Lemma ridge} of the supplementary. This result provides some insights into the behavior of the algorithm and possible refinements. Firstly, the estimation error of the ridge regression estimator can be bounded in terms of the estimation error of $\widehat{\Gamma}_1$ and $\widehat{\Sigma}_1$. Because these are averages of $d$-dimensional random vectors and random matrices, the behavior of the estimation errors depends crucially on the moments and tails of the underlying distribution of $(X, Y)$. Secondly, the result does not actually depend on the definitions of $\widehat{\Sigma}_1$ and $\widehat{\Gamma}_1$. In particular, if the empirical ridge regression estimator $\widehat{\beta}_{\lambda}$ is redefined as 
% \[
$\widehat{\beta}_{\lambda} = \argmin_{\theta\in\mathbb{R}^d}\,-2\theta^{\top}\widehat{\Gamma}_1 + \theta^{\top}(\widehat{\Sigma}_1 + \lambda I_d)\theta,$
% \]
then the conclusion~\eqref{eq:oracle-ridge} continues to hold true. This suggests that if $\Sigma$ and $\Gamma$ can be estimated better using alternative estimators even under heavy-tailed observations, then this would give better estimation error of ridge regression and also yields a better oracle inequality. For instance, we mention~\cite{mendelson2020robust}, where the authors propose improved estimators of $\Sigma$ that have a sub-Gaussian tail behavior even when $X$ only has four moments. The same work also provides algorithms for estimating $\Gamma$ with heavy-tailed data; also see~\cite{cherapanamjeri2020optimal}. Lemma~\ref{lem:ridge}, hence, can be readily combined with these works to obtain a modified ridge regression estimator and an improvement of the oracle inequality that will be derived below. Given the popularity of ridge regression, we will restrict ourselves to the standard ridge regression. 

Getting back to the problem of oracle inequality, note that Lemma~\ref{lem:ridge} does not readily bound $\mathcal{E}(\mathcal{I}_1)$ in Proposition~\ref{prop:oracle-final}. The reason is that Lemma~\ref{lem:ridge} provides the bound on the estimation error for $\lambda \ge -c\lambda_{\min}(\Sigma)$ while in $\mathcal{E}(\mathcal{I}_1)$, we want the estimation error controlled over $\lambda\in\Lambda_1 = \{\lambda \ge -\kappa\lambda_{\min}(\widehat{\Sigma}_1)\}.$ The resolution is to prove that for any given $\kappa < 1$, there exists a $c \le 1$ such that with high-probability $\Lambda_1 = \{\lambda:\, \lambda \ge -\kappa\lambda_{\min}(\widehat{\Sigma}_1)\} ~\subseteq~ \{\lambda:\,\lambda \ge -c\lambda_{\min}(\Sigma)\}.$
This is equivalent to $\kappa\lambda_{\min}(\widehat{\Sigma}_1) \le c\lambda_{\min}(\Sigma)$.
% , which can be proved by noting that
% \[
% \lambda_{\min}(\widehat{\Sigma}_1) = \frac{\lambda_{\min}(\widehat{\Sigma}_1)}{\lambda_{\min}({\Sigma})}\lambda_{\min}({\Sigma}) \le (1 + \mathcal{D}_{\Sigma})\lambda_{\min}(\Sigma).
% \]
Therefore, it suffices to take $c = \kappa(1 + \mathcal{D}_{\Sigma})$. For small sample sizes, this choice of $c$ need not be smaller than $1$, but with increasing sample sizes $\mathcal{D}_{\Sigma}$ converges to zero and for any $\kappa < 1$, $c$ will eventually be less than $1$. If $\mathcal{D}_{\Sigma}\le (1 - \kappa)/(2 + \kappa)$, then $c = \kappa(1 + \mathcal{D}_{\Sigma}) < 1$ for any $\kappa < 1$ and Lemma~\ref{lem:ridge} implies that
\begin{equation}\label{eq:final-ridge-regression}
\sup_{\lambda \ge -\kappa\lambda_{\min}(\widehat{\Sigma}_1)}\|\widehat{\beta}_{\lambda} - \beta^*_{\lambda}\|_{\Sigma} ~\le~ \frac{2+\kappa}{1-\kappa}\left[\|\Sigma^{-1/2}(\widehat{\Gamma}_1 - \Gamma)\|_2 + \mathcal{D}_{\Sigma}\frac{\|\Sigma^{-1/2}\Gamma\|_2(2 + \kappa)}{2(1-\kappa)}\right].    
\end{equation}
The only random quantities on the right hand side are $\|\Sigma^{-1/2}(\widehat{\Gamma}_1 - \Gamma)\|_2$ and $\mathcal{D}_{\Sigma}$, which we will bound under the moment assumptions listed above. The condition $\mathcal{D}_{\Sigma} \le (1-\kappa)/(2 + \kappa)$ will also be shown to hold with a high probability under these conditions.
Lemmas~\ref{lem:D_Sigma} and~\ref{lem:covariance-bound} provided in supplementary Appendix~\ref{appsec:finite-sample-results-D-Sigma-covariance} provide bounds on $\mathcal{D}_{\Sigma}$ and $\|\Sigma^{-1/2}(\widehat{\Gamma} - \Gamma)\|_2$ under assumptions~\ref{assump:DGP}--\ref{assump:finite-X-indep}. These results may be of independent interest.

Combining the conclusions of Lemmas~\ref{lem:D_Sigma},~\ref{lem:covariance-bound} with~\eqref{eq:oracle-ridge}, Lemma~\ref{lem:R2-bound}, Proposition~\ref{prop:oracle-final}, and~\eqref{eq:empirical-oracle-inequality-ridge}, we can obtain a total of 9 different asymptotic oracle inequalities with different combinations of assumptions. Given that this is basic substitution, we will not present all these results separately and only consider two extreme cases: sub-Gaussian covariates with~\ref{assump:density_conditional} and heavy-tailed covariates with~\ref{assump:density}. Also, for these cases, we will assume that $|\mathcal{I}_1| = |\mathcal{I}_2| = |\mathcal{I}_3| \asymp N/3$ for simplicity. 

Under assumptions~\ref{assump:general-holder-VFCP}--\ref{assump:subGauss}, and~\ref{assump:density_conditional}, we obtain with probability at least $1 - 4\sqrt{d/N}$,
\begin{equation}\label{eq:best-assump-asymp-oracle-ineq}
\begin{split}
\mbox{Width}(\widehat{C}_{\alpha}^{\texttt{EF-ridge}}) &\le \min_{\lambda\in\Lambda_1}\mbox{Width}({C}_{\alpha,\beta^*_{\lambda}}^{\texttt{orc-lin}})\\ 
&\quad+ \mathfrak{C}(K_x, K_y, q_y)L_{\Lambda_1}\left(\left(\frac{d + \log(N/d)}{N}\right)^{\gamma/2} + \left(\frac{d^{q_y - 1}}{N^{2q_y - 3}}\right)^{\gamma/(2q_y)}\right).
\end{split}
\end{equation}
Here $\mathfrak{C}(K_x, K_y, q_y)$ is a constant depending on $K_x, K_y,$ and $q_y$. It is easier to think of this as an asymptotic result rather than a finite sample result because to satisfy $\mathcal{D}_{\Sigma} \le (1-\kappa)/(2 + \kappa)$ on this event, one would require $d/N$ to be small enough. If $q_y \ge 2$, asymptotic oracle inequality~\eqref{eq:best-assump-asymp-oracle-ineq} implies that $\widehat{C}_{\alpha}^{\texttt{EF-ridge}}$ is asymptotically close to the optimal ridge regression based fixed width prediction set whenever $d = o(N)$.

Similarly, under assumptions~\ref{assump:general-holder-VFCP}--\ref{assump:finite-Y},~\ref{assump:finite-X}, and~\ref{assump:density}, with $p = q_xq_y/(q_x + q_y)$, we obtain with probability at least $1 - C(d/N^{1-2/q_x})^{q_x/8} - (d/N^{2-2/p})^{p/4}$,
\begin{equation}\label{eq:worst-assump-asymp-oracle-ineq}
\begin{split}
\mbox{Width}(\widehat{C}_{\alpha}^{\texttt{EF-ridge}}) &\le \min_{\lambda\in\Lambda_1}\mbox{Width}({C}_{\alpha,\beta^*_{\lambda}}^{\texttt{orc-lin}})\\ 
&+ \mathfrak{C}(K_x, K_y, q_y)L_{\Lambda_1}\left(\left(\frac{d + \log(N/d)}{N}\right)^{\gamma/2} + \left(\frac{dN^{2/q_x}}{N}\right)^{\frac{3q_{x}\gamma}{(4+4q_x)}} + \left(\frac{d^{1/2}N^{1/p}}{N}\right)^{\frac{q_x}{(2+2q_x)}}\right).
\end{split}
\end{equation}
% Recall that $p = q_xq_x/(q_x + q_y).$
The rate under these heavy-tailed assumptions converges to zero if $d = o(N^{1 - 2/q_x})$. This result holds true without assuming $q_x \ge 4$, we only need $q_x \ge 2$. If $q_x = 2$, then the oracle inequality yields a constant rate and \emph{fails} to imply that $\widehat{C}_{\alpha}^{\texttt{EF-ridge}}$ is asymptotically close to optimal ridge regression based fixed width prediction set. Both the inequalities above also hold true for $\widehat{C}_{\alpha}^{\texttt{VF-ridge}}$.
\section{Finite sample results related to ridge regression}\label{appsec:finite-sample-results-D-Sigma-covariance}
The following two lemmas prove deviation bounds for $\|\Sigma^{-1/2}(\widehat{\Gamma}_1 - \Gamma)\|_2$ and $\mathcal{D}_{\Sigma}$, and might be of independent interest. Recall $\mathcal{D}_{\Sigma}$ and $\widehat{\Gamma}_1$ defined in \eqref{eq:D_Sigma_def} and~\eqref{eq:gammahat_def}, respectively. In the following two lemmas, $C_{q_x}$ and $C_{q_y}$ represent constants depending only on $q_x$ and $q_y$, respectively. Proof of Lemma~\ref{lem:D_Sigma} is provided in supplementary Appendix \ref{sec:proof of Lemma D_Sigma} and the proof of Lemma~\ref{lem:covariance-bound} is given in Appendix \ref{sec:proof of Lemma covariance-bound} of the supplementary.
\begin{lemma}\label{lem:D_Sigma}
Let $n:=| \mathcal{I}_1|$ and fix any $\delta \in (0,1)$. 
\begin{itemize}
\item Under assumptions \ref{assump:DGP} and \ref{assump:subGauss},
with probability at least $1-\delta$,
\begin{equation}\label{eq:D_subgaussian}
\mathcal{D}_{\Sigma} \leq \mathfrak{C} K_{x}^{2} \left\{ \sqrt{\frac{d+\log (1 / \delta)}{ n }}+ \frac{d+\log (1 / \delta)}{ n } \right\}.
\end{equation}
{\color{black}With $\delta =\sqrt{d/n}$, the right hand side tends to zero if $d = o(n)$.}

\item Under assumptions \ref{assump:DGP} and \ref{assump:finite-X},
with probability at least $1-1/n-\delta$,
\begin{equation}\label{eq:D_moment}
\mathcal{D}_{\Sigma} \leq C_{q_x}K_x^2 \left\{  \frac{d\delta^{-2/q_x}}{n^{1-2/q_x}}  + \left(\frac{d}{n}\right)^{1-2 / q_x} \log ^{4}\left(\frac{n}{d}\right)+ \left(\frac{d}{n}\right)^{1-2/\min\{q_x,4 \} } \right\}.
\end{equation}
{\color{black}With $\delta = (d/n^{1-2/q_x})^{q_x/8}$, the right hand side tends to zero if $d = o(n^{1-2/q_x})$.}
\item Under assumptions \ref{assump:DGP} and \ref{assump:finite-X-indep},
with probability at least $1-1/n-\delta$, 
\begin{equation}\label{eq:D_structure}
\mathcal{D}_{\Sigma} \leq C_{q_x}K_x^2 \left\{ \frac{ \sqrt{d} }{n} \sqrt{2 \log \left(\frac{n}{\delta} \right) }+   \frac{(d/\delta)^{2/q_x}}{n^{1-2/q_x}}   + \left(\frac{d}{n}\right)^{1-2 / q_x} \log ^{4}\left(\frac{n}{d}\right)+ \left(\frac{d}{n}\right)^{1-2/\min\{q_x,4 \} } \right\}.
\end{equation}
{\color{black}With $\delta = d^{1/2}/n^{ \min \{q_x/4,1 \} - 1/2}$, the right hand side tends to zero if $d = o(n)$.}
\end{itemize}
\end{lemma}
It is interesting to note that $\mathcal{D}_{\Sigma}$ can converge to zero even for $d = o(n)$ and heavy-tailed covariates when assumption~\ref{assump:finite-X-indep} is satisfied. The dimension requirement here matches the one for sub-Gaussian covariates and this stems from the independence assumed in~\ref{assump:finite-X-indep}. Without such an independence, the growth rate of dimension is dampened by the heavy-tailed nature of the covariate, i.e., we require $d = o(n^{1 - 2/q_x})$.

%{\blue{Question: in Lemma 22 of your 2020 paper, why is the constant there depending on $q$, i think it's a universal constant from proposition 20 of your paper? Ans: the constant $C$ in \eqref{eq:einmah} depends on $s$ and in turn $q$. }}

\begin{lemma}\label{lem:covariance-bound}
% For $\widehat{\Gamma}:=|\mathcal{I}_1|^{-1}\sum_{i \in \mathcal{I}_1 }X_i Y_i, \Gamma:=|\mathcal{I}_1|^{-1}\sum_{i \in \mathcal{I}_1 } \E [X_i Y_i] $, $\Sigma := |\mathcal{I}_1|^{-1}\sum_{i \in \mathcal{I}_1 } \E[X_i X_i^\top]$ and 
Let $n:=| \mathcal{I}_1|$ and fix any $\delta \in (0,1)$.
\begin{itemize}

\item Under assumptions \ref{assump:DGP}, \ref{assump:finite-Y}, and \ref{assump:subGauss}, with probability at least $1-\delta$,
\begin{equation}\label{eq:gammahat_subgaussian}
 \|\Sigma^{-1/2} (\widehat{\Gamma}_1 - \Gamma ) \|_2 ~\leq~ C_{q_y} K_x K_y   \left\{  \sqrt{ \frac{d+ \log( {1}/{ \delta } )}{ n } }+  \frac{\sqrt{d}\delta^{-1/q_y}}{ n ^{1-1/q_y} } \left( \log \bigl( \frac{  n  }{\delta} \bigr) \right)^{-1/2} \right\}.
\end{equation}
{\color{black}With $\delta = \sqrt{d/n}$, the right hand side tends to 0 if $d = o(n)$}

\item Under assumptions \ref{assump:DGP}, \ref{assump:finite-Y} and \ref{assump:finite-X}, with probability at least $1-\delta$,

\begin{equation}\label{eq:gammahat_moment}
\|\Sigma^{-1/2} (\widehat{\Gamma}_1 - \Gamma ) \|_2 ~\leq~ C_q  K_x K_y \left\{\sqrt{ \frac{d +\log (1/\delta) }{n}  } + \frac{  \sqrt{d}\delta^{-(q_x+q_y)/(q_x q_y)}}{n^{1 - (q_x+q_y) / q_x q_y}}\right\} .
\end{equation}
{\red{With $p:={q_x q_y}/{(q_x +q_y)}$, and $\delta=(d/n^{2-2/p})^{p/4}$, the right hand side tends to zero if $d=o(n^{2-2/p})$.}}
\end{itemize}
\end{lemma}
Note that the dimension requirement for convergence to zero of $\|\Sigma^{-1/2}(\widehat{\Gamma}_1 - \Gamma)\|_2$ is weaker than that required for convergence of $\mathcal{D}_{\Sigma}$. In the case of heavy-tailed covariates and response, if $q_xq_y/(q_x + q_y) \ge 2$, then the dimension requirement is always $d = o(n)$. 
%%%%%%%%%%%%%%%%%%%%%%%%%%%%%%%%%%%%%%%%%%%%%%%%%%%%%%%%%%%%%%%%%%
\section{Application: Conformalized Quantile Regression}
\label{appsec:application-conformalized-quantile-regression}
In the regression context, prediction sets centered at the conditional mean estimators are popular~\citep{lei2018distribution} but are not necessarily optimal even asymptotically because they are symmetric around the estimator and cannot account for conditional skewness of the response. Citing this reason,~\cite{romano2019conformalized} proposed conditional quantile based prediction sets. Formally, $\widehat{C}_k$ in Algorithm~\ref{alg:efficiency-first-conformal} is given by
\[
\widehat{C}_{k,\beta}^{(1)} := \{(x, y)\in\mathcal{X}\times\mathbb{R}:\,y\in[\widehat{q}_{\beta}^{(k)}(x) - T_{\alpha,k,\beta}^{(1)}, \widehat{q}_{1-\beta}^{(k)}(x) + T_{\alpha,k,\beta}^{(1)}]\},
\]
where $\widehat{q}_{\beta}^{(k)}(\cdot), k\in[K]$ represent estimators of the $\beta$-th conditional quantile function and $\widehat{q}_{1-\beta}^{(k)}(\cdot), k\in[K]$ represent estimators of the $(1-\beta)$-th conditional quantile function. The different $k$ here can represent different tuning parameters in the estimation of conditional quantiles. The prediction sets $\widehat{C}_{k,\beta}^{(1)}$ consider fixed width inflation/deflation of the set $[\widehat{q}_{\beta}(x), \widehat{q}_{1-\beta}(x)]$. For a possible improvement,~\cite{kivaranovic2020adaptive} and~\cite{sesia2020comparison} proposed alternatives that consider width inflation/deflation of the set $[\widehat{q}_{\beta}(x), \widehat{q}_{1-\beta}(x)]$. These alternatives are formally given by
\begin{align*}
\widehat{C}_{k,\beta}^{(2)} ~&:=~ \left\{(x, y)\in\mathcal{X}\times\mathbb{R}:\, \max\left\{\frac{\widehat{q}_{\beta}^{(k)}(x) - y}{\widehat{q}_{1/2}^{(k)}(x) - \widehat{q}_{\beta}^{(k)}(x)},\,\frac{y - \widehat{q}_{1-\beta}^{(k)}(x)}{\widehat{q}_{1-\beta}^{(k)}(x) - \widehat{q}_{1/2}^{(k)}(x)}\right\} \le T_{\alpha,k,\beta}^{(2)}\right\},\\
\widehat{C}_{k,\beta}^{(3)} ~&:=~ \left\{(x, y)\in\mathcal{X}\times\mathbb{R}:\,\max\left\{\frac{\widehat{q}_{\beta}^{(k)}(x) - y}{\widehat{q}_{1-\beta}^{(k)}(x) - \widehat{q}_{\beta}^{(k)}(x)},\,\frac{y - \widehat{q}_{1-\beta}^{(k)}(x)}{\widehat{q}_{1-\beta}^{(k)}(x) - \widehat{q}_{\beta}^{(k)}(x)}\right\} \le T_{\alpha,k,\beta}^{(3)}\right\}.
\end{align*}
\cite{sesia2020comparison} consider the comparison of these three prediction sets theoretically and empirically. Using EFCP and VFCP, we can take $\widehat{C}_k, k\in[K]$ to be the set of all prediction sets above with a grid of $\beta$ values. If we take a grid of size $10$ for $\beta\in(0, 1/2)$ and consider quantile regression estimators based on 100 different values of tuning parameters, then we get a total of $3000$ conformal predictions to choose from. The final prediction set would be of the form $\widehat{C}_{\widehat{k},\widehat{\beta}}^{(\widehat{j})}$. Note that $\log(3000)\approx 8$,  hence it does not adversely impact the slack in coverage and width shown in Table~\ref{tab:comparison-VFCP-EFCP}. We mention that the choice of $\beta$ and the number of trees $(k)$ in quantile random forests for $\widehat{q}_{\beta}^{(k)}(\cdot)$ is also discussed in~\citet[Section 6]{gupta2019nested}.

%Unlike the density level set application, 
Note that in this setting, it is unclear how to apply cross-validation to choose tuning parameters (such as $\beta$) here. 
In the following, we will compare EFCP and VFCP in real as well as synthetic data settings. For our experiments, we use random forest quantile estimators from~\cite{meinshausen2006quantile} with 10 values for \texttt{mtry} from $d/10$ to $d$ (the dimension of covariates) and 4 values for \texttt{ntree} from 100 to 400. Finally, we use 10 values for $\beta$ as an equi-spaced sequence from $10^{-4}\alpha$ to $4\alpha$. Along with the three types of prediction sets above, we get $K = 1200$ prediction sets to choose from.

\medskip
\noindent\textbf{Real Data.} We use the six datasets from UCI repository: blog feedback, concrete strength, kernel performance, news popularity, protein structure, and superconductivity; see Appendix G of~\cite{gupta2019nested} for details about these datasets. In order to estimate the coverage and assess the variance in the obtained width, we create 100 independent versions of the datasets by drawing 1000 observations without replacement from each of the six datasets; 768 of these observations are used to construct EFCP and VFCP prediction regions, and the remaining 232 observations are used to estimate the coverage probability. This procedure is taken from Section 6 of~\cite{gupta2019nested}. 

\medskip
\noindent\textbf{Synthetic Data.} The generative model for synthetic data is similar to the one from Section 6.5 of~\cite{gupta2019nested}; also see~\cite{romano2019conformalized}. The data generating process is as follows: $X_i$ is generated from a multivariate student distribution $t_d(\nu, \Sigma)$ with $\nu = 3$ and $d \times d$ covariance matrix $\Sigma$ satisfying $\Sigma_{ij}=0.5$ if $i \neq j$ (equi correlation) and $\Sigma_{jj}=1$ with dimension $d=3$. Conditional on $X_i$, $Y_i$ is generated as 
\begin{equation}\label{eq:nonlinear_pois_fm}
 Y \sim \mathrm{Pois}\left(\sin ^{2}(X_1)+\cos^4(X_2)+0.01\right)+0.03 X_1 \epsilon_{1}+25 \mathbbm{1}_{\{u<0.01\}} \epsilon_{2},
\end{equation}
and each noise component of $\epsilon_k, k=1,2$, follows $t(\nu) \times (1+ \sqrt{ X_1^{2k}+X_2^{2k}})$ independently with $\nu=3$.

The performance of EFCP and VFCP based on conformalized quantile regression sets are presented in the two columns of Figure \ref{fig:cqr}, where the left column corresponds to the real data and the right column corresponds to the synthetic one.
\begin{figure}[!h]
    \centering
    \includegraphics[width=\textwidth]{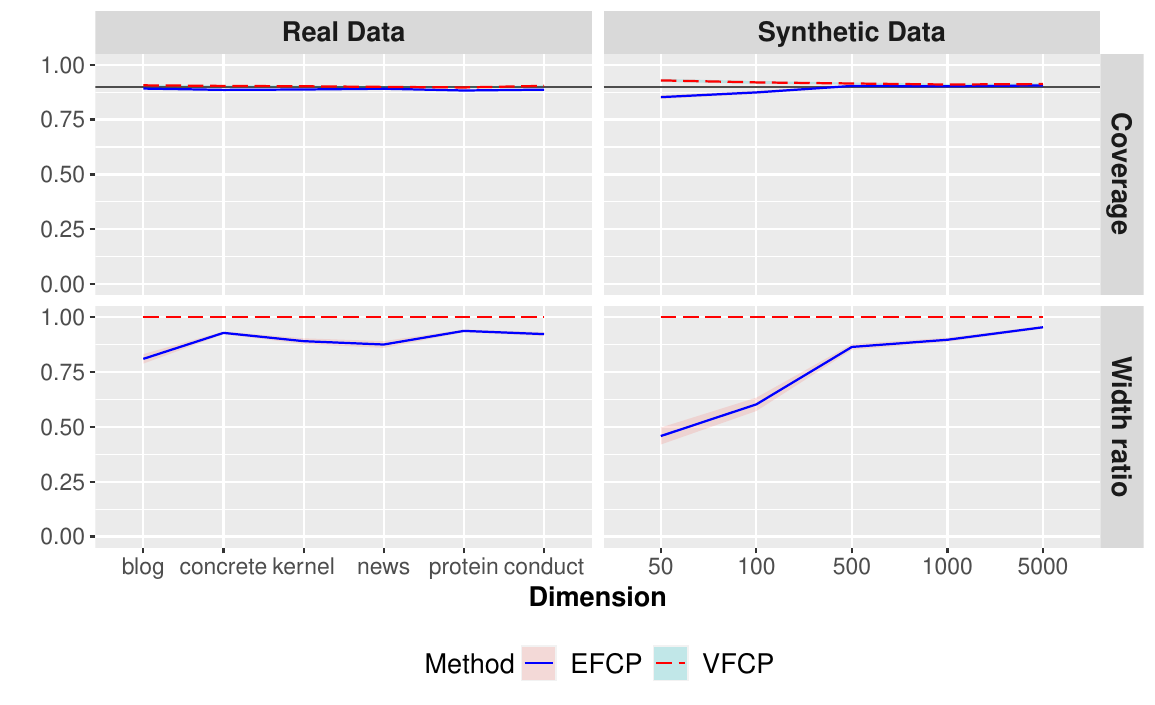}
    \caption{Coverage and width ratio of VFCP and EFCP for conformalized quantile regression on synthetic and real data. The standard deviation for both coverage and width are shown as a ribbon around the average value. Clearly, EFCP improves on VFCP in terms of width while maintaining coverage close to the nominal level of $0.9$.}
    \label{fig:cqr}
\end{figure}

Below, we note some observations from our simulations:
% \paragraph{Synthetic data}
\begin{itemize}
    \item For the real data, VFCP and EFCP has near perfect valid coverage $0.9$. For synthetic data, VFCP maintains valid coverage while EFCP slightly undercovers attaining at worst coverage of $0.85$ across all sample sizes. The undercoverage happens only at smaller sample sizes which is expected from Theorem~\ref{thm:coverage-guarantee-minimum-width}.
    \item Compared with VFCP, EFCP has an improved efficiency (smaller average width) across real datasets and sample sizes. For the real data, EFCP improves on VFCP with 7--18\% smaller widths and on average EFCP has 10\% smaller width. Notice that the improvement of efficiency of EFCP is particularly noticeable with the ``blog" and ``news" datasets, where the dimension is relatively large (280 and 59 respectively) compared to the sample size 1000. For synthetic data, EFCP results in 50\% smaller width when the sample size is small. As the sample size becomes larger, the gap closes because VFCP becomes efficient. On average, EFCP improves the width of VFCP by 25\%. 
    \item Before computing the ratio of the widths, we find that the standard error of the widths of EFCP is, on average, over 50\% smaller than that of VFCP over 100 repetitions.
\end{itemize}
%%%%%%%%%%%%%%%%%%%%%%%%%%%%%%%%%%%%%%%%%%%%%%%%%%%%%%%%%%%%%%%%%%
\section{Proof of Theorem~\ref{thm:coverage-guarantee-minimum-width}}\label{appsec:proof-coverage-guarantee-minimum-width}
Recall from step 1 of Algorithm~\ref{alg:efficiency-first-conformal} that $\mathcal{I}_1, \mathcal{I}_2$ represent the splits of $[n]$ and $\mathcal{D}_1, \mathcal{D}_2$ represent the corresponding datasets.
Recall that the nested sequences $\mathcal{F}_t^{(k)}$ are indexed by $t\in\mathcal{T}\subseteq\mathbb{R}$.
This implies that $t_k(Z_i), i\in\mathcal{I}_2$ are independent and identically distributed real-valued random variables. 
% \[
% \frac{1}{|\mathcal{I}_2|}\sum_{i\in\mathcal{I}_2}\mathbbm{1}\{Z_i\in\mathcal{F}_t^{(k)}\} ~=~ \frac{1}{|\mathcal{I}_2|}\sum_{i\in\mathcal{I}_2}\mathbbm{1}\{t_k(Z_i) \le t\}.
% \]
% Because $\mathcal{T}\subseteq\mathbb{R}$, the right-hand side above is an empirical distribution function of real-valued random variables $t_k(Z_i), i\in\mathcal{I}_2$. 
Corollary 1 of~\cite{massart1990tight} implies that for all $u \ge 0$,
% This readily follows from the fact that nested sequence of sets cannot shatter an arbitrary set of two points. Theorem 1.1 (Eq. (1.4)) with $\nu = 1$ of~\cite{talagrand1994sharper} implies that
\[
\mathbb{P}\left(\sup_{t\in\mathcal{T}}\left|\frac{1}{|\mathcal{I}_2|}\sum_{i\in\mathcal{I}_2} \mathbbm{1}\{t_k(Z_i)\le t\} - \mathbb{P}\left(t_k(Z_{N+1})\le t\big|\mathcal{D}_1, \mathcal{D}_2\right)\right| \ge \frac{u}{\sqrt{|\mathcal{I}_2|}}\,\bigg|\,\mathcal{D}_1\right) ~\le~ 2e^{-2u^2}.
\]
% Note that condition (1.2) of~\cite{talagrand1994sharper} holds true because $\mathcal{F}_t^{(k)}, t\in\mathcal{T}$ has VC dimension $1$ and the discussion on page 30 of~\cite{talagrand1994sharper}. 
Applying the union bound over $k\in[K]$ implies that for all $u \ge 0$,
\begin{equation}\label{eq:talagrand-union-bound}
\mathbb{P}\left(\max_{k\in[K]}\sup_{t\in\mathcal{T}}\left|\frac{1}{|\mathcal{I}_2|}\sum_{i\in\mathcal{I}_2} \mathbbm{1}\{t_k(Z_i)\le t\} - \mathbb{P}\left(t_k(Z_{N+1})\le t\big|\mathcal{D}_1, \mathcal{D}_2\right)\right| \ge \frac{u}{\sqrt{|\mathcal{I}_2|}}\,\bigg|\,\mathcal{D}_1\right) ~\le~ 2Ke^{-2u^2}.
\end{equation}
For notational convenience, set
\[
W := \sqrt{|\mathcal{I}_2|}\max_{k\in[K]}\sup_{t\in\mathcal{T}}\left|\frac{1}{|\mathcal{I}_2|}\sum_{i\in\mathcal{I}_2} \mathbbm{1}\{t_k(Z_i)\le t\} - \mathbb{P}\left(t_k(Z_{N+1})\le t\big|\mathcal{D}_1,\mathcal{D}_2\right)\right|.
\]
It is clear that
\[
\frac{1}{|\mathcal{I}_2|}\sum_{i\in\mathcal{I}_2} \mathbbm{1}\{t_{\widehat{k}}(Z_i)\le T_{\alpha,\widehat{k}}\} - \mathbb{P}\left(t_{\widehat{k}}(Z_{N+1})\le T_{\alpha,\widehat{k}}\big|\mathcal{D}_1,\mathcal{D}_2\right) ~\le~ \frac{W}{\sqrt{|\mathcal{I}_2|}}.
\]
This implies that
\begin{equation}\label{eq:simple-implication}
    \mathbb{P}\left(t_{\widehat{k}}(Z_{N+1})\le T_{\alpha,\widehat{k}}\big|\mathcal{D}_1,\mathcal{D}_2\right) ~\ge~ \frac{1}{|\mathcal{I}_2|}\sum_{i\in\mathcal{I}_2} \mathbbm{1}\{t_{\widehat{k}}(Z_i)\le T_{\alpha,\widehat{k}}\} - \frac{W}{\sqrt{|\mathcal{I}_2|}}.
\end{equation}
Note that $\{t_{\widehat{k}}(Z_{N+1})\le T_{\alpha,\widehat{k}}\} = \{Z_{N+1}\in\widehat{C}_{\widehat{k}}\}$ and $(\widehat{k}, T_{\alpha,\widehat{k}})$ are measurable with respect to $\mathcal{D}_2$ so that the probability on the left hand side only integrates the randomness of $Z_{N+1}$.
Recall that by definition of $T_{\alpha,k}$ in step 3 of Algorithm~\ref{alg:efficiency-first-conformal}, we conclude that for all $k\in[K]$,
\begin{equation}\label{eq:empirical-lower-bound-frequency}
\frac{1}{|\mathcal{I}_2|}\sum_{i\in\mathcal{I}_2}\mathbbm{1}\{t_k(Z_i)\le T_{\alpha,k}\} ~=~ \frac{\lceil(1 + |\mathcal{I}_2|)(1 - \alpha)\rceil}{|\mathcal{I}_2|} ~\ge~ \left(1 + \frac{1}{|\mathcal{I}_2|}\right)(1 - \alpha).
\end{equation}
In particular, this holds for $\widehat{k}$. This implies that with probability 1,
\begin{equation}\label{eq:deterministic-union-bound}
\begin{split}
    \mathbb{P}\left(t_{\widehat{k}}(Z_{N+1})\le T_{\alpha,\widehat{k}}\big|\mathcal{D}_1, \mathcal{D}_2\right) ~&\ge~ \left(1 + \frac{1}{|\mathcal{I}_2|}\right)(1 - \alpha) - \frac{W}{\sqrt{|\mathcal{I}_2|}}.
\end{split}
\end{equation}
Using~\eqref{eq:talagrand-union-bound}, we can bound $\mathbb{E}[W|\mathcal{D}_1]$ as
\begin{equation}\label{eq:expectation-maximum-bound}
\begin{split}
% &\mathbb{E}\left[\sqrt{|\mathcal{I}_2|}\max_{k\in[K]}\sup_{t\in\mathcal{T}}\left\{\frac{1}{|\mathcal{I}_2|}\sum_{i\in\mathcal{I}_2} \mathbbm{1}\{t_k(Z_i)\le t\} - \mathbb{P}\left(t_k(Z_{N+1})\le t\big|\mathcal{D}_1,\mathcal{D}_2\right)\right\}\right]
% \\
\mathbb{E}[W|\mathcal{D}_1] &= \int_0^{\infty} \mathbb{P}(W \ge u|\mathcal{D}_1)du = \int_0^{\sqrt{\log(2K)/2}} \mathbb{P}(W \ge u|\mathcal{D}_1)du + \int_{\sqrt{\log(2K)/2}}^{\infty} \mathbb{P}(W \ge u|\mathcal{D}_1)du\\
&\le \sqrt{\log(2K)/2} + 2K\int_{\sqrt{\log(2K)/2}}^{\infty} e^{-2u^2}du = \sqrt{\log(2K)/2} + \frac{2K}{\sqrt{2}}\int_{\sqrt{\log(2K)}}^{\infty} e^{-t^2}dt\\
&\overset{(a)}{\le} \sqrt{\log(2K)/2} + \frac{\sqrt{2}K\exp(-\log(2K))}{\sqrt{\log(2K)} + \sqrt{\log(2K) + 4/\pi}} \le \sqrt{\log(2K)/2} + 1/3.
\end{split}
\end{equation}
Inequality (a) follows from Formula 7.1.13 of~\cite{abramowitz1964handbook}.
Taking expectations on both sides of~\eqref{eq:deterministic-union-bound} and using~\eqref{eq:expectation-maximum-bound}, we obtain
\[
\mathbb{P}\left(Z_{N+1}\in\widehat{C}_{\widehat{k}}\big|\mathcal{D}_1\right) \ge \left(1 + \frac{1}{|\mathcal{I}_2|}\right)(1 - \alpha) - \frac{\sqrt{\log(2K)/2} + 1/3}{\sqrt{|\mathcal{I}_2|}}.
\]
This implies~\eqref{eq:probability-unconditional-guarantee}.
Further, combining inequality~\eqref{eq:talagrand-union-bound} with $u = \sqrt{\log(2K/\delta)/2}$ with~\eqref{eq:deterministic-union-bound} yields~\eqref{eq:probability-conditional-guarantee}.

\section{Proof of Theorem~\ref{thm:multiple-splits-conformal}}\label{appsec:proof-of-thm-multiple-splits}
{
\begin{proof}
For any set $C$, let $\mu(C) = \mathbb{P}(Z_{N+1}\in C)$. For a random set $\widehat{C}$, $\mu(\widehat{C})$ is a random variable. The split conformal prediction set satisfies (by Corollary 1 of~\cite{massart1990tight})
\[
\mathbb{P}\left(\mu(\widehat{C}_k) \ge \left(1 + \frac{2}{N}\right)(1 - \alpha) - \frac{u}{\sqrt{N/2}}\right) \le 2e^{-2u^2}\quad\mbox{for all}\quad u\ge0.
\]
Now applying the union bound, we get
\[
\mathbb{P}\left(\min_{1\le k\le K}\mu(\widehat{C}_k) \ge \left(1 + \frac{2}{N}\right)(1 - \alpha) - \frac{u}{\sqrt{N/2}}\right) \le 2Ke^{-2u^2}\quad\mbox{for all}\quad u\ge0.
\]
This implies that
\[
\mathbb{E}\left[\min_{1\le k\le K}\mu(\widehat{C}_k)\right] \ge \left(1 + \frac{2}{N}\right)(1 - \alpha) - \frac{\sqrt{\log(2K)/2} + 1/3}{N/2}.
\]
(See the proof of Theorem 1 in Section~\ref{appsec:proof-coverage-guarantee-minimum-width} for details on converting the probability tail bound to expectation bound.) Now $\widehat{k}\in\{1, 2, \ldots, K\}$, and hence,
\[
\mathbb{P}(Z_{N+1}\in\widehat{C}_{\widehat{k}}) \ge \mathbb{E}[\mu(\widehat{C}_{\widehat{k}})] \ge \mathbb{E}\left[\min_{1\le k\le K}\mu(\widehat{C}_k)\right] \ge \left(1 + \frac{2}{N}\right)(1 - \alpha) - \frac{\sqrt{\log(2K)/2} + 1/3}{N/2}.
\]
This completes the proof.
\end{proof}
\section{Proof of Conditional Coverage of CQR-EFCP}\label{appsec:proof-of-conditional-coverage-EFCP-CQR}
Observe that
\begin{align*}
&\mathbb{P}\left((X_{N+1}, Y_{N+1})\in\widehat{C}_{N,\alpha}^{\mathrm{EFCP}}\bigg|X_{N+1} = x, \mathcal{D}\right)\\
&= \mathbb{P}\left(Y_{N+1} \le q^*_{1-\alpha/2}(x) + \{\widehat{q}_{1-\alpha/2}^{(\widehat{k})}(x) - q^*_{1-\alpha/2}(x)\} + T_{\alpha,\widehat{k}}\bigg|X_{N+1} = x, \mathcal{D}\right)\\
&\quad- \mathbb{P}\left(Y_{N+1} \le q^*_{\alpha/2}(x) + \{\widehat{q}_{\alpha/2}^{(\widehat{k})}(x) - q^*_{\alpha/2}(x)\} -T_{\alpha,\widehat{k}}\bigg|X_{N+1} = x, \mathcal{D}\right).
\end{align*}
Using the assumption of bounded conditional density of $Y$ given $X$, we get
\begin{equation}\label{eq:error-conditional-coverage}
\begin{split}
&\sup_{x}\left|\mathbb{P}\left((X_{N+1}, Y_{N+1})\in\widehat{C}_{N,\alpha}^{\mathrm{EFCP}}\bigg|X_{N+1} = x, \mathcal{D}\right) - (1 - \alpha)\right|\\ 
&\quad\le M\left\{\max_{1\le k\le K}\|\widehat{q}_{\alpha/2}^{(k)} - q^*_{\alpha/2}\|_{\infty} + \max_{1\le k\le K}\|\widehat{q}_{1-\alpha/2}^{(k)} - q^*_{1-\alpha/2}\|_{\infty} + 2\max_{1\le k\le K}|T_{\alpha,k}|\right\}.
\end{split}
\end{equation}
Hence, to show that the EFCP prediction set has asymptotic conditional coverage, it suffices to show that $\max_{1\le k\le K}|T_{\alpha,k}| = o_p(1)$ under the assumption that the quantile estimators are uniformly consistent. Following the same argument as above and using the DKW inequality as in the proof of Theorem~\ref{thm:coverage-guarantee-minimum-width}, it follows that conditional on $\mathcal{D}_1$, with probability at least $1 - \delta$,
\begin{align*}
&\max_{1\le k\le K}\sup_{t\in\mathbb{R}}\left|\frac{1}{|\mathcal{I}_2|}\sum_{i\in\mathcal{I}_2} \mathbbm{1}\{t_k(X_i, Y_i) \le t\} - F(t)\right|\\ &\le \sqrt{\frac{\log(2K/\delta)}{2|\mathcal{I}_2|}} + \max_{1\le k\le K}\|\widehat{q}_{\alpha/2}^{(k)} - q^*_{\alpha/2}\|_{\infty} + \max_{1\le k\le K}\|\widehat{q}_{1-\alpha/2}^{(k)} - q^*_{1-\alpha/2}\|_{\infty},
\end{align*}
where $F(t) = \mathbb{P}(t^*(X, Y) \le t)$ with $t^*(x, y) := \max\{\widehat{q}_{\alpha/2}^{*}(x) - y, y - \widehat{q}_{1-\alpha/2}^*(x)\}$.
Taking $t = T_{\alpha,k}$, we get
\[
\max_{1\le k\le K}|(1 - \alpha) - F(T_{\alpha,k})| \le \sqrt{\frac{\log(2K/\delta)}{2|\mathcal{I}_2|}} + \max_{1\le k\le K}\|\widehat{q}_{\alpha/2}^{(k)} - q^*_{\alpha/2}\|_{\infty} + \max_{1\le k\le K}\|\widehat{q}_{1-\alpha/2}^{(k)} - q^*_{1-\alpha/2}\|_{\infty}.
\]
Under the assumption that the Lebesgue density of $t^*(X, Y)$ is bounded away from zero in a neighborhood of $0$ by $r$, this implies that
\[
\max_{1\le k\le K}|T_{\alpha,k}| \le \frac{1}{r}\left[\sqrt{\frac{\log(2K/\delta)}{2|\mathcal{I}_2|}} + \max_{1\le k\le K}\|\widehat{q}_{\alpha/2}^{(k)} - q^*_{\alpha/2}\|_{\infty} + \max_{1\le k\le K}\|\widehat{q}_{1-\alpha/2}^{(k)} - q^*_{1-\alpha/2}\|_{\infty}\right].
\]
Combining this with~\eqref{eq:error-conditional-coverage}, we conclude that with probability at least $1 - \delta$,
\begin{align*}
&\sup_{x}\left|\mathbb{P}\left((X_{N+1}, Y_{N+1})\in\widehat{C}_{N,\alpha}^{\mathrm{EFCP}}\bigg|X_{N+1} = x, \mathcal{D}\right) - (1 - \alpha)\right|\\ 
&\quad\le M(1 + 1/r)\left[\max_{1\le k\le K}\|\widehat{q}_{\alpha/2}^{(k)} - q^*_{\alpha/2}\|_{\infty} + \max_{1\le k\le K}\|\widehat{q}_{1-\alpha/2}^{(k)} - q^*_{1-\alpha/2}\|_{\infty}\right] + \frac{1}{r}\sqrt{\frac{\log(2K/\delta)}{2|\mathcal{I}_2|}}.
\end{align*}
Hence, the result follows using~\eqref{eq:uniform-quantile-consistency}.
}
\section{Proof of Theorem~\ref{thm:empirical-oracle-inequality-general-VFCP}}\label{appsec:proof-thm-empirical-oracle-inequality-general-VFCP}
The continuity of $t\mapsto\mathcal{F}_t^{(k)}$ implies that
\[
\widehat{C}_k = \mathcal{F}_{T_{\alpha,k}}^{(k)},\quad\mbox{and}\quad \widehat{C}_k^* = \mathcal{F}_{T^*_{\alpha,k}}^{(k)}\quad\Rightarrow\quad \widehat{C}_{N,\alpha}^{\mathrm{EFCP}} = \mathcal{F}_{T_{\alpha,\widehat{k}}}^{(\widehat{k})}.
\]
By definition of $\widehat{k}$ and assumption~\ref{eq:width-continuity},
\begin{align*}
\mbox{Width}\left(\widehat{C}_{N,\alpha}^{\mathrm{VFCP}}\right) = \mbox{Width}\left(\mathcal{F}_{T_{\alpha,\widehat{k}}^*}^{(\widehat{k})}\right) &\le \mbox{Width}\left(\mathcal{F}_{T_{\alpha,\widehat{k}}}^{(\widehat{k})}\right) + L_W|T_{\alpha,\widehat{k}}^* - T_{\alpha,\widehat{k}}|\\
&= \min_{1\le k\le K}\mbox{Width}\left(\mathcal{F}_{T_{\alpha,k}}^{(k)}\right) + L_W|T_{\alpha,\widehat{k}}^* - T_{\alpha,\widehat{k}}|\\
&= \min_{1\le k\le K}\mbox{Width}\left(\widehat{C}_k\right) + L_W|T_{\alpha,\widehat{k}}^* - T_{\alpha,\widehat{k}}|.
\end{align*}
Furthermore, $$\mbox{Width}(\mathcal{F}_{T_{\alpha,k}}^{(k)}) \le \mbox{Width}(\mathcal{F}_{\widehat{Q}_{\alpha,k}}^{(k)}) + L_W|T_{\alpha,k} - \widehat{Q}_{\alpha,k}|.$$
Therefore,
\begin{equation}\label{eq:penultimate-general-oracle-VFCP}
    \begin{split}
        \mbox{Width}\left(\widehat{C}_{N,\alpha}^{\mathrm{VFCP}}\right) &\le \min_{1\le k\le K}\mbox{Width}\left(\widehat{C}_k\right) + L_W\max_{k\in[K]}|T_{\alpha,k}^* - T_{\alpha,k}|\\
        &\le \min_{1\le k\le K}\mbox{Width}(\mathcal{F}_{\widehat{Q}_{\alpha,k}}^{(k)}) + L_W\max_{k\in[K]}|T_{\alpha,k}^* - T_{\alpha,k}| + L_W\max_{k\in[K]}|T_{\alpha,k} - \widehat{Q}_{\alpha,k}|.
    \end{split}
\end{equation}
Hence, to prove the result it suffices to bound $|T^*_{\alpha,k} - \widehat{Q}_{\alpha,k}|$ and $|T_{\alpha,k} - \widehat{Q}_{\alpha,k}|$ for all $k\in[K]$. Recall that $\widehat{Q}_k = F_k^{-1}(1-\alpha)$ and that $T_{\alpha,k}$ is the $\widehat{F}_{k,2}^{-1}(\lceil(1-\alpha)(1 + |\mathcal{I}_2|)\rceil/|\mathcal{I}_2|)$, where $\widehat{F}_{k,2}(s) = |\mathcal{I}_2|^{-1}\sum_{i\in\mathcal{I}_2} \mathbbm{1}\{t_k(Z_i) \le s\}$. Similarly, $\widehat{F}_{k,3}(\cdot)$ and $T_{\alpha,k}^*$ can be defined. By Corollary 1 of~\cite{massart1990tight}, it follows that
\[
\mathbb{P}\left(\sup_{s\in\mathbb{R}}|\widehat{F}_{k,2}(s) - F_k(s)| \ge \sqrt{\frac{\log(2/\delta)}{2|\mathcal{I}_2|}}\,\bigg|\,\mathcal{D}_1\right) \le \delta.
\]
Employing a union bound, we obtain
\begin{equation}\label{eq:split-2-ecdf-massart}
\mathbb{P}\left(\max_{k\in[K]}\sup_{s\in\mathbb{R}}|\widehat{F}_{k,2}(s) - F_k(s)| \ge \sqrt{\frac{\log(2K/\delta)}{2|\mathcal{I}_2|}}\,\bigg|\,\mathcal{D}_1\right) \le \delta.
\end{equation}
Similarly,
\begin{equation}\label{eq:split-3-ecdf-massart}
\mathbb{P}\left(\max_{k\in[K]}\sup_{s\in\mathbb{R}}|\widehat{F}_{k,3}(s) - F_k(s)| \ge \sqrt{\frac{\log(2K/\delta)}{2|\mathcal{I}_3|}}\,\bigg|\,\mathcal{D}_1\right) \le \delta.
\end{equation}
Combining~\eqref{eq:split-2-ecdf-massart},~\eqref{eq:split-3-ecdf-massart} with Proposition~\ref{prop:holder}, we conclude with probability at least $1 - \delta$, the following inequalities hold true simultaneously:
\begin{equation}\label{eq:T-T-star-VFCP}
\begin{split}
\left|T_{\alpha,k} - F_k^{-1}\left(\frac{\lceil(1-\alpha)(1+|\mathcal{I}_2|)\rceil}{|\mathcal{I}_2|}\right)\right| &\le L_k\left(\frac{\log(4K/\delta)}{|\mathcal{I}_2|}\right)^{\gamma/2},\quad\mbox{for all}\quad k\in[K],\\
\left|T_{\alpha,k}^* - F_k^{-1}\left(\frac{\lceil(1-\alpha)(1+|\mathcal{I}_3|)\rceil}{|\mathcal{I}_3|}\right)\right| &\le L_k\left(\frac{\log(4K/\delta)}{|\mathcal{I}_3|}\right)^{\gamma/2},\quad\mbox{for all}\quad k\in[K].\\
\end{split}
\end{equation}
Furthermore, an application of~\ref{eq:Holder-VFCP} yields
\begin{equation}\label{eq:conformal-to-non-conformal-VFCP}
\left|\widehat{Q}_{\alpha,k} - F_k^{-1}\left(\frac{\lceil(1-\alpha)(1+|\mathcal{I}_2|)\rceil}{|\mathcal{I}_2|}\right)\right| \le L_k\left(\frac{2}{|\mathcal{I}_2|}\right)^{\gamma},\quad\mbox{for all}\quad k\in[K].
\end{equation}
Similar inequality holds with $\mathcal{I}_3$. Combining~\eqref{eq:T-T-star-VFCP} and~\eqref{eq:conformal-to-non-conformal-VFCP} with triangle inequality in~\eqref{eq:penultimate-general-oracle-VFCP}, we obtain the result. 

{
\section{Proof of Proposition~\ref{prop:validity-EFCP-fixed-width}}\label{appsec:proof-of-prop-validity-EFCP-fixed-width}
Observe that for $\widehat{C}_{\alpha}^{\texttt{EF-lin}}$, $\mathcal{D}_1 = \mathcal{D} = \{(X_i, Y_i), 1\le i\le n\}$. Hence, 
\begin{align*}
\mathbb{P}\left((X_{N+1}, Y_{N+1})\in\widehat{C}_{\alpha}^{\texttt{EF-lin}}\big|\mathcal{D}_1\right) &= F_{\widehat{\theta}}(T_{\alpha,\widehat{\theta}})\\ 
&\ge \frac{1}{N}\sum_{i=1}^n \mathbbm{1}\{|Y_i - X_i^{\top}\widehat{\theta}| \le T_{\alpha,\widehat{\theta}}\} - R([N])\\
&\ge \frac{\lceil(1 + N)(1 - \alpha)\rceil}{N} - R([N]), 
\end{align*}
which proves the first statement; the inequalities here follow from the definition of $T_{\alpha,\widehat{\theta}}$ and that of $R([N])$. The second statement about $R([N])$ follows from Proposition~\ref{prop:R1-bound}.
% \end{proof}
}
\section{Proof of Theorems~\ref{thm:general-oracle-final} and~\ref{thm:general-oracle-VFCP}}\label{sec:Proof of Theorem general-oracle-final}
\subsection{Proof of Theorem~\ref{thm:general-oracle-final}}
Recall that
\[
\mathrm{Width}(\widehat{C}_{\alpha}^{\texttt{EF-lin}}) = 2T_{\alpha, \widehat{\theta}} = 2\min_{\theta\in\Theta}T_{\alpha,\theta} \le 2T_{\alpha,\theta^*}.
\]
It suffices to bound $T_{\alpha,\theta^*}$ in terms of the $Q_{\alpha,\theta^*}$. By DKW inequality, we have with probability at least $1 - \delta$,
\[
\sup_{t\ge0}\left|\frac{1}{N}\sum_{i=1}^N \mathbbm{1}\{|Y_i - X_i^{\top}\theta^*| \le t\} - F_{\theta^*}(t)\right| ~\le~ \sqrt{\frac{\log(2/\delta)}{2N}}.
\]
Taking $t = F_{\theta^*}^{-1}(1 - \alpha + \sqrt{\log(2/\delta)/(2N)})$, we get
\[
T_{\alpha,\theta^*} \le F_{\theta^*}^{-1}(1 - \alpha + \sqrt{\log(2/\delta)/(2N)}).
\]
Now, assumption~\ref{assump:general-holder-EFCP} implies
\[
T_{\alpha,\theta^*} \le F_{\theta^*}^{-1}(1 - \alpha) + L_{\theta^*}\left(\frac{\log(2/\delta)}{2N}\right)^{\gamma/2} = Q_{\alpha,\theta^*} + L_{\theta^*}\left(\frac{\log(2/\delta)}{2N}\right)^{\gamma/2}.
\]
The definition of $\theta^*$ implies the result.
\subsection{Proof of Theorem~\ref{thm:general-oracle-VFCP}}
We start off by defining some quantiles:
\begin{equation}\label{def:general-q}
\begin{cases}
T_{\alpha,\theta} ~:=~ \lceil(1 + |\mathcal{I}_1|)(1 - \alpha)\rceil\mbox{-th quantile of } |\mathcal{I}_1|^{-1}\sum_{i \in \mathcal{I}_1} \mathbbm{1}{ \{ |Y_i - \theta^\top X_i| \leq t \}};\\
T_{\alpha,\theta}^* ~:=~ \lceil(1 + |\mathcal{I}_2|)(1 - \alpha)\rceil\mbox{-th quantile of } |\mathcal{I}_2| ^{-1}\sum_{i \in \mathcal{I}_2} \mathbbm{1}{ \{ |Y_i - \theta^\top X_i| \leq t \}};\\
Q_{\alpha,\theta} ~:=~ (1 - \alpha)\mbox{-th quantile of } F_{\theta} (t).
\end{cases}
\end{equation}
By definition of $R(\mathcal{I}_1)$ and $R(\mathcal{I}_2)$, we have that
\[
\sup_{\theta\in\mathbb{R}^d}\sup_{t\ge0}\left|\frac{1}{|\mathcal{I}_1|}\sum_{i\in\mathcal{I}_1}\mathbbm{1}\{|Y_i - X_i^{\top}\theta| \le t\} - F_{\theta}(t)\right| \le R(\mathcal{I}_1).
\]
Similar inequality holds true for $\mathcal{I}_2$.
% where for the last distribution the randomness of $(X_i, Y_i)$ are all accounted for. 
Applying Proposition~\ref{prop:holder}, we conclude that
\[
\left|T_{\alpha,\theta} - F_{\theta}^{-1}\left(\frac{\lceil(1 + |\mathcal{I}_1|)(1 - \alpha)\rceil}{|\mathcal{I}_1|}\right)\right| \le L_{\theta}R(\mathcal{I}_1)^{\gamma}\quad\mbox{for all}\quad \theta\in\mathbb{R}^d. 
\]
Because $Q_{\alpha,\theta} = F_{\theta}^{-1}(1 - \alpha)$, we get under~\ref{assump:general-holder-VFCP},
\[
\left|Q_{\alpha,\theta} - F_{\theta}^{-1}\left(\frac{\lceil(1 + |\mathcal{I}_1|)(1 - \alpha)\rceil}{|\mathcal{I}_1|}\right)\right| \le L_{\theta}\left(\frac{2}{|\mathcal{I}_1|}\right)^{\gamma}.
\]
Combining both inequalities above, we obtain
\begin{equation}\label{eq:I1-quantile-bound}
\left|T_{\alpha,\theta} - Q_{\alpha,\theta}\right| ~\le~ L_{\theta}\left[R(\mathcal{I}_1)^{\gamma} + \left(\frac{2}{|\mathcal{I}_1|}\right)^{\gamma}\right],\quad\mbox{for all}\quad \theta\in\mathbb{R}^d.
\end{equation}
Similarly, 
\begin{equation}\label{eq:I2-quantile-bound}
\left|T_{\alpha,\theta}^* - Q_{\alpha,\theta}\right| ~\le~ L_{\theta}\left[R(\mathcal{I}_2)^{\gamma} + \left(\frac{2}{|\mathcal{I}_2|}\right)^{\gamma}\right],\quad\mbox{for all}\quad \theta\in\mathbb{R}^d.
\end{equation}
Recall that $\widehat{\theta}$ in Algorithm~\ref{alg:general-conformal} satisfies
\[
T_{\alpha,\widehat{\theta}} = \min_{\theta\in\Theta}T_{\alpha,\theta},
\]
which coupled with~\eqref{eq:I1-quantile-bound} implies that
\begin{equation}\label{eq:second-last-inequality-empirical-oracle}
\begin{split}
T_{\alpha,\widehat{\theta}} ~=~ \min_{\theta\in\Theta}T_{\alpha,\theta} ~&\le~ \min_{\theta\in\Theta}\left\{Q_{\alpha,\theta} + L_{\theta}\left[R(\mathcal{I}_1)^{\gamma} + \left(\frac{2}{|\mathcal{I}_1|}\right)^{\gamma}\right]\right\}\\ ~&\le~ \min_{\theta\in\Theta}Q_{\alpha,\theta} + L_{\Theta}\left[R(\mathcal{I}_1)^{\gamma} + \left(\frac{2}{|\mathcal{I}_1|}\right)^{\gamma}\right].
\end{split}
\end{equation}
Inequality~\eqref{eq:I1-quantile-bound} also implies that
\[
Q_{\alpha,\widehat{\theta}} ~\le~ T_{\alpha,\widehat{\theta}} +  L_{\widehat{\theta}}\left[R(\mathcal{I}_1)^{\gamma} + \left(\frac{2}{|\mathcal{I}_1|}\right)^{\gamma}\right] ~\le~ T_{\alpha,\widehat{\theta}} +  L_{\Theta}\left[R(\mathcal{I}_1)^{\gamma} + \left(\frac{2}{|\mathcal{I}_1|}\right)^{\gamma}\right]. 
\]
Combining the inequalities above with~\eqref{eq:I2-quantile-bound} for $\theta = \widehat{\theta}$ yields
\begin{align*}
T_{\alpha,\widehat{\theta}}^* &\le~ Q_{\alpha,\widehat{\theta}} + L_{\widehat{\theta}}\left[R(\mathcal{I}_2)^{\gamma} + \left(\frac{2}{|\mathcal{I}_2|}\right)^{\gamma}\right]\\ 
&\le~ Q_{\alpha,\widehat{\theta}} + L_{\Theta}\left[R(\mathcal{I}_2)^{\gamma} + \left(\frac{2}{|\mathcal{I}_2|}\right)^{\gamma}\right]\\
&\le~ T_{\alpha,\widehat{\theta}} ~+~  L_{\Theta}\left[R(\mathcal{I}_1)^{\gamma} + \left(\frac{2}{|\mathcal{I}_1|}\right)^{\gamma}\right] + L_{\Theta}\left[R(\mathcal{I}_2)^{\gamma} + \left(\frac{2}{|\mathcal{I}_2|}\right)^{\gamma}\right]\\
&\le~ \min_{\theta\in\Theta}Q_{\alpha,\theta} + 2L_{\Theta}\left[R(\mathcal{I}_1)^{\gamma} + \left(\frac{2}{|\mathcal{I}_1|}\right)^{\gamma}\right] + L_{\Theta}\left[R(\mathcal{I}_2)^{\gamma} + \left(\frac{2}{|\mathcal{I}_2|}\right)^{\gamma}\right].
\end{align*}
Hence, inequality~\eqref{eq:general-final_quantile_strong} holds true. The penultimate inequality here also proves that the oracle inequality holds true with $Q_{\alpha,\theta}$ replaced by $T_{\alpha,\theta}$. 
\section{Proof of Theorem~\ref{thm:width-efficiency-conditional-coverage}}\label{appsec:proof-of-width-efficiency}
Recall that 
\[
\mbox{Width}(\widehat{C}_{\alpha}^{\texttt{agg}}) = 2T_{\alpha,\widehat{\theta}},
\]
and
by the definition of $T_{\alpha,\widehat{\theta}}$, we know
\[
T_{\alpha,\widehat{\theta}} = \min_{\theta\in\Theta}T_{\alpha,\theta}.
\]
Let $\theta^*\in\Theta$ be any vector such that
\[
\|\widehat{\mu}_{\theta^*} - \mu_0\|_2^2 \le \inf_{\theta\in\Theta}\|\widehat{\mu}_{\theta} - \mu_0\|_2^2 + \frac{1}{|\mathcal{I}_2|}.
\]
By definition of $T_{\alpha,\widehat{\theta}}$, we know
\[
T_{\alpha,\widehat{\theta}} = \min_{\theta\in\Theta}T_{\alpha,\theta} \le T_{\alpha,\theta^*}.
\]
Note that $2T_{\alpha,\theta^*}$ is the width of the split conformal prediction set computed with $\widehat{\mu}_{\theta^*}(\cdot)$ as the non-parametric regression method. By DKW inequality, we have
\[
\mathbb{P}\left(\sup_{t\ge0}\left|\frac{1}{|\mathcal{I}_2|}\sum_{i\in\mathcal{I}_2}\mathbbm{1}\{|Y_i - \widehat{\mu}_{\theta^*}(X_i)| \le t\} - \mathbb{P}(|Y - \widehat{\mu}_{\theta^*}(X)| \le t|\mathcal{D}_1)\right| \ge \sqrt{\frac{\log(2/\delta)}{2|\mathcal{I}_2|}}\bigg|\mathcal{D}_1\right) \le \delta.
\]
Following the proof of Theorem 3.1 of~\cite{lei2018distribution}, we obtain
\begin{align*}
&\sup_{t\ge0}\left|\mathbb{P}(|Y - \widehat{\mu}_{\theta^*}(X)| \le t|\mathcal{D}_1) - (2F(t) - 1)\right| \le (M/2)\mathbb{E}\left[(\widehat{\mu}_{\theta^*}(X) - \mu_0(X))^2|\mathcal{D}_1\right].
\end{align*}
Therefore, with probability at least $1 - \delta$ (conditional on $\mathcal{D}_1$),
\begin{equation}\label{eq:uniform-consistency-aggregation}
\sup_{t\ge0}\left|\frac{1}{|\mathcal{I}_2|}\sum_{i\in\mathcal{I}_2}\mathbbm{1}\{|Y_i - \widehat{\mu}_{\theta^*}(X_i)| \le t\} - (2F(t) - 1)\right| \le (M/2)\mathbb{E}[(\widehat{\mu}_{\theta^*}(X) - \mu_0(X))^2|\mathcal{D}_1] + \sqrt{\frac{\log(2/\delta)}{2|\mathcal{I}_2|}}.
\end{equation}
Now following the steps in the proof of Theorem~\ref{thm:general-oracle-final}, proves~\eqref{eq:oracle-width-efficiency-aggregation}.

Observe that the conditional coverage under~\ref{assump:symmetric-error-dist} and~\ref{assump:lower-bound-differentiable} satisfies
\begin{align*}
\mathrm{CC}_{\alpha} &=\mathbb{P}\left((X_{N+1}, Y_{N+1})\in\widehat{C}_{\alpha}^{\texttt{agg}}\big|X_{N+1} = x, \mathcal{D}_1\right)\\ 
&= \mathbb{P}\left(\sum_{\ell = 1}^L \widehat{\theta}_{\ell}\widehat{\mu}_{\ell}(x) - T_{\alpha,\widehat{\theta}} \le Y_{N+1} \le \sum_{\ell = 1}^L \widehat{\theta}_{\ell}\widehat{\mu}_{\ell}(x) + T_{\alpha,\widehat{\theta}}\bigg|X_{N+1} = x, \mathcal{D}_1\right)\\
&= \mathbb{P}\left(\sum_{\ell = 1}^L \widehat{\theta}_{\ell}(\widehat{\mu}_{\ell}(x) - \mu_0(x)) - T_{\alpha,\widehat{\theta}} \le \xi_{N+1} \le \sum_{\ell = 1}^L \widehat{\theta}_{\ell}(\widehat{\mu}_{\ell}(x) - \mu_0(x)) + T_{\alpha,\widehat{\theta}}\bigg|\mathcal{D}_1\right)\\
&= F\left(\sum_{\ell = 1}^L \widehat{\theta}_{\ell}(\widehat{\mu}_{\ell}(x) - \mu_0(x)) + T_{\alpha,\widehat{\theta}}\right) - F\left(\sum_{\ell = 1}^L \widehat{\theta}_{\ell}(\widehat{\mu}_{\ell}(x) - \mu_0(x)) - T_{\alpha,\widehat{\theta}}\right).
\end{align*}
Under~\ref{assump:lower-bound-differentiable}, 
\begin{align*}
\left|F\left(\sum_{\ell = 1}^L \widehat{\theta}_{\ell}(\widehat{\mu}_{\ell}(x) - \mu_0(x)) \pm T_{\alpha,\widehat{\theta}}\right) - F(\pm T_{\alpha,\widehat{\theta}})\right| &\le M\left|\sum_{\ell = 1}^L \widehat{\theta}_{\ell}(\widehat{\mu}_{\ell}(x) - \mu_0(x))\right|\\ 
&\le M\max_{1\le \ell \le L}\|\widehat{\mu}_{\ell} - \mu_0\|_{\infty}.
\end{align*}
Therefore, using the symmetry of the error distribution,
\begin{equation}\label{eq:CC-conditional-coverage}
\left|\mathrm{CC}_{\alpha} - (2F(T_{\alpha,\widehat{\theta}}) - 1)\right| \le 2M\max_{1\le \ell \le L}\|\widehat{\mu}_{\ell} - \mu_0\|_{\infty.}
\end{equation}
Now it suffices to study $(2F(T_{\alpha,\widehat{\theta}}) - 1) - (1 - \alpha)$.
Taking $t = T_{\alpha,\widehat{\theta}}$, we get with probability at least $1 - \delta$ (conditional on $\mathcal{D}_1$),
\begin{equation}\label{eq:distribution-func-aggreagtion}
\left|(2F(T_{\alpha,\widehat{\theta}}) - 1) - (1 - \alpha)\right| \le (M/2)\mathbb{E}[(\widehat{\mu}_{\theta^*}(X) - \mu_0(X))^2|\mathcal{D}_1] + \sqrt{\frac{\log(2/\delta)}{2|\mathcal{I}_2|}} + \frac{1}{|\mathcal{I}_2|}.
\end{equation}
Combined with~\eqref{eq:CC-conditional-coverage}, we conclude that
\[
\mathbb{P}\left(\mathrm{CC}_{\alpha} \ge 1 - \alpha - 2M\|\widehat{\mu}_{\ell} - \mu_0\|_{\infty} - (M/2)\|\widehat{\mu}_{\theta^*} - \mu_0\|_2^2 - \sqrt{\frac{\log(2/\delta)}{2|\mathcal{I}_2|}} - \frac{1}{|\mathcal{I}_2|}\bigg|\mathcal{D}_1\right) \le \delta.
\]
Hence, under~\eqref{eq:assump-conditional-coverage}, this implies
\[
\mathbb{P}\left((X_{N+1}, Y_{N+1})\in\widehat{C}_{\alpha}^{\texttt{agg}}\big|X_{N+1} = x\right) \ge 1 - \alpha - (2Mr_{n,\infty} + Mr_{n,2}^2/2) - \frac{\mathfrak{C}}{|\mathcal{I}_2|^{1/2}} - \eta_{n,2} - \eta_{2,\infty}.
\]
% \end{proof}
\section{Proof of Results in Section~\ref{sec:linear-prediction-fixed-width}}
\subsection{Proof of Lemma \ref{lem:R2-bound}} 
\label{sec:Proof of Lemma R2-bound}
Recall our definition in \eqref{F} that $F_{\theta}(t):= \mathbb{P} \bigl( |Y- X^\top \theta | \leq t \bigr)$.
\paragraph{Proof of~\eqref{eq:cdf-beta-hat-beta-bound} under~\ref{assump:density} and~\ref{assump:subGauss}:}
For all $\forall t\ge 0,  \theta, \theta_1 \in \mathbb{R}^d$,
\begin{equation}
\begin{aligned}
    F_{\theta}(t) - F_{\theta_1}(t)& = \PP \bigl(|Y-X^\top \theta| \leq t \bigr) -\PP \bigl(|Y-X^\top \theta_1| \leq t \bigr) \\
    &\leq \PP \bigl(|Y-X^\top \theta| \le t <  |Y-X^\top \theta_1| \bigr)\\
    &\le \inf_{\epsilon \ge 0} \PP \bigl( |Y-X^\top \theta| \in (t-\epsilon, t) \bigr) + \PP \bigl( |Y-X^\top \theta_1|-|Y-X^\top \theta| > \epsilon \bigr)\\
    & \leq \inf_{\epsilon \ge 0} \biggl( \PP \bigl(|Y-X^\top \theta| \in [t-\epsilon, t] \bigr) + \PP \bigl(|X^\top (\theta-\theta_1)| \geq \epsilon \bigr) \biggr)\\
    & \leq \inf_{\epsilon \ge 0} \biggl( \Psi \epsilon + \PP \bigl(|X^\top (\theta-\theta_1)| \geq \epsilon \bigr) \biggr) \text{ by assumption } \ref{assump:density} \\
    & \leq \inf_{\epsilon \ge 0} \biggl(  \Psi \epsilon + \frac{\E |X^\top (\theta-\theta_1)|^{q} }{ \epsilon^{q}} \biggr) \,\, \forall q \geq 0\\
    & \leq \inf_{\epsilon \ge 0, q \ge 0} \biggl( \Psi \epsilon + \frac{ \|\theta-\theta_1 \|_{\Sigma}^q K_{x}^{q} {q}^{\frac{q}{2}} }{ \epsilon^{q}} \biggr) \text{ by assumption } \ref{assump:subGauss}. 
\end{aligned}
\end{equation}
Taking the derivative of the last line w.r.t $\epsilon$, we get $$\inf_{\epsilon} \biggl(  \Psi \epsilon + \frac{ \|\theta-\theta_1 \|_{\Sigma}^q K_{x}^{q} {q}^{\frac{q}{2}} }{ \epsilon^q} \biggr) =  \frac{q+1}{q} ( \Psi K_{x} \| \theta - \theta_1 \|_{\Sigma})^{\frac{q}{q+1}} \sqrt{q} q^{1/(2q+2)}.$$ And then note that $(q+1)/q \le 2$ and $q^{1/(2q+2)} \le e^{1/(2e)}$ for all $q\ge 1$. Hence, it remains to choose $q$ and bound
\[
( \Psi K_{x}\|\theta - \theta_1\|_{\Sigma})^{q/(q+1)} =  \Psi K_{x}\|\theta - \theta_1\|_{\Sigma} \, \bigl( \Psi K_{x}\|\theta - \theta_1\|_{\Sigma}\bigr)^{-1/(q+1)}.
\]
Taking $q + 1 = \max\{2,\log(1/( \Psi K_{x}\|\theta - \theta_1\|))\}$ gives us
\[
(\Psi K_{x}\|\theta - \theta_1\|_{\Sigma})^{-1/(q+1)} \leq \max\{1, ( \Psi K_{x}\|\theta - \theta_1\|_{\Sigma})^{-1/\log(1/( \Psi K_{x}\|\theta - \theta_1\|_{\Sigma}))}\} = \max\{1, e^{-1}\} = 1.
\]
Substituting this choice of $q$ yields
\[
\min_{q\ge1}\inf_{\varepsilon > 0} \biggl(  \Psi \epsilon + \frac{ \|\theta-\theta_1 \|_{\Sigma}^q K_{x}^q q^{\frac{q}{2}} }{ \epsilon^q} \biggr) \leq  \Psi K_{x}\|\theta - \theta_1\|_{\Sigma} \sqrt{\max\{2, \log\left(1/(2 \Psi K_x\|\theta - \theta_1\|_{\Sigma} )\right)\}}.
\]
The argument for $ F_{\theta_1}(t) - F_{\theta}(t)$ is similar.
Therefore, there exists a universal constant $\mathfrak{C}$ such that for all $\lambda > -\lambda_{\min}(\Sigma)$,
\begin{equation}
\sup_{t \ge 0} |F_{\widehat{\beta}_{\lambda}}(t) - F_{\beta_{\lambda}^*}(t)| \leq \mathfrak{C} \Psi K_{x} \| \widehat{\beta}_{\lambda} - \beta_{\lambda}^* \|_{\Sigma} \sqrt{ \max\{2, \log \bigl(1/ (\Psi  K_{x} \| \widehat{\beta}_{\lambda} - \beta_{\lambda}^* \|_{\Sigma}) \bigr) \}},
\end{equation}
and the desired inequality~\eqref{eq:cdf-beta-hat-beta-bound} follows.

\paragraph{Proof of~\eqref{eq:cdf-beta-hat-beta-bound} under~\ref{assump:density} and~\ref{assump:finite-X}:}
For all $t\ge 0,  \theta, \theta_1 \in \mathbb{R}^d$,
\begin{align*}
    F_{\theta}(t) - F_{\theta_1}(t)& = \PP \bigl(|Y-X^\top \theta| \leq t \bigr) -\PP \bigl(|Y-X^\top \theta_1| \leq t \bigr) \\
    & \leq \inf_{\epsilon \ge 0} \biggl( \PP \bigl(|Y-X^\top \theta| \in [t-\epsilon, t] \bigr) + \PP \bigl(|x^\top (\theta-\theta_1)| \geq \epsilon \bigr) \biggr)\\
    & \leq \inf_{\epsilon \ge 0} \biggl( \Psi \epsilon + \PP \bigl(|X^\top (\theta-\theta_1)| \geq \epsilon \bigr) \biggr) \text{ by assumption } \ref{assump:density} \\
    & \leq \inf_{\epsilon \ge 0} \biggl(  \Psi \epsilon + \frac{\E |X^\top (\theta-\theta_1)|^{q_x} }{ \epsilon^{q_x}} \biggr) \\
    & \leq \inf_{\epsilon \ge 0} \biggl( \Psi \epsilon + \frac{ \|\theta-\theta_1 \|_{\Sigma}^{q_x} K_{x}^{q_x} }{ \epsilon^{q_x}} \biggr) \text{ by assumption } \ref{assump:finite-X}.
\end{align*}
Taking the derivative of the last line w.r.t $\epsilon$, we get 
$$
\inf_{\epsilon \ge 0} \biggl(  \Psi \epsilon + \frac{ \|\theta-\theta_1 \|_{\Sigma}^{q_x} K_{x}^{q_x}  }{ \epsilon^{q_x}} \biggr) =  \frac{q_x+1}{q_x} ( \Psi K_{x} \| \theta - \theta_1 \|_{\Sigma})^{q_x/(q_x+1)} q_x^{1/(q_x+1)}.
$$ 
Note that $(q_x+1)/q_x \le 2$ and $q_x^{1/(q_x+1)} \le e^{1/e}$. Hence, it remains to bound
\[
(\Psi K_{x}\|\theta - \theta_1\|_{\Sigma})^{q_x/(q_x+1)} =  \Psi K_{x}\|\theta - \theta_1\|_{\Sigma} \, \bigl( \Psi K_{x}\|\theta - \theta_1\|_{\Sigma}\bigr)^{-1/(q_x+1)}.
\]
The argument for $ F_{\theta_1}(t) - F_{\theta}(t)$ is similar. 
Therefore, there exists a universal constant $\mathfrak{C}$ such that for all $\lambda > -\lambda_{\min}(\Sigma)$,
\begin{equation}
\, \sup_{t \ge 0} |F_{\widehat{\beta}_{\lambda}}(t) - F_{\beta_{\lambda}^*}(t)| \leq \mathfrak{C} (\Psi K_{x}\| \widehat{\beta}_{\lambda} - \beta_{\lambda}^* \|_{\Sigma} )^{q_x/(1+q_x)},
\end{equation}
and the desired inequality~\eqref{eq:cdf-beta-hat-beta-bound} follows.

\paragraph{Proof of~\eqref{eq:cdf-beta-hat-beta-bound} under~\ref{assump:density_conditional}:}
For the third part of Lemma \ref{lem:R2-bound} where we impose the stronger density assumption \ref{assump:density_conditional}, we have from the triangle inequality
\begin{equation}\label{eq:sandwich}
\PP \bigl( |Y-X^\top \theta_1 | \le t - |X^\top (\theta - \theta_1)| \bigr)  \leq \PP \bigl( |Y-X^\top \theta| \le t \bigr) \le \PP \bigl( |Y-X^\top \theta_1 | \le t + |X^\top (\theta - \theta_1)| \bigr).
\end{equation}
We bound the rightmost term in \eqref{eq:sandwich} from above,
\begin{align*}
 \PP \bigl( |Y-X^\top \theta_1 | \le t + |X^\top (\theta - \theta_1)| \bigr) &= \E \bigg[ \PP \biggl( |Y-X^\top \theta_1  | \le t + |X^\top (\theta - \theta_1)| \mathrel{\stretchto{\mid}{4ex}} X \biggr) \biggr]\\
 &= \E \bigg[ \PP \biggl( |Y-X^\top \theta_1  | \le t + |X^\top (\theta - \theta_1)| \mathrel{\stretchto{\mid}{4ex}} X \biggr) - \PP \biggl( | Y-X^\top \theta_1 | \le t \mathrel{\stretchto{\mid}{4ex}}
 X \biggr) \biggr]\\
& \qquad \qquad +\PP \biggl( |Y-X^\top \theta_1  | \le t |  \biggr) \\
& \leq \psi \E\bigl[ |X^\top (\theta - \theta_1) | ] + \PP \bigl ( |Y - X^\top \theta_1 | \le t )\\
& \leq \psi \bigl( \E\bigl[ |X^\top (\theta - \theta_1)|^2] \bigr)^{1/2}+ \PP \bigl ( |Y - X^\top \theta_1 | \le t )\\
& = \psi \| \theta - \theta_1 \|_{\Sigma}+ \PP \bigl ( |Y - X^\top \theta_1 | \le t ).
\end{align*} 
Therefore, 
\begin{align*}\PP \bigl( |Y-X^\top \theta| \leq t \bigr) -\PP \bigl ( |Y - X^\top \theta_1 | \le t ) &\overset{\eqref{eq:sandwich}}{\le} \PP \bigl( |Y-X^\top \theta_1 | \le t + |X^\top (\theta - \theta_1)| \bigr) -\PP \bigl ( |Y - X^\top \theta_1 | \le t ) \\
&\, \,\le  \psi \| \theta - \theta_1 \|_{\Sigma}.
\end{align*} 
The other side can be proven in a similar way by bounding the leftmost term of \eqref{eq:sandwich} from below and the desired inequality~\eqref{eq:cdf-beta-hat-beta-bound} follows.
%\end{proof}

\subsection{Proof of Lemma~\ref{lem:ridge}}\label{sec:Proof of Lemma ridge}
% \begin{proof}[Proof of Lemma \ref{lem:ridge}]
By definition $\widehat{\beta}_{\lambda}$ and $\beta^*_{\lambda}$ satisfy 
\[
(\widehat{\Sigma}_1 + \lambda I_d)\widehat{\beta}_{\lambda} = \widehat{\Gamma}_1\quad\mbox{and}\quad (\Sigma + \lambda I_d)\beta^*_{\lambda} = \Gamma.
\]
Subtracting $\beta^*_{\lambda}$ from $\widehat{\beta}_{\lambda}$ from the first equation, we obtain
\begin{equation}
\begin{split}
(\widehat{\Sigma}_1 + \lambda I_d)(\widehat{\beta}_{\lambda} - \beta^*_{\lambda}) &= \widehat{\Gamma}_1 - (\widehat{\Sigma}_1 + \lambda I_d)\beta^*_{\lambda}\\
&= \widehat{\Gamma}_1 - (\widehat{\Sigma}_1 + \lambda I_d)\beta^*_{\lambda} - \left\{\Gamma - (\Sigma + \lambda I_d)\beta^*_{\lambda}\right\}\\
&= (\widehat{\Gamma}_1 - \Gamma) - (\widehat{\Sigma}_1 - \Sigma)\beta^*_{\lambda}.
\end{split}
\end{equation}
Therefore, we obtain
\[
\|\Sigma^{-1/2}(\widehat{\Sigma}_1 + \lambda I_d)\Sigma^{-1/2}\Sigma^{1/2}(\widehat{\beta}_{\lambda} - \beta^*_{\lambda})\|_2 \le \|\Sigma^{-1/2}(\widehat{\Gamma}_1 - \Gamma)\|_2 + \|\Sigma^{-1/2}(\widehat{\Sigma}_1 - \Sigma)\Sigma^{-1/2}\Sigma^{1/2}\beta^*_{\lambda}\|_2,
\]
and
\begin{align*}
&\left(\inf_{\theta\in\mathbb{R}^d}\,\frac{\|\Sigma^{-1/2}(\widehat{\Sigma}_1 + \lambda I_d)\Sigma^{-1/2}\theta\|_2}{\|\theta\|_2}\right)\|\widehat{\beta}_{\lambda} - \beta^*_{\lambda}\|_{\Sigma}\\ 
&\quad\le \|\Sigma^{-1/2}(\widehat{\Gamma}_1 - \Gamma)\|_2 + \|\Sigma^{-1/2}(\widehat{\Sigma}_1 - \Sigma)\Sigma^{-1/2}\|_{op}\|\Sigma^{1/2}\beta^*_{\lambda}\|_2.
\end{align*}
Observe that
\[
\inf_{\theta\in\mathbb{R}^d}\frac{\|\Sigma^{-1/2}(\widehat{\Sigma}_1 + \lambda I_d)\Sigma^{-1/2}\theta\|_2}{\|\theta\|_2} \ge \inf_{\theta\in\mathbb{R}^d}\frac{\|\Sigma^{-1/2}(\Sigma + \lambda I_d)\Sigma^{-1/2}\theta\|_2}{\|\theta\|_2} - \mathcal{D}_{\Sigma} \ge 1 - c - \mathcal{D}_{\Sigma}.
\]
The lower bound of $1 - c$ follows from the fact that $\lambda \ge -c\lambda_{\min}(\Sigma)$ implies $\Sigma + \lambda I_d \succeq (1 - c)\Sigma$ (in the matrix positive definiteness sense). Hence, we conclude
\[
\|\widehat{\beta}_{\lambda} - \beta^*_{\lambda}\|_{\Sigma} \le \frac{1}{(1 - c - \mathcal{D}_{\Sigma})_+}\left[\|\Sigma^{-1/2}(\widehat{\Gamma}_1 - \Gamma)\|_2 + \mathcal{D}_{\Sigma}\|\Sigma^{1/2}\beta^*_{\lambda}\|_2\right].
\]
Furthermore, from the definition of $\beta^*_{\lambda}$, we have
\begin{equation}
(\Sigma + \lambda I_d)\beta^*_{\lambda} = \Gamma\quad\Rightarrow\quad \left(\inf_{\theta\in\mathbb{R}^d}\frac{\|\Sigma^{-1/2}(\Sigma + \lambda I_d)\Sigma^{-1/2}\theta\|_2}{\|\theta\|_2}\right)\|\Sigma^{1/2}\beta^*_{\lambda}\|_2 \le \|\Sigma^{-1/2}\Gamma\|_2.
\end{equation}
This implies that
\begin{equation}
\|\Sigma^{1/2}\beta^*_{\lambda}\|_2 \le \frac{1}{1 - c}\|\Sigma^{-1/2}\Gamma\|_2.
\end{equation}
% The right hand side is bounded by $ \frac{2}{\sqrt{1-c}}\|\Sigma^{-1/2}\Gamma\|_2$ whenever $\lambda \ge -c\lambda_{\min}(\Sigma)$ for $c<1$. Combining all these inequalities, we get
% \begin{equation}
% \sup_{\lambda \geq -c\lambda_{\min}(\Sigma) } \|\widehat{\beta}_{\lambda} - \beta^*_{\lambda}\|_{\Sigma(\lambda)} ~\le~ \frac{ (1-c)^{-3/2}  }{(1 - \mathcal{D}_{\Sigma} /(1-c) )_+}\left[\| (1-c) \Sigma^{-1/2}(\widehat{\Gamma} - \Gamma)\|_2 + 2\mathcal{D}_{\Sigma}\|\Sigma^{-1/2}\Gamma\|_2\right]. 
% \end{equation}
Lastly, under assumption \ref{assump:finite-Y}, we have
\begin{equation}
\begin{aligned}
    \| \Sigma^{-1/2} \Gamma \|_2 &= \sup_{a \in \mathbb{R}^d, \|a\|=1 } a^\top \Sigma^{-1/2} \E[XY] = \sup_{a \in \mathbb{R}^d, \|a\|=1 } \E [(a^\top \Sigma^{-1/2} X) Y ] = \sup_{a \in \mathbb{R}^d, \|a\|=1 } \E [(a^\top \Sigma^{-1/2} X) Y ]\\
    &\leq \sup_{a \in \mathbb{R}^d, \|a\|=1 } \bigl( \E [a^\top \Sigma^{-1/2} X X^\top \Sigma^{-1/2} a ] \bigr)^{1/2} (\E[Y^2])^{1/2}\\
    & = (\E[Y^2])^{1/2} \leq K_y ,
\end{aligned}
\end{equation}
where we use the Cauchy--Schwarz inequality for the penultimate inequality and assumption \ref{assump:finite-Y} for the last inequality. Combining these inequalities, we obtain
\[
\|\widehat{\beta}_{\lambda} - \beta^*_{\lambda}\|_{\Sigma} \le \frac{1}{(1 - c - \mathcal{D}_{\Sigma})_+}\left[\|\Sigma^{-1/2}(\widehat{\Gamma}_1 - \Gamma)\|_2 + \mathcal{D}_{\Sigma}\frac{K_y}{1 - c}\right].
\]
% \end{proof}

\subsection{Proof of Lemma \ref{lem:D_Sigma}:}
\label{sec:proof of Lemma D_Sigma}
\paragraph{Proof of \eqref{eq:D_subgaussian}}
This result is standard: see, e.g., Theorem 4.7.1 of \cite{vershynin2018high} or Theorem 1 of \cite{koltchinskii2014concentration}. 

\paragraph{Proof of \eqref{eq:D_moment}:}
We use Theorem 1.1 of \cite{tikhomirov2018sample} which we state in Proposition \ref{prop:D_Sigma_relaxed}, namely with probability at least $1-1/n$,
\begin{equation}\label{eq:tikhomirov}
\mathcal{D}_{\Sigma} \leq \frac{C_{q_x}}{n} \max _{1 \leq i \leq n}\bigl\|\Sigma^{-1 / 2} X_{i}\bigr\|^{2}+C_{q_x} K_x^{2}\left(\frac{d}{n}\right)^{1-2 / q_x} \log ^{4}\left(\frac{n}{d}\right)+C_{q_x} K_x^{2}\left(\frac{d}{n}\right)^{1-2 / \min\{q_x,4\}}.
\end{equation}
Since
%\begin{equation}
\begin{align*}
\PP(\frac{1}{d}  \| \Sigma^{-1/2 } X_i  \|^2 \ge t) &\le \frac{\E[(\frac{1}{d} \| \Sigma^{-1/2 } X_i \|^2)^{q_x/2}]}{t^{q_x/2}}\\
&\le \frac{1}{d}  \sum_{j=1}^d \frac{\E[ |e_j^\top \Sigma^{-1/2}X_i |^{q_x} ]}{t^{q_x/2}}\\
& \le \frac{K_x^{q_x}}{t^{q_x/2}},
\end{align*}
%\end{equation}
we have 
\begin{equation}\label{temp}
\begin{aligned}
\PP(  \max_{1 \le i \le n} \| \Sigma^{-1/2} X_i\|^2 \ge t ) &\le   \sum_{i=1}^n \left(\frac{K_x^2 }{ t/ d} \right)^{q_x/2}=  n\left( \frac{  K_x^2 d}{t} \right)^{q_x/2}.  \\
\end{aligned}
\end{equation}

Take $t= K_x^2 d (\delta/n)^{-2/q_x}$, \eqref{temp} becomes
\begin{equation}
\begin{aligned}
\PP( \frac{C_{q_x}}{n} \max_{1 \le i \le n} \| \Sigma^{-1/2} X_i\|^2 \ge C_{q_x} K_x^2 d n^{2/q_x-1} \delta^{-2/q_x}  ) &\le \delta.  \\
\end{aligned}
\end{equation}

Therefore, \eqref{eq:tikhomirov} reduces to our desired result in equation \eqref{eq:D_moment}.

\paragraph{Proof of \eqref{eq:D_structure}:}
We first prove that under assumption \ref{assump:finite-X-indep}, the first term on the right of \eqref{eq:tikhomirov} can have a better bound, notably $\mathcal{D}_\Sigma$ goes to zero if $d=o(n)$.
\begin{equation}\label{eq:D-first-term-structure}
\begin{aligned}
 \max_{1 \le i \le n} \| \Sigma^{-1/2} X_i \|^2 &= d \max_{1 \le i \le n} \frac{1}{d} \| \Sigma^{-1/2} X_i \|^2\\
 &= d + d \max_{1 \le i \le n} \frac{1}{d} \| \Sigma^{-1/2} X_i \|^2 - \E [ \frac{1}{d} \| \Sigma^{-1/2} X_i \|^2 ]\\
 &= d + d \max_{1 \le i \le n} \frac{1}{d} \sum_{j=1}^d \{ Z_{ij}^2 - \E [ Z_{ij}^2 ] \}. 
\end{aligned}
\end{equation}

\subparagraph{Case 1. When $q_x \geq 4$:}
Applying equation (1.9) of \cite{rio2017constants} with $q=q_x/2$ gives us
\begin{equation}\label{temp3}
\mathbb{P}\left(\sum_{j=1}^d \{ Z_{ij}^2 - \E [ Z_{ij}^2 ] \} >\sigma \sqrt{2 \log (\frac{1}{\delta} ) }+\left(1+ 4 / q_x + q_x / 6  \right) C_{q} ( \frac{1}{\delta} ) ^{2 / q_x }\right) \leq \delta,
\end{equation}
where 
$$
\begin{cases}\sigma^2 = \mathrm{var}(\sum_{j=1}^d \{ Z_{ij}^2 - \E [ Z_{ij}^2 ])= d \bigl( \E [ Z_{i1}^4 ] - (E[Z_{i1}^2])^2 \bigr) \leq d \E [Z_{i1}^4] \leq d K_x^4;\\
C_q= ( \sum_{j=1}^d \E | Z_{ij }^2 - \E [Z_{ij}^2 ] |^{q_x/2} )^{2/q_x} \leq d^{2/q_x} K_x^2.\\
\end{cases}
$$
Therefore, inequality \eqref{temp3} leads to (we can get this same result using Theorem 3.1 of \cite{einmahl2008characterization})
$$
\mathbb{P}\left(\frac{1}{d} \sum_{j=1}^d \{ Z_{ij}^2 - \E [ Z_{ij}^2 ] \} >d^{-1/2} K_x^2 \sqrt{2 \log (\frac{1}{\delta} ) }+d^{-1} \left(1+ 4 / q_x + (q_x / 6) \mathbbm{1}_{ \{q_x/2>3 \} } \right) K_x^2 ( \frac{d}{\delta} ) ^{2 / q_x }\right) \leq \delta.
$$
Because the indicator function is upper bounded by $1$, this gives
$$
\mathbb{P}\left( d+ d \max_{1 \le i \le n} \frac{1}{d} \sum_{j=1}^d \{ Z_{ij}^2 - \E [ Z_{ij}^2 ] \}> d + d^{1/2} K_x^2 \sqrt{2 \log (\frac{1}{\delta} ) }+ \left(1+ 4 / q_x + q_x / 6  \right) K_x^2 ( \frac{d}{\delta} ) ^{2 / q_x }\right) \leq n \delta.
$$
Together with equation \eqref{eq:D-first-term-structure}, we have
\begin{equation}\label{temp2}
\mathbb{P}\left(  \max_{1 \le i \le n} \| \Sigma^{-1/2} X_i\|^2  > d + d^{1/2} K_x^2 \sqrt{2 \log (\frac{n}{\delta} ) }+ \left(1+ 4 / q_x + q_x / 6  \right) K_x^2 ( \frac{nd}{\delta} ) ^{2 / q_x} \right) \leq \delta.
\end{equation}
Therefore, plugging the above into inequality \eqref{eq:tikhomirov} along with $K_x \geq 1$ gives us our desired result in equation \eqref{eq:D_structure} in this case.

\subparagraph{Case 2. When $2 < q_x < 4$:}
It can be easily verified that $\alpha:= q_x/2 \in (1,2]$ and that $Z_{ij}^2$ has $\alpha$-th moment $K_x$. Therefore, taking $X_j= Z_{ij}^2, 1 \leq j \le d $, in Lemma 3.1 of \cite{chen2020generalized} yields
$$
\mathbb{P}\left(   \frac{1}{d} \sum_{j=1}^d \{ Z_{ij}^2 - \E [ Z_{ij}^2 ] \}  \geq  \frac{\mathfrak{C} K_x^2   }{ \delta ^{2/q_x} d^{1-2/q_x} } \right) \leq \delta,
$$
where $\mathfrak{C}$ is a universal constant.
Lemma 3 of \cite{bubeck2013bandits} gives similar result as above as well.
This implies
$$
\mathbb{P}\left( d+d \max_{1 \le i \le n}  \frac{1}{d} \sum_{j=1}^d \{ Z_{ij}^2 - \E [ Z_{ij}^2 ] \}  \geq  d+\frac{\mathfrak{C} K_x   }{ \delta ^{2/q_x} d^{-2/q_x} } \right) \leq n\delta.
$$
Together with equation \eqref{eq:D-first-term-structure}, we have
\begin{equation}
\mathbb{P}\left(  \max_{1 \le i \le n} \| \Sigma^{-1/2} X_i\|^2  >d+\mathfrak{C} K_x   (\frac{nd}{\delta})^{2/q_x} \right) \leq \delta.
\end{equation}
Therefore, plugging the above into inequality \eqref{eq:tikhomirov} along with $K_x \geq 1$ leads to the desired result in equation \eqref{eq:D_structure} in this case as well.

\subsection{Proof of Lemma \ref{lem:covariance-bound}}
\label{sec:proof of Lemma covariance-bound}
\begin{proof}[Proof of \eqref{eq:gammahat_subgaussian}]
For $\xi_i := \Sigma^{-1/2} ( X_i Y_i - \E [X_i Y_i] )/n, i\in \mathcal{I}_1$, without loss of generality assume the indexes in $\mathcal{I}_1$ is $1, \dots, n$ where $n:= |\mathcal{I}_1| $.
Then, 
$$
\ \Sigma^{-1/2} ( \widehat{\Gamma} - \Gamma ) = \sum_{i=1}^n \xi_i, \text{ with } \E[\xi_i]=0.
$$
Theorem 3.1 of \cite{einmahl2008characterization} with $\eta=\delta=1, Z_i=\xi_i$ implies that
\begin{equation}\label{eq:einmah}
\mathbb{P}\left( \|\Sigma^{-1/2} (\widehat{\Gamma} - \Gamma ) \|_2 \geq 2\E \left[ \|\Sigma^{-1/2} (\widehat{\Gamma} - \Gamma ) \|_2 \right]+t\right) \leq c \exp \left(-\frac{t^{2}}{3 \sigma^2}\right)+\frac{C}{t^{s}} \sum_{i=1}^{n} \mathbb{E}\left[\left( \xi_{i}^{\top} \xi_{i} \right)^{s / 2}\right],
\end{equation}
for any $s>2$ such that $\mathbb{E}\left[\left( \xi_{i}^{\top} \xi_{i} \right)^{s / 2}\right]$ is finite and 
\begin{equation}\label{eq:sigma}
\sigma^2 := \sup_{\| \theta \|_2 \leq 1 } \sum_{i=1}^n \E (\theta^\top \xi_i)^2.
\end{equation}
Here $C \in(0, \infty)$ is a constant depending only on $s$.
We tackle the three terms in~\eqref{eq:einmah}: $\E \bigl[ \|\Sigma^{-1/2} (\widehat{\Gamma} - \Gamma ) \|_2 \bigr],\sum_{i=1}^{n} \mathbb{E}\left[\left( \xi_{i}^{\top} \xi_{i}\right)^{s / 2} \right]  $ and $\sigma^2$ respectively. 

First,
\begin{align*}
    \E \bigl[ \|\Sigma^{-1/2} (\widehat{\Gamma} - \Gamma ) \|_2 \bigr] &\leq (\E \bigl[ \|\Sigma^{-1/2} (\widehat{\Gamma} - \Gamma ) \|_2^2 \bigr])^{1/2}\\
   % & = \biggl( \E \bigl[ (\sum_{i=1}^n Z_i)^\top (\sum_{i=1}^n Z_i)  \bigr] \biggr)^{1/2}\\
   % &= \biggl( \E \bigl[ \sum_{i=1}^n Z_i^\top Z_i \bigr] \biggr)^{1/2}\\
   % &= \biggl( \sum_{i=1}^n  \E \bigl[  Z_i^\top Z_i \bigr] \biggr)^{1/2}\\
     &= \frac{1}{n} \biggl( \sum_{i=1}^n \E \bigl[\| \Sigma^{-1/2} X_i Y_i - \E[\Sigma^{-1/2} X_i Y_i] \|_2^2 \bigr] \biggr)^{1/2}\\
    &\leq \frac{1}{\sqrt{n} } \biggl( \E \bigl[\|  \Sigma^{-1/2} X_i Y_i \|_2^2 \bigr]\biggr)^{1/2}\\
    &=\frac{1}{\sqrt{n} } \biggl( \E \bigl[  Y_i^2 X_i^\top \Sigma^{-1} X_i \bigr]\biggr)^{1/2}\\
    &=\sqrt{\frac{d}{n} } \biggl( \frac{1}{d} \sum_{j=1}^d \E \bigl[  Y_i^2 (e_j^\top \Sigma^{-1/2} X_i)^2 \bigr]\biggr)^{1/2}\\
    &\leq C_q  K_x K_y \sqrt{\frac{d}{n} },
\end{align*}
where the last inequality follows by taking  $\eta=\frac{q_y}{2}-1 >0$ in Proposition \ref{prop:s2} with $C_q$ a constant depending on $q_y$.

Next, since
$$
\xi_{i}^{\top} \xi_{i}=\frac{1}{n^2} (X_i Y_i - \E [X_i Y_i ] )^\top \Sigma^{-1} (X_i Y_i - \E [X_i Y_i ] ).
$$
From assumption \ref{assump:DGP}, we have
\begin{align*}
\sum_{i=1}^{n} \mathbb{E}\left[\left( \xi_{i}^{\top} \xi_{i} \right)^{s / 2}\right] 
&=\frac{d^{s / 2}}{n^{s-1}} \mathbb{E}\left[\left(\frac{1}{d} \sum_{j=1}^{d}\left(e_{j}^{\top} \Sigma ^{-1 / 2}   X_i Y_i  - \E [e_j^\top \Sigma^{-1/2} X_i Y_i ] \right)^{2}\right)^{s / 2}\right] \\
& \leq \frac{d^{s / 2}}{n^{s-1}} \max _{1 \leq j \leq d} \mathbb{E}\left[\left|e_{j}^{\top} \Sigma^{-1 / 2} \left( X_i Y_i - \E [X_i Y_i ] \right)\right|^{s}\right]\\
& \leq \frac{2^s d^{s / 2}}{n^{s-1}} \max _{1 \leq j \leq d} \mathbb{E}\left[ |e_{j}^{\top} \Sigma^{-1 / 2} X_i Y_i |^{s}\right]\\
& \leq \frac{2^s d^{s / 2}}{n^{s-1}} \left\{ \mathfrak{C} K_x K_{y} \sqrt{q_y / \eta}\right\}^{q_y /(1+\eta)},
\end{align*}
where the last inequality follows by taking $s=q_y /(1+\eta)$ and applying Proposition \ref{prop:s2} again and $\mathfrak{C}$ is a universal constant. 

Lastly we look at $\sigma^2$ which we recall was defined in~\eqref{eq:sigma}.
\begin{align*}
    \sigma^2 &= \sup_{\| \theta \|_2 \leq 1 } \sum_{i=1}^n \E (\theta^\top \xi_i)^2\\
    & = \sup_{\| \theta \|_2 \leq 1 } \frac{1}{n} \E \biggl[ | \theta^\top \Sigma^{-1/2}X_i Y_i - \E[  \theta^\top \Sigma^{-1/2}X_i Y_i ] |^2\biggr]\\
    &\leq  \frac{2}{n} \sup_{\| \theta \|_2 \leq 1 } \E[|\theta^\top \Sigma^{-1/2}X_i Y_i |^2]\\
    &\leq \frac{C_q ( K_x K_y )^2}{n}, \text{by applying $\eta = \frac{q_y}{2} -1$ in Proposition } \ref{prop:s2},
\end{align*}
where $C_q$ is a constant that only depends on $q_y$.
Therefore, taking for possibly different constants $C_1, C_2 \in(0, \infty)$ that only depend on $q_y$,
$$
t=K_x K_y \sqrt{ 3C_1  } \sqrt{ \frac{\log  (1/\delta  )}{n} }+\left(C_2 \delta^{-1}\right)^{(1+\eta) / q_y} \frac{  d^{1 / 2}}{n^{1 - (1+\eta) / q_y}}\left\{K_x K_{y} \sqrt{q_y / \eta}\right\}
$$
in equation~\eqref{eq:einmah} yields for any $\eta>0$
\begin{align*}
\mathbb{P}\biggl( \|\Sigma^{-1/2} (\widehat{\Gamma} - \Gamma ) \|_2 &\geq 2C_q K_x K_y \sqrt{\frac{d}{n}} + K_x K_y \sqrt{ 3C_1  } \sqrt{ \frac{\log  (1/\delta  )}{n} }+ \\
&K_x K_{y}  \left(C_2 \delta^{-1}\right)^{(1+\eta) / q_y} \frac{ \sqrt{ d q_y / \eta}}{n^{1 - (1+\eta) / q_y}}  \biggr) \leq \delta.
\end{align*}
Optimize the above line over $\eta$ gives us $\eta = \frac{q_y}{2 \log \frac{C_2 n}{ \delta}}$, and this leads to inequality~\eqref{eq:gammahat_subgaussian} (for a possibly different constant)
$$
\mathbb{P}\left( \|\Sigma^{-1/2} (\widehat{\Gamma} - \Gamma ) \|_2 \geq C_q  K_x K_y \sqrt{ \frac{d +\log (1/\delta) }{n}  } + C_q K_x K_y \left( \frac{1}{\delta} \right)^{1/q_y} \frac{\sqrt{d}}{ n^{1-1/q_y} }  \left( \log (\frac{ n}{\delta} ) \right)^{-1/2} \right) \leq \delta.
$$
\end{proof}

\begin{proof}[Proof of \eqref{eq:gammahat_moment}]
We divide the proof into two cases.
\paragraph{Case 1. ${q_x q_y}/{(q_x +q_y)} \leq 2$:}
It can be easily verified that $p:={q_x q_y}/{(q_x +q_y)} \in (1, 2]$. Therefore, taking $W_i= \Sigma^{-1/2}X_i Y_i, 1\le i\le n$, in Proposition \ref{prop:vector-moment} yields 
\begin{align*}
\mathbb{E}\|\Sigma^{-1/2}(\widehat{\Gamma} - \Gamma)\|_2^p &\le \frac{\mathfrak{C}^p}{n^p}\sum_{i=1}^n \mathbb{E}[\|\Sigma^{-1/2}X_iY_i\|_2^p] = \frac{\mathfrak{C}^pd^{p/2}}{n^{p-1}}\sup_{\|\theta\|_2\le 1}\mathbb{E}[|\theta^{\top}\Sigma^{-1/2}X_iY_i|^p]\\ 
&\le \frac{\mathfrak{C}^pd^{p/2}}{n^{p-1}}(K_xK_y)^p.
\end{align*}
The last inequality above follows from Proposition~\ref{prop:s2-relaxed}. Hence, by Markov's inequality, we obtain
\begin{equation}\label{eq:tail-bound-Gamma-hat-less-2}
\mathbb{P}\left(\|\Sigma^{-1/2}(\widehat{\Gamma} - \Gamma)\|_2 \ge \mathfrak{C}K_xK_y\frac{\sqrt{d}\delta^{-(q_x + q_y)/(q_xq_y)}}{n^{1-(q_x+q_y)/(q_xq_y)}}\right) \le \delta.
\end{equation}

\paragraph{Case 2. ${q_x q_y}/{(q_x +q_y)} > 2$:}
Define $\xi_i := \Sigma^{-1/2} ( X_i Y_i - \E [X_i Y_i] )/n, 1 \leq i \leq n$, 
Then $\ \Sigma^{-1/2} ( \widehat{\Gamma} - \Gamma ) = \sum_{i=1}^n \xi_i, \text{ with } \E[\xi_i]=0.$
Apply Theorem 3.1 of \cite{einmahl2008characterization} with $\eta=\delta=1$,  $Z_i = \xi_i$ yields
\begin{equation}\label{eq:einmah2}
\mathbb{P}\left( \|\Sigma^{-1/2} (\widehat{\Gamma} - \Gamma ) \|_2 \geq 2\E \left[ \|\Sigma^{-1/2} (\widehat{\Gamma} - \Gamma ) \|_2 \right]+t\right) \leq c \exp \left(-\frac{t^{2}}{3 \sigma^2}\right)+\frac{C}{t^{s}} \sum_{i=1}^{n} \mathbb{E}\left[\left(\xi_{i}^{\top} \xi_{i}\right)^{s / 2}\right],
\end{equation}
for any $s>2$ such that $\mathbb{E}\left[\left( \xi_{i}^{\top} \xi_{i} \right)^{s / 2}\right]$ is finite and 
\begin{equation}\label{eq:sigma2}
\sigma^2 := \sup_{\| \theta \|_2 \leq 1 } \sum_{i=1}^n \E (\theta^\top \xi_i)^2.
\end{equation}
Here $C \in(0, \infty)$ is a constant depending only on $s$.
We tackle the three terms in~\eqref{eq:einmah2}: $\E \bigl[ \|\Sigma^{-1/2} (\widehat{\Gamma} - \Gamma ) \|_2 \bigr],\sum_{i=1}^{n} \mathbb{E}\left[\left( \xi_{i}^{\top} \xi_{i} \right)^{s / 2} \right]  $ and $\sigma^2$ respectively. 
First,
\begin{equation}\label{cov-bound-heavy-x}
\begin{aligned}
    \E \bigl[ \|\Sigma^{-1/2} &(\widehat{\Gamma} - \Gamma ) \|_2 \bigr] \leq (\E \bigl[ \|\Sigma^{-1/2} (\widehat{\Gamma} - \Gamma ) \|_2^2 \bigr])^{1/2}\\
     &= \frac{1}{n} \biggl( \sum_{i=1}^n \E \bigl[\| \Sigma^{-1/2} X_i Y_i - \E[\Sigma^{-1/2} X_i Y_i] \|_2^2 \bigr] \biggr)^{1/2}\\
    &\leq \frac{1}{\sqrt{n} } \biggl( \E \bigl[\|  \Sigma^{-1/2} X_i Y_i \|_2^2 \bigr]\biggr)^{1/2}\\
    &=\frac{1}{\sqrt{n} } \biggl( \E \bigl[  Y_i^2 X_i^\top \Sigma^{-1} X_i \bigr]\biggr)^{1/2}\\
    &=\sqrt{\frac{d}{n} } \biggl( \frac{1}{d} \sum_{j=1}^d \E \bigl[  Y_i^2 (e_j^\top \Sigma^{-1/2} X_i)^2 \bigr]\biggr)^{1/2}\\
     &\leq \sqrt{\frac{d}{n} } \biggl(\frac{1}{d} \sum_{j=1}^d \sup_{\| \theta \|_2 \leq 1 } \bigl( \E[|\theta^\top \Sigma^{-1/2}X_i Y_i |^p] \bigr)^{2/p} \biggr)^{1/2}\text{ from Jensen's inequality and } p:=\frac{q_x q_y}{q_x +q_y}\\
    &\leq \sqrt{\frac{d}{n} } K_x K_y , \text{by applying $q_1 =q_x, q_2=q_y$ in Proposition } \ref{prop:s2-relaxed}.
\end{aligned}
\end{equation}
%where the last inequality follows by taking  $q_1=q_2=4$ in Proposition \ref{prop:s2-relaxed}.

Next, since
$$
\xi_{i}^{\top} \xi_{i}=\frac{1}{n^2} (X_i Y_i - \E [X_i Y_i ] )^\top \Sigma^{-1} (X_i Y_i - \E [X_i Y_i ] ).
$$
From assumption \ref{assump:DGP}, we have
\begin{align*}
\sum_{i=1}^{n} \mathbb{E}\left[\left( \xi_{i}^{\top} \xi_{i}\right)^{s / 2}\right] 
&=\frac{d^{s / 2}}{n^{s-1}} \mathbb{E}\left[\left(\frac{1}{d} \sum_{j=1}^{d}\left(e_{j}^{\top} \Sigma ^{-1 / 2}   X_i Y_i  - \E [e_j^\top \Sigma^{-1/2} X_i Y_i ] \right)^{2}\right)^{s / 2}\right] \\
& \leq \frac{d^{s / 2}}{n^{s-1}} \max _{1 \leq j \leq d} \mathbb{E}\left[\left|e_{j}^{\top} \Sigma^{-1 / 2} \left( X_i Y_i - \E [X_i Y_i ] \right)\right|^{s}\right]\\
& \leq \frac{2^s d^{s / 2}}{n^{s-1}} \max _{1 \leq j \leq d} \mathbb{E}\left[ |e_{j}^{\top} \Sigma^{-1 / 2} X_i Y_i |^{s}\right]\\
& \leq \frac{2^s d^{s / 2}}{n^{s-1}} ( K_x K_{y} )^s,
\end{align*}
where in the last inequality we substitute $s$ by $q_x q_y/(q_x+q_y)$ and apply Proposition \ref{prop:s2-relaxed}.

Lastly we look at $\sigma^2$ which we recall was defined in~\eqref{eq:sigma},
\begin{align*}
    \sigma^2 &= \sup_{\| \theta \|_2 \leq 1 } \sum_{i=1}^n \E (\theta^\top \xi_i )^2\\
    & = \sup_{\| \theta \|_2 \leq 1 } \frac{1}{n} \E \biggl[ | \theta^\top \Sigma^{-1/2}X_i Y_i - \E[  \theta^\top \Sigma^{-1/2}X_i Y_i ] |^2\biggr]\\
    &\leq  \frac{2}{n} \sup_{\| \theta \|_2 \leq 1 } \E[|\theta^\top \Sigma^{-1/2}X_i Y_i |^2]\\
    &\leq \frac{2}{n} \sup_{\| \theta \|_2 \leq 1 } \bigl( \E[|\theta^\top \Sigma^{-1/2}X_i Y_i |^p] \bigr)^{2/p} \text{ from Jensen's inequality and } p:=\frac{q_x q_y}{q_x +q_y}\\
    &\leq \frac{2 (K_x K_y)^2}{n}, \text{by applying $q_1 =q_x, q_2=q_y$ in Proposition } \ref{prop:s2-relaxed}.
 \end{align*}
Therefore, taking for a universal constant $\mathfrak{C} \in(0, \infty)$ 
$$
t=K_x K_y \sqrt{ 3\mathfrak{C} } \sqrt{ \frac{\log  (1/\delta  )}{n} }+\left( \mathfrak{C} \right)^{(q_x+q_y) /q_x q_y} \frac{  \sqrt{d}\delta^{-(q_x+q_y)/(q_xq_y)}}{n^{1 - (q_x+q_y) / q_x q_y}} K_x K_y,
$$
in equation \eqref{eq:einmah2} yields
$$
\mathbb{P}\biggl( \|\Sigma^{-1/2} (\widehat{\Gamma} - \Gamma ) \|_2 \geq 2 \sqrt{\frac{d}{n}} K_x K_y+ K_x K_y \sqrt{ 3 \mathfrak{C}  } \sqrt{ \frac{\log  (1/\delta  )}{n} }+\left( \mathfrak{C}\right)^{(q_x+q_y) /q_x q_y}  \frac{  \sqrt{d}\delta^{-(q_x+q_y)/(q_xq_y)}}{n^{1 - (q_x+q_y) / q_x q_y}} K_x K_{y}   \biggr) \leq \delta.
$$
%This gives us inequality~\eqref{eq:covariance-bound} (for a possibly different constant)
This results in
\begin{equation}\label{eq:tail-bound-Gamma-hat-greater-2}
\mathbb{P}\left( \|\Sigma^{-1/2} (\widehat{\Gamma} - \Gamma ) \|_2 \geq \mathfrak{C}  K_x K_y \sqrt{ \frac{d +\log (1/\delta) }{n}  } + \mathfrak{C}    \frac{  \sqrt{d}\delta^{-(q_x+q_y)/(q_x q_y)}}{n^{1 - (q_x+q_y) / q_x q_y}} K_x K_y  \right) \leq \delta.
\end{equation}
Combining inequalities \eqref{eq:tail-bound-Gamma-hat-less-2} and~\eqref{eq:tail-bound-Gamma-hat-greater-2} completes the proof of inequality~\eqref{eq:gammahat_moment}.
\end{proof}

\section{Some useful propositions}\label{appsec:useful-propositions}
We state here some propositions used throughout in the paper, in particular in the proofs above.

\begin{prop}\label{prop:R1-bound}
Under assumption \ref{assump:DGP}, for any $\delta \in(0,1)$ and $\mathcal{I}\subseteq[N]$, with probability at least $1-\delta$, we have
\begin{equation}\label{eq:empirical-process-bound}
\sup_{t\ge 0} \sup _{\theta \in \mathbb{R}^d} \biggl|\frac{1}{ |\mathcal{I}| } \sum_{i \in \mathcal{I}} \mathbbm{1}{\{ |y_i - x_i^\top \theta | \leq t \} } - F_{\theta}(t) \biggr| \leq \mathfrak{C} \sqrt{\frac{d + \log (1/\delta) }{| \mathcal{I} | }} ,
\end{equation}
where $\mathfrak{C}$ is a universal constant.
% In particular, with probability at least $1 - \delta$,
% \begin{equation}\label{eq:estimator-empirical-process-bound}
% \mathcal{E}_1(\mathcal{I}_2) \le \mathfrak{C}\sqrt{\frac{d+\log(1/\delta)}{|\mathcal{I}_2|}},
% % \sup_{t\ge0}\sup_{\lambda > -\lambda_{\min}(\Sigma) }\, \biggl|\frac{1}{ |\mathcal{I}_2| } \sum_{i \in \mathcal{I}_2} \mathbbm{1}{\{ |y_i - x_i^\top \widehat{\beta}_\lambda | \leq t \} } - F_{\widehat{\beta}_{\lambda}}(t) \biggr| \leq \mathfrak{C} \sqrt{\frac{d + \log (1/\delta)}{| \mathcal{I}_2 | } }:=\mathcal{E}_1(\mathcal{I}_2 ).
% \end{equation}
% and with probability at least $1 - \delta$, $\mathcal{E}_1(\mathcal{I}_3) \le \mathfrak{C}\sqrt{(d+\log(1/\delta))/|\mathcal{I}_3|}.$
\end{prop}
\begin{proof}
For $Z_i = (X_i, Y_i), i\in \mathcal{I}$, and any function $f:\mathbb{R}^{d+1}\to\mathbb{R}$, without loss of generality assume the indexes in $\mathcal{I}$ is $1, \dots, n$ with $n:= |\mathcal{I}| $, and define
\[
\mathbb{G}_{n} f ~:=~ \frac{1}{\sqrt{n}}\sum_{i=1}^n \{f(Z_i) - \mathbb{E}[f(Z_i)]\}.
\]
We have a class of functions $\mathcal{F}=\{f: f(z)=\mathbbm{1}{\{ | y-x^\top \theta | \leq t\} }, t\ge0  \}$. From Chapter 3.3 of \cite{kearns1994introduction}, we know that the VC dimension of this class is $V(\mathcal{F})=d+1$, and we take the envelope function to be $F \equiv 1$.  The notation of supremum norm $\| \cdot \|_{\mathcal{F}}$ and covering number $N \bigl(\varepsilon, \mathcal{F}, L_r(Q) \bigr)$ are all defined in the same way as in \cite{kearns1994introduction}.    
We use Theorem 2.6.7 of \cite{van1996weak} with $r=2$ to get a bound on the covering number bound. Specifically, for any probability measure $Q$, there exists a universal constant $\mathfrak{C}$ such that
\begin{equation}\label{eq:covering}
 N\left(\varepsilon \| F \|_{Q,2}, \mathcal{F}, L_{2}(Q)\right) \leq \mathfrak{C} (d+1)(16 e)^{d+1}\left(\frac{1}{\varepsilon}\right)^{2d}.   
\end{equation}
Next, we obtain an upper bound of the uniform-entropy 
$$
J(\delta, \mathcal{F}):=\sup _{Q} \int_{0}^{\delta} \sqrt{1+\log N\left(\varepsilon\|F\|_{Q, 2}, \mathcal{F}, L_{2}(Q)\right)} d \varepsilon,
$$
where the supremum is taken over all discrete probability measures $Q$ with $\|F\|_{Q, 2}>0 .$ 
Then applying $p=1$ to Theorem 2.14.1 of \cite{van1996weak} gives us
\begin{equation}\label{eq:expectation-bound}
 \E \| \mathbb{G}_n \|_{\mathcal{F}} ~\leq~ J(1, \mathcal{F}) ~\lesssim~ \int_{0}^{1} \left(\log d + d +d \log \frac{1}{\varepsilon}\right)^{1/2}  d \varepsilon ~\leq~ \mathfrak{C} \sqrt{d}, 
\end{equation}
where the second inequality is from \eqref{eq:covering} and $\mathfrak{C}$ is a universal constant.
Using McDiarmid's inequality we have 
\begin{equation}\label{eq:McDiarmid-application}
\PP ( \| \mathbb{G}_n \|_{\mathcal{F}}- \E  \| \mathbb{G}_n \|_{\mathcal{F} } \geq u ) \leq  \exp \left(-\frac{2u^2}{ \sum_{i=1}^n c_i^2}\right) \leq  \exp \left(-\frac{2u^2}{ \sum_{i=1}^n 1/n}\right) = \exp(-2 u^2),
\end{equation}
where 
\begin{align*}
c_i &:= \sup_{\substack{(x_i, y_i), (x_i', y_i'),\\ (x_1, y_1), \ldots, (x_n, y_n)}} \sup_{\theta,t} \frac{\sqrt{n}}{n} \Bigl|\sum_{j=1}^n \mathbbm{1}{\{ |y_j - \theta^{\top}x_j| \leq t \} } -  \sum_{j=1, j\neq i}^n \mathbbm{1}{\{ |y_j - \theta^\top x_j | \leq t \} } - \mathbbm{1}{\{ |y_i' - \theta^\top x_i' | \leq t \} } \Bigr|  \\
& \leq \sup_{(x_i, y_i), (x_i', y_i')}  \sup_{\theta,t} \frac{1}{\sqrt{n}} \biggl|   \mathbbm{1}{\{ |y_i - \theta^{\top}x_i | \leq t \} } -  \mathbbm{1}{\{ |y_i' - \theta^\top x_i' | \leq t \} }  \biggr| \le \frac{1}{\sqrt{n}}.
\end{align*}
Substituting the expectation bound~\eqref{eq:expectation-bound} in~\eqref{eq:McDiarmid-application} and setting the right hand side of~\eqref{eq:McDiarmid-application} to $\delta$ yields for another absolute constant $\mathfrak{C}^{\prime}$
$$
\mathbb{P}\left(\| \mathbb{G}_n \|_{\mathcal{F}} \geq  \mathfrak{C}^\prime \sqrt{d +\log \left(\frac{1}{\delta} \right)} \right) \leq \mathbb{P}\left(\| \mathbb{G}_n \|_{\mathcal{F}} \geq \mathfrak{C} \sqrt{d} +\sqrt{\frac{1}{2}\log \left(\frac{1}{\delta}\right)}\right) \leq \delta.
$$
Hence, inequality~\eqref{eq:empirical-process-bound} follows. 
% Inequality~\eqref{eq:estimator-empirical-process-bound} is based on the observation that $(X_i, Y_i), i \in \mathcal{I}_2$, are independent with $\widehat{\beta}_{\lambda}$, and that
% \begin{align*}
% &\sup_{t\ge0}\sup_{\lambda > -\lambda_{\min}(\Sigma) }\,\left|\frac{1}{ | \mathcal{I}_2 |}\sum_{i\in \mathcal{I}_2} \mathbbm{1}\{|Y_i - X_i^{\top}\widehat{\beta}_{\lambda}| \leq t\} - F_{\widehat{\beta}_{\lambda}}(t)\right|\\ &\quad\leq \sup_{t\ge0}\sup_{\theta\in\mathbb{R}^d}\,\left|\frac{1}{ | \mathcal{I}_2 |}\sum_{i\in \mathcal{I}_2 } \mathbbm{1}\{|Y_i - X_i^{\top}\theta| \leq t\} - F_{\theta}(t)\right|.
% \end{align*}
% This completes the proof.
\end{proof}

\begin{prop}\label{prop:holder}
For two cumulative distribution functions $F_1$ and $F_2$, set
\[
\Delta := \sup_{t}|F_1(t) - F_2(t)|.
\]
If $F_2^{-1}(\cdot)$ is $\gamma$-H{\"o}lder continuous on $[q-\Delta, q+\Delta]$ for $\gamma \in (0,1)$, then it holds that 
\begin{equation}\label{eq:quantile}
|F_1^{-1}(q) - F_2^{-1}(q)| \le \mathfrak{L} \Delta^{\gamma},    
\end{equation}
where $\mathfrak{L}$ is the H\"older continuity constant.
\end{prop}
\begin{proof}%[Proof of Proposition \ref{prop:holder}]
% Set $\Delta:=\sup_{t}|F_1(t) - F_2(t)|$. 
Note that
\begin{align*}
&|F_1(F_1^{-1}(q)) - F_2(F_1^{-1}(q))| \le \Delta\\ \quad\Rightarrow\quad &q - \Delta \le F_2(F_1^{-1}(q)) \le q + \Delta\\\quad\Rightarrow\quad &F_2^{-1}(q - \Delta) \le F_1^{-1}(q) \le F_2^{-1}(q + \Delta).
\end{align*}
Therefore, using the H{\"o}lder continuity assumption, we obtain
$$
F_1^{-1}(q) - F_2^{-1}(q) \le F_2^{-1}(q + \Delta) - F_2^{-1}(q) \le \mathfrak{L}\Delta^{\gamma}
$$ 
and
$$
F_1^{-1}(q) - F_2^{-1}(q) \ge F_2^{-1}(q - \Delta) - F_2^{-1}(q) \ge -\mathfrak{L}\Delta^{\gamma}.
$$
Hence inequality \eqref{eq:quantile} follows.
\end{proof}

\begin{prop} (Theorem 1 of \cite{tikhomirov2018sample})\label{prop:D_Sigma_relaxed}
Under assumptions \ref{assump:DGP} and \ref{assump:finite-X}, when $n \geq 2d$ there exists a constant $C_{q_x}>0$, depending only on $q_x$, such that with probability at least $1-1 / n$,
$$
\mathcal{D}_{\Sigma} \leq \frac{C_{q_x}}{n} \max _{1 \leq i \leq n}\bigl\|\Sigma^{-1 / 2} X_{i}\bigr\|^{2}+C_{q_x} K_x^{2}\left(\frac{d}{n}\right)^{1-2 / q_x} \log ^{4}\left(\frac{n}{d}\right)+C_{q_x} K_x^{2}\left(\frac{d}{n}\right)^{1-2 / \min \{q_x, 4\}}.
$$
\end{prop}

\begin{prop}\label{prop:vector-moment} If $W_1, \dots, W_n$ are independent random vectors,
% and $S=\sum_{i=1}^n W_i$
then there exists a universal constant $\mathfrak{C}$ such that 
\iffalse for any $p \geq 1$,
\begin{equation}
\Bigl( \E\bigl\| \sum_{i=1}^{n} ( W_i- \E W_i ) \bigr\|_2^p \Bigr)^{1/p} \leq \mathfrak{C} \Bigl\{ p^{1/2} \Bigl( \sum_{i=1}^n  \E \| W_i \|_2 \mathbbm{1}_{\{ \|W_i \| \leq B  \} } \Bigr)^{1/2} +   p\Bigl( \sum_{i=1}^n  \E \| W_i \|_2^p \Bigr)^{1/p} \Bigr\},
\end{equation}
where $B=K (\sum_{i=1}^n \E [ \| W_i \|_2^p ] )^{1/p} $ for a universal constant $K$.
In addition, if $p \leq 2$, then
\fi
\begin{equation}
\Bigl( \E\bigl\| \sum_{i=1}^{n} ( W_i- \E W_i ) \bigr\|_2^p \Bigr)^{1/p} \leq \mathfrak{C} \Bigl( \sum_{i=1}^n  \E \| W_i \|_2^p \Bigr)^{1/p},\quad\mbox{for}\quad 1\le p\le 2.
\end{equation}
\end{prop}

\begin{proof}
Set $M^p:=8 \cdot 3^p \E [\max_{i=1, \dots, n} \| W_i \|_2^p]$
\begin{align*}
    \Bigl( \E\bigl\| \sum_{i=1}^n ( W_i- \E W_i ) \bigr\|_2^p \Bigr)^{1/p} &\leq 2 \Bigl( \E \bigl\| \sum_{i=1}^n  \varepsilon_i W_i \bigr\|_2^p \bigr) ^{1/p} \text{ by symmetrization}\\
    &\leq 2 \Bigl( \E\bigl\| \sum_{i=1}^n  \varepsilon_i W_i \mathbbm{1}_{ \{ \| W_i \|_2 \leq M \} }  \bigr\|_2^p \Bigr)^{1/p} + 2 \Bigl( \E \bigl\| \sum_{i=1}^n  \varepsilon_i W_i \mathbbm{1}_{ \{ \| W_i \|_2 > M \} }  \bigr\|_2^p \Bigr)^{1/p} \\
    &:=2(\E[S_1^p])^{1/p}+2(\E[S_2^p])^{1/p}.
\end{align*}
Using inequality (3.7) of \cite{gine2000exponential} paper, we have
\begin{equation}\label{eq:S1-first}
(\E [S_{1}^{p}])^{1/p} \leq \mathfrak{C} \left( \E[S_1] +\sqrt{p}  \sigma + p M\right),
\end{equation}
where $$\sigma^2 = \sup_{\| \theta \|_2 \le 1} \sum_{i=1}^n \E [ ( \theta^\top W_i) ^2 \mathbbm{1}_{ \{ \| W_i \|_2 \leq M \} }  ] \leq \sum_{i=1}^n \E [ \| W_i \|_2 ^2 \mathbbm{1}_{ \{ \| W_i \|_2 \leq M \} }  ] .$$
Notice that $M \leq \mathfrak{C} (\sum_{i=1}^n \E [ \| W_i \|_2^p ] )^{1/p}$ for a possibly different universal constant $\mathfrak{C}$.
For the first term on the right hand side of \eqref{eq:S1}, we have
\begin{align*}
\E[S_1] &=  \E\Bigl\| \sum_{i=1}^n  \varepsilon_i W_i \mathbbm{1}_{ \{ \| W_i \|_2 \leq M \} }  \Bigr\|_2\\
&\leq \Bigl( \E\Bigl\| \sum_{i=1}^n  \varepsilon_i W_i \mathbbm{1}_{ \{ \| W_i \|_2 \leq M \} }  \Bigr\|_2^2 \Bigr)^{1/2} \text{ by Jensen's inequality}\\
&=\Bigl(  \sum_{i=1}^n  \E\Bigl\|  W_i \mathbbm{1}_{ \{ \| W_i \|_2 \leq M \} }  \Bigr\|_2^2 \Bigr)^{1/2} \text{ by independence}\\
&\leq \Bigl(  \sum_{i=1}^n  \E \|  W_i \|_2^p M^{2-p} \Bigr)^{1/2} \text{ for } p \leq 2\\
&\leq \mathfrak{C}\Bigl(  \bigl( \sum_{i=1}^n  \E \|  W_i \|_2^p \bigr) \bigl( \sum_{i=1}^n \E [ \| W_i \|_2^p ] \bigr)^{ \frac{2}{p}-1} \Bigr)^{1/2} \text{ by the upper bound on } M\\
&= \mathfrak{C}\bigl( \sum_{i=1}^n  \E \|  W_i \|_2^p \bigr)^{1/p}.
\end{align*}
Therefore, using $1 \leq p \leq 2$, inequality \eqref{eq:S1-first} reduces to for a possibly different universal constant $\mathfrak{C}$
\begin{equation}\label{eq:S1}
(\E [S_{1}^{p}])^{1/p} \leq \mathfrak{C} \left(\sum_{i=1}^n \E [ \| W_i \|_2^p ] \right)^{1/p}.
\end{equation}

To estimate $\E [ S_{2}^{p}]$, we apply the original Hoffmann-Jørgensen inequality (from e.g., (6.9) in page 156 of \cite{ledoux2013probability}) to get
\begin{equation}\label{eq:S2-first}
\E [S_{2}^{p}] \leq 2 \cdot 3^{p}\left(t_{0}^{p}+\E [\max_{i=1,\dots,n} \|W_i\|^{p} ] \right),
\end{equation}
where $t_{0}$ is any number such that $\PP(S_{2}>t_{0} )\leq\left(8 \cdot 3^{p}\right)^{-1} .$ But the choice of $M$ implies that we can take $t_{0}=0$ because
$$
\PP( S_{2}>0 )=\PP( \max _{i=1,\dots,n} \| W_i \|_2^p >M ) \leq \frac{1}{8 \cdot 3^{p}}.
$$
Therefore, inequality \eqref{eq:S2-first} reduces to for a possibly different universal constant $\mathfrak{C}$,
\begin{equation}\label{eq:S2}
\E [S_{2}^{p}] \leq \mathfrak{C}\E \left[\max_{1\le i\le n} \|W_i\|^{p} \right] \leq \mathfrak{C} \left(\sum_{i=1}^n \E [ \| W_i \|_2^p ] \right)^{1/p}.
\end{equation}
Finally, combining inequalities \eqref{eq:S1} and \eqref{eq:S2} gives us our desired inequality in Proposition \ref{prop:vector-moment}.
\end{proof}

\begin{prop}\label{prop:s2}
Under assumptions \ref{assump:DGP}, \ref{assump:finite-Y} and \ref{assump:subGauss}, for every $\eta >0$, there exists a universal constant $C \in(0, \infty)$ such that
$$
\max _{a \in \mathbb{R}^{d},\|a\|_{2} \leq 1} \mathbb{E}\left[\bigl|a^{\top} \Sigma^{-1 / 2} X_{i} Y_{i}  \bigr|^{q_y /(1+\eta)}\right] \leq \left\{C K_x K_{y} \sqrt{q_y / \eta}\right\}^{q_y /(1+\eta)}.
$$
\end{prop}

\begin{proof}
Fix $a \in \mathbb{R}^{d}$ with Euclidean norm bounded by 1. H{\"o}lder's inequality yields
$$
\begin{aligned}
\left(\mathbb{E}\left[\bigl|a^{\top} \Sigma^{-1 / 2} X_{i} Y_{i}\bigr|^{q_y /(1+\eta)}\right]\right)^{\left(1+\eta\right) / q_y} & \leq\left(\mathbb{E}\left[ |Y_{i}   |^{q_y}\right]\right)^{1 / q_y}\left(\mathbb{E}\left[\bigl|a^{\top} \Sigma^{-1 / 2} X_{i}\bigr|^{q_y / \eta}\right]\right)^{\eta / q_y} \\
& \leq K_y \bigl(C K_{x} \sqrt{q_y / \eta} \bigr),
\end{aligned}
$$
where the second inequality follows from assumption \ref{assump:subGauss}.
\end{proof}

\begin{prop}\label{prop:s2-relaxed}
Under assumptions \ref{assump:DGP}, \ref{assump:finite-Y} and \ref{assump:finite-X}, for any $q_1 \le q_x$ and $q_2 \le q_y$, it holds that
$$
\max _{a \in \mathbb{R}^{d},\|a\|_{2} \leq 1} \mathbb{E}\left[\bigl|a^{\top} \Sigma^{-1 / 2} X_{i} Y_{i}  \bigr|^\frac{q_1 q_2}{q_1+q_2} \right] \leq (K_x K_{y} )^\frac{q_1 q_2}{q_1+q_2}.
$$
\end{prop}
\begin{proof}
Fix $a \in \mathbb{R}^{d}$ with Euclidean norm bounded by 1. H{\"o}lder's inequality yields
$$
\begin{aligned}
\left(\mathbb{E}\left[\bigl|a^{\top} \Sigma^{-1 / 2} X_{i} Y_{i}\bigr|^{ \frac{q_1 q_2}{q_1+q_2}}\right]\right)^{1/q_1+1/ q_2} & \leq\left(\mathbb{E}\left[ |Y_{i}   |^{q_2}\right]\right)^{1 / q_2}\left(\mathbb{E}\left[\bigl|a^{\top} \Sigma^{-1 / 2} X_{i}\bigr|^{q_1}\right]\right)^{ 1 / q_1} \\
& \leq K_x K_{y},
\end{aligned}
$$
where the second inequality follows from assumption \ref{assump:finite-X}.% that $ \Sigma^{-1 / 2} X_{i}$ is $K_{x}$-sub-Gaussian.
\end{proof}

\begin{prop}\label{prop:simple-sufficient-quantile-holder}
If the density of $Y - X^{\top}\theta$ is bounded away from zero on $[F_\theta^{-1}\bigl( 1-\alpha - r^* \bigr),F_\theta^{-1}\bigl( 1-\alpha + r^* \bigr) ]$, then assumptions~\ref{assump:general-holder-VFCP} and~\ref{assump:general-holder-EFCP} hold true with $\gamma = 1$.
\end{prop}
\begin{proof}
Assume the density of $Y - X^{\top}\theta$ is larger than some $a>0$ on $[F_\theta^{-1}\bigl( 1-\alpha - r^* \bigr),F_\theta^{-1}\bigl( 1-\alpha + r^* \bigr) ]$ for some $r^*>0$, then for any $q_1,q_2 \in [1-\alpha - r^*,1-\alpha + r^* ]$, assume WLOG that $q_2 \geq q_1$,
\begin{align*}
F^{-1}_\theta(q_2) - F^{-1}_\theta(q_1)&=  \{t_2-t_1 : \PP(|Y-X^\top \theta| \leq t_2)=q_2, \PP(|Y-X^\top \theta| \leq t_1)=q_1 \} \\
&=\{t_2-t_1 : \PP( |Y-X^\top \theta| \in (t_1,t_2])=q_2-q_1, \PP(|Y-X^\top \theta| \leq t_2)=q_2\} \\
&\le\{t_2-t_1 : 2a(t_2-t_1) \leq \PP( |Y-X^\top \theta| \in (t_1,t_2])=q_2-q_1\} \\
&\leq (q_2-q_1)/2a.
\end{align*}
\end{proof}

% {\color{red}Make a proposition when $Y|X$ is bounded away from 0, the assumption (A0) will hold true with $\gamma=1$. Proof is similar to Prop.8.}
{
\begin{prop}\label{prop:simple-sufficient-quantile-holder-conditional}
If the density of $Y|X$ is bounded away from zero on bounded sets of $\mathbb{R}$, i.e., for all $-\infty < a < b < \infty$ there exists $C_{a,b} > 0$ such that $\inf_{y\in[a, b]}\inf_x p_{Y|X}(y|x) \ge C_{a,b}$ and the covariates are bounded (i.e., $\|X\|_2 \le M < \infty$ almost surely), then assumptions~\ref{assump:general-holder-EFCP} and~\ref{assump:general-holder-VFCP} hold true for any bounded set $\Theta$ with $\gamma = 1$.
\end{prop}
\begin{proof}
By Proposition~\ref{prop:simple-sufficient-quantile-holder}, it suffices to show that the density of $Y - X^{\top}\theta$ is bounded away from zero on bounded sets. Fix any bounded set $\Theta$ satisfying $\sup_{\theta\in\Theta}\|\theta\|_2 \le K$. 
\begin{align*}
\inf_{t\in[t_1, t_2]}\inf_{x}p_{Y - X^{\top}\theta|X}(t|x) &= \inf_{t\in[t_1, t_2]}\inf_{x}p_{Y|X}(t + x^{\top}\theta|x)\\ 
&\ge \inf_{y\in[t_1 - MK, t_2 + MK]}\inf_{x}p_{Y|X}(y|x) \ge C_{t_1 - MK, t_2 + MK} > 0.
\end{align*}
This implies the density of $Y - X^{\top}\theta$ is bounded away from zero. And Proposition \ref{prop:simple-sufficient-quantile-holder} directly leads to the conclusion.
\end{proof}}

\begin{prop}\label{prop:joint-normal}
If $(Y,X)$ has a joint normal distribution with $$(Y,X) \sim \Normal \biggl(\begin{pmatrix}
0  \\
0
\end{pmatrix}, \Sigma= \begin{pmatrix}
1 & \Sigma_{yx}  \\
\Sigma_{xy} & I_d
\end{pmatrix}\biggr),$$
and that all the eigenvalues of $\Sigma$ is bounded between $[\underline{\lambda}, \overline{\lambda}]$ for some $\underline{\lambda} > 0$,  then for any $q_0 \in (0,0.5)$, assumption~\ref{assump:general-holder-EFCP} will hold true for any $q_1, q_2 \in [q_0, 1-q_0]$ with  $\gamma=1$ and some $L=L_\theta$ depending on $\theta$. Further, for any $\theta \in \mathbb{R}^d$ it holds that $L_\theta \leq c(\underline{\lambda}, \overline{\lambda}) (1+ \| \theta \|_2 )$,  where $c$ is a constant that depends only on $\underline{\lambda}$ and  $\overline{\lambda}$.
\end{prop}

\begin{proof}

Let $\Phi_\theta$ denote the CDF for $|Y-X^\top \theta|$ and $f_\theta$ denote its corresponding density function, then for any $q_1,q_2 \in [q_{\alpha}(\theta) - r^*,q_{\alpha}(\theta) + r^* ]$, assume WLOG that $q_2 \geq q_1$,
\begin{align*}
F^{-1}_\theta(q_2) - F^{-1}_\theta(q_1)&=  \{t_2-t_1 : \PP( |Y-X^\top \theta| \leq t_2)=q_2, \PP( |Y-X^\top \theta| \leq t_1)=q_1 \} \\
&=\Phi_{\theta}^{-1}(q_2)-\Phi_{\theta}^{-1}(q_1)\\
&=(q_2-q_1)/f_\theta(\Phi_\theta^{-1}(q)) \text{ for some } q \in (q_1,q_2)\\
&\le (q_2-q_1)/f_\theta (\Phi_\theta^{-1}(q_0)).
\end{align*}
Further, 
\begin{align*}
    \sigma_{\theta}^2 := \mathrm{var} (Y-X^\top \theta) &= \begin{pmatrix}
    1 &-\theta^\top
    \end{pmatrix} \mathrm{var} \{ \begin{pmatrix}
    Y \\
    X
    \end{pmatrix} \} \begin{pmatrix}
    1\\
    -\theta
    \end{pmatrix}\\
    &\leq \overline{\lambda} \begin{pmatrix}
    1 &-\theta^\top
    \end{pmatrix} \begin{pmatrix}
    1\\
    -\theta
    \end{pmatrix} = \overline{\lambda} ( 1 + \| \theta \|_2^2).
\end{align*}
Similarly, we have $\sigma_\theta^2 \geq \underline{\lambda}  ( 1 + \| \theta \|_2^2 ).$

Further, from Anderson's Lemma (Lemma 3.11.4 of \cite{van2000asymptotic})  we have that for any $t \geq 0$, 
\begin{equation}
    2 \Phi \biggl( \frac{t}{ \sqrt{\overline{\lambda} ( 1 + \| \theta \|_2^2)}} \biggr) -1 = \PP ( \bigl| \Normal \bigl( 0, \overline{\lambda} ( 1 + \| \theta \|_2^2)  \bigr) \bigr| \le t) \leq \Phi_{\theta} (t).
\end{equation}
Taking $t=\sqrt{\overline{\lambda} ( 1 + \| \theta \|_2^2) } \bigl\{ \Phi^{-1} \bigl( (1+q_0)/2  \bigr) \bigr\}=:\overline{t}$ yields $q_0 \leq \Phi_\theta (\overline{t})$ and therefore, $\Phi_\theta^{-1} (q_0) \leq \overline{t}$.

Finally,
\begin{align}
L_\theta &=1/f_\theta ( \Phi_\theta^{-1} (q_0) )\\
 &= \sqrt{\pi/2} \sigma_\theta  \exp \biggl(  \frac{ \Phi_\theta ^{-1} (q_0)^2 }{2 \sigma_\theta^2 }  \biggr)  \\
&\leq \sqrt{\pi \overline{\lambda} ( 1 + \| \theta \|_2^2 ) /2 }  \exp \biggl\{ \frac{ \Phi_\theta ^{-1} (q_0) ^2 }{2 \underline{\lambda} ( 1 + \| \theta \|_2^2) }   \biggr\}\\
&\leq \sqrt{\pi \overline{\lambda} ( 1 + \| \theta \|_2^2 ) / 2 }  \exp \biggl\{ \frac{ \overline{t}^2 }{2 \underline{\lambda} ( 1 + \| \theta \|_2^2) }   \biggr\}\\
&= \sqrt{\pi \overline{\lambda} ( 1 + \| \theta \|_2^2 ) /2 }  \exp \biggl\{ \frac{ \overline{\lambda}  \bigl\{ \Phi_\theta ^{-1} \bigl( (1+q_0)/2 ) \bigr) \bigr\}^2 }{2 \underline{\lambda}  }   \biggr\}\\
&=: c(\underline{\lambda}, \overline{\lambda}) \sqrt{1+ \| \theta \|_2^2}\\
&\leq c(\underline{\lambda}, \overline{\lambda}) ( 1+ \| \theta \|_2 ) , \forall \theta \in \mathbb{R}^d.
\end{align}

\end{proof}

\end{document}